\documentclass[10pt, a4paper, oneside, reqno]{article}
\usepackage[total={7in,    10.5in}]{geometry}
\date{}
\usepackage{authblk} 
\usepackage[numbers,compress]{natbib}

\usepackage{amsmath,amssymb,mathrsfs,amsfonts} 
\usepackage{bm,bbm,dsfont} 

\usepackage{color,graphicx,graphics, epstopdf} 

\usepackage{algorithm,algorithmicx,multirow, titlecaps}  
\usepackage[english]{babel}
\usepackage[noend]{algpseudocode} 
\usepackage{threeparttable}
\usepackage{tikzsymbols}
\usepackage{multirow, titlecaps} 
\usepackage{tabularx}
\usepackage{booktabs}
\usepackage{multirow}
\usepackage{courier}
\usepackage{float}
\usepackage{multirow, titlecaps} 
\usepackage{tabularx}
\usepackage{booktabs}
\usepackage{multirow}
\usepackage{courier}
\usepackage{float}
\usepackage{makecell}
\usepackage{diagbox}
\usepackage{physics}
\usepackage{tikz}
\usepackage{circuitikz}

\allowdisplaybreaks
\usepackage{verbatim}

\DeclareMathAlphabet{\mathbbold}{U}{bbold}{m}{n}

\usepackage{xspace}

\usepackage{url}
 \iftrue    
    \usepackage[colorlinks=true]{hyperref}
	\usepackage{xcolor}
	\hypersetup{
		colorlinks,
		linkcolor={blue!90!black},
		citecolor={red!90!black},
		urlcolor={blue!20!black}
	}
	\allowdisplaybreaks
\fi

\newcommand{\mtiny}[1]{{\scalebox{.81}{#1}}}

\newcommand{\smtiny}[1]{{\scalebox{.5}{#1}}}
\newcommand{\stiny}[1]{{\scalebox{.38}{#1}}}

\newcommand*{\argminOp}{\operatornamewithlimits{argmin}\limits}
\newcommand*{\argmaxOp}{\operatornamewithlimits{argmax}\limits}

\renewcommand{\tr}{{\text{\mtiny{$\mathsf{T}$}}}}              
\iftrue

\else

\fi
\newcommand{\vc}[1]{{ \mathrm{#1} }}  
\newcommand{\mx}[1]{{ \mathrm{#1} }}  
\newcommand{\drm}{\mathrm{d}}           
\newcommand{\linspan}{\mathrm{span}}    
\newcommand{\nth}{{\text{\tiny{th}}}}   
\renewcommand{\trace}{\mathrm{tr}}        
\newcommand{\logdet}{\mathrm{logdet}}       
\renewcommand{\vec}{\mathrm{vec}}           

\newcommand{\Dcal}{{\mathcal{D}}}

\newcommand{\Ncal}{{\mathcal{N}}}

\newcommand{\Scal}{{\mathcal{S}}}

\newcommand{\Nbb}{{\mathbb{N}}}

\newcommand{\Rbb}{{\mathbb{R}}}
\newcommand{\Sbb}{{\mathbb{S}}}

\newcommand{\Vbb}{{\mathbb{V}}}
\newcommand{\Xbb}{{\mathbb{X}}}

\newcommand{\Zbb}{{\mathbb{Z}}}

\newcommand{\Exp}{\mathbb{E}}
\iffalse 
    \input{mystyleproof}
    \theoremstyle{plain}
    \newtheorem{theorem}{Theorem}
    \newtheorem{definition}{Definition}
    \newtheorem{assumption}{Assumption}
    \newtheorem{proposition}[theorem]{Proposition}
    \newtheorem{corollary}[theorem]{Corollary}
    \newtheorem{lemma}[theorem]{Lemma}
    \newtheorem{remark}{Remark}
    \newtheorem{example}{Example}
    \newtheorem{problem}{Problem}

\else

    \usepackage{amsthm}
    \renewcommand{\qedsymbol}{$\blacksquare$}

    \newtheorem{theorem}{Theorem}
    \newtheorem*{theorem*}{Theorem}
    
    \newtheorem*{definition*}{Definition}
    \newtheorem{assumption}{Assumption}
    \newtheorem{proposition}[theorem]{Proposition}
    \newtheorem{corollary}[theorem]{Corollary}
    \newtheorem{lemma}[theorem]{Lemma}
    \newtheorem{remark}{Remark}
    
    \newtheorem*{example*}{Example}

    \newtheorem*{problem*}{Problem}
\fi

\newcommand{\mxA}{\mx{A}} 
\newcommand{\mxB}{\mx{B}} 
\newcommand{\mxC}{\mx{C}} 
\newcommand{\mxD}{\mx{D}} 
\newcommand{\mxE}{\mx{E}} 
\newcommand{\mxK}{\mx{K}} 
\newcommand{\mxL}{\mx{L}} 
\newcommand{\mxM}{\mx{M}} 
\newcommand{\mxN}{\mx{N}} 
\newcommand{\mxQ}{\mx{Q}} 
\newcommand{\mxP}{\mx{P}} 
\newcommand{\mxR}{\mx{R}} 
\newcommand{\mxS}{\mx{S}} 
\newcommand{\mxX}{\mx{X}} 
\newcommand{\mxY}{\mx{Y}} 
\newcommand{\mxG}{\mx{G}} 
\newcommand{\mxH}{\mx{H}} 
\newcommand{\mxU}{\mx{U}} 
\newcommand{\mxV}{\mx{V}} 
\newcommand{\mxF}{\mx{F}} 

\newcommand{\vca}{\vc{a}} 
\newcommand{\vce}{\vc{e}} 
\newcommand{\vcu}{\vc{u}} 
\newcommand{\vcv}{\vc{v}} 
\newcommand{\vcw}{\vc{w}} 
\newcommand{\vcx}{\vc{x}} 
\newcommand{\vcy}{\vc{y}} 
\newcommand{\vcz}{\vc{z}} 
\newcommand{\vctheta}{\vc{\theta}} 
\newcommand{\vcmu}{\vc{\mu}} 
\newcommand{\vcq}{\vc{q}} 

\newcommand{\Nu}{n_\vc{u}} 
\newcommand{\Nx}{n_\vc{x}} 
\newcommand{\Ny}{n_\vc{y}} 

\newcommand{\nD}{{n_{\smtiny{$\Dcal$}}}} 
\newcommand{\nsD}{{n_{\stiny{$\Dcal$}}}} 
\newcommand{\nsmD}{{n_{\smtiny{$\Dcal$}}}} 

\newcommand{\mux}{\mu_{\vcx_0}}     
\newcommand{\mxSx}{\mx{S}_{\vcx_0}}   
\newcommand{\mxSw}{\mx{S}_{\vcw}}     
\newcommand{\mxSv}{\mx{S}_{\vcv}}     

\newcommand{\vcthetaML}{\theta^{\text{\mtiny{$\mathrm{(ML)}$}}}}     

\newcommand{\Sbfxtheta}{\mxS^{{(\mathbf{x})}}\!(\theta)}
\newcommand{\Sbfytheta}{\mxS^{{(\mathbf{y})}}\!(\theta)}
\newcommand{\mubfxtheta}{\mu^{{(\mathbf{x})}}\!(\theta)}
\newcommand{\mubfythetax}{\mu^{{(\mathbf{y})}}\!(\theta, \mathbf{x})}
\newcommand{\Xitheta}{\Xi(\theta)}

\newcommand{\Sqtheta}{\mxS_\vcq(\theta)}

\newcommand{\extendedR}{\overline{\Rbb}}

\title{\LARGE \textbf{Probabilistic Formulations for System Identification of Linear Dynamics with Bilinear Observation Models}}
\author[$\dagger$]{Diyou Liu}
\author[$\star$]{Mohammad Khosravi}
\affil[$\dagger\star$]{{\small Delft Center for Systems and  Control, Delft University of Technology, Delft, Netherlands
\authorcr
{\texttt{
\{d.liu-9,mohammad.khosravi\}@tudelft.nl}}}}  
\begin{document}
\maketitle
\begin{abstract}
In this paper, we address the identification problem for the systems characterized by linear time-invariant dynamics with bilinear observation models. More precisely, we consider a suitable parametric description of the system and formulate the identification problem as the estimation of the parameters defining the mathematical model of the system using the observed input-output data. To this end, we propose two probabilistic frameworks. The first framework employs the Maximum Likelihood (ML) approach, which accurately finds the optimal parameter estimates by maximizing a likelihood function. Subsequently, we develop a tractable first-order method to solve the optimization problem corresponding to the proposed ML approach. Additionally, to further improve tractability and computational efficiency of the estimation of the parameters, we introduce an alternative framework based on the Expectation--Maximization (EM) approach, which estimates the parameters using an appropriately designed cost function. We show that the EM cost function is invex, which ensures the existence and uniqueness of the optimal solution. Furthermore, we derive the closed-form solution for the optimal parameters and also prove the recursive feasibility of the EM procedure. 
Through extensive numerical experiments, the practical implementation of the proposed approaches is demonstrated, and their estimation efficacy is verified and compared, highlighting the effectiveness of the methods to accurately estimate the system parameters and their potential for real-world applications in scenarios involving bilinear observation structures.
\end{abstract}
       

\section{Introduction} \label{sec:intro}
%
%
%
%
System identification, as originally introduced in \cite{zadeh1956identification}, is an active area of research \cite{ljung2010perspectives} that focuses on the theory and methods for data-driven modeling of dynamical systems \cite{luenberger1979dynamic}. 
This field has gained substantial attention \cite{LjungBooK2,chiuso2019system} due to the ubiquitous nature of dynamical systems in various science and technology domains such as physics, economics, biology, and engineering, and, considering the particular importance of accurate models for tasks such as prediction, monitoring, and control \cite{khansari2014learning,tang2021expectation,ferrari2012system,franco2023parameter,shastry2023system}. 
With vast areas of application relying on precise models that accurately represent reality, numerous methodologies have been developed for different categories of systems, ranging from linear to nonlinear dynamics, discrete-time to continuous-time systems, deterministic to stochastic models, as well as lumped-parameter and distributed-parameter systems, and many other classes of systems \cite{schoukens2019nonlinear,sattar2022non,khosravi2021Koopman,umenberger2018maximum,khansari2017learning,liu2024learning,khosravi2021SSG}.
The key role of linear dynamics in practical applications has motivated extensive research into the identification problem for different classes of linear systems, covering various structures, types, and features, such as 
stability.
\cite{lacy2003subspace,pillonetto2014kernel}, 
positivity 
\cite{khosravi2021POS,zheng2021bayesian,khosravi2020regularized},
being compartmental 
\cite{benvenuti2002model}, 
positive-realness 
\cite{goethals2003identification}, 
passivity 
\cite{rodrigues2021novel},
low internal complexity 
\cite{smith2014frequency,khosravi2020low,pillonetto2016AtomicNuclearKernel}, 
specific frequency-domain attributes 
\cite{khosravi2021FDI,marconato2016filter}, 
and many other significant properties \cite{risuleo2019bayesian,everitt2018empirical}.
%
Consequently, linear system identification techniques have undergone significant advancements, often demonstrating remarkable estimation performance and computational efficiency, making them well-suited for systems exhibiting predominantly linear behavior \cite{aastrom1971system,orlov2003adaptive}.
%
However, since many real-world systems exhibit nonlinear, stochastic, or generally complex behaviors, the resulting complexity presents significant challenges in accurately capturing their properties in the model.
To address this issue, various advanced tools and techniques from statistics and optimization theory are employed in linear and nonlinear system identification, leading to the development of parametric and nonparametric approaches. Nonparametric estimation techniques \cite{khosravi2022Lut,zorzi2022nonparametric,khosravi2021ROA,khosravi2021grad}, which impose fewer structural assumptions and rely more significantly on the data, are particularly well-suited when the system structure is less known. In contrast, parametric approaches are especially useful when a specific well-defined and appropriate structure for the system is given, which is often the case in many real-world scenarios \cite{astrom1979maximum,bremaud2012introduction,little2019statistical}.

Parametric system identification is an extensively used framework that assumes a sufficiently comprehensive class of models with a predefined specific structure for the mathematical representation of the system and then estimates the parameters characterizing these models using measurement data. Accordingly, it can be applied to a broad range of systems of various types and natures, provided that a suitable parametric mathematical characterization is available. The subsequent primary objective is to fine-tune the parameters of the model so that the resulting system most accurately describes the observed data.
To this end, different parameter estimation methods are employed depending on the nature of the system and the objectives of the identification process. 
For example, the prediction error method (PEM) \cite{astrom1979maximum}  is widely used in the identification of linear time-invariant systems, which iteratively adjusts the parameters to minimize the fitting loss, evaluated based on the difference between the output values predicted by the resulting model and the actual output from the real system.
Furthermore, to improve the accuracy and robustness of parameter estimation in noisy systems, probabilistic frameworks are widely employed.
Unlike deterministic methods, these approaches incorporate uncertainty into the estimation process, making them well suited for handling noisy measurement data and unmodeled system dynamics. For example, the Bayesian approach \cite{bernardo2009bayesian} treats parameters as random variables with probability distributions of specific forms. 
The Bayesian approaches exhibit remarkable performance in providing accurate probability distributions for parameters of interest, particularly when dealing with noisy or limited data, thus making it a universally applicable choice of technique in numerous system identification applications. 
The maximum likelihood (ML) estimation approach \cite{rossi2018mathematical} is one of the most prominent Bayesian methods used in system identification, aiming to determine the parameter values that maximize the likelihood function, which indicates the probability of observing the given measurement data under specified parameter settings.
In practical terms, ML estimation approach identifies the parameter values that maximize the likelihood of the observed data under the model. As the dataset size increases, the ML estimation results often converge to the true parameter values.
Aside from the ML approach, Maximum a Posteriori (MAP) estimation \cite{gauvain1994maximum} is another widely used method, particularly when prior knowledge is available for incorporation into the estimation process. MAP estimation involves selecting a suitable prior distribution for the parameters, deriving the posterior distribution by combining the likelihood function with the prior, and identifying the parameter values that maximize the obtained posterior distribution. It is essential to note that MAP can be applied only when, along with the measurement data, additional information about the parameters describing the system is available.

In addition to linear system identification, there is a growing interest in methods for identifying nonlinear systems \cite{sattar2022non,bouvrie2017kernel,khosravi2021Koopman,umenberger2018maximum}. The identification of nonlinear systems often requires advanced techniques that can accurately model and predict their complex dynamics. One classic approach is to employ schemes based on linearization \cite{chen2004identification}, e.g.,  \cite{sharabiany2024nonlinear} presents a method to identify nonlinear systems through local linear approximations along the trajectory. Additionally, probabilistic frameworks are also commonly used in nonlinear systems identification. In \cite{wills2013identification}, the authors proposed an ML-based algorithm for Hammerstein--Wiener models, which is a serial combination of systems, 
which consist of a linear dynamic system cascaded with two  static nonlinear blocks.
Moreover, in \cite{liu2021identification}, a Bayesian approach-based method is introduced to identify Wiener--Hammerstein models, which involve a static nonlinearity between two linear dynamic systems. 
Among general nonlinear systems, bilinear systems are of particular interest as the bridging case between linear and nonlinear systems, and also due to their technical tractability and relevance in various fields \cite{pardalos2010optimization}.
Bilinear systems combine linear dynamics with input-dependent nonlinearities, which can be exploited in designing effective identification methods with computational tractability. For example, in \cite{verdult2001identification,verdult2005kernel}, subspace techniques are used to identify bilinear state space systems. Furthermore, in \cite{liu2020moving} and \cite{liu2022expectation}, a bilinear system in canonical form of observability is considered. In \cite{liu2020moving}, the authors proposed an approach using the Kalman filter to estimate the state of the system and a gradient-based iterative algorithm to identify the parameters of the system. In \cite{liu2022expectation}, the Rauch–Tung–Striebel smoother (RTS) is used to estimate the state variables and the expectation-maximization (EM) algorithm to identify parameters. These papers consider systems with bilinear dynamics and linear observation models. Inspired by the Wiener-Hammerstein models and Hammerstein–Wiener models, in this current work, we extend the approach to consider bilinear systems with linear dynamics and bilinear observation models. By developing identification methods for systems with bilinear observations, aiming to provide more flexible and realistic models capable of capturing interactions that standard linear or bilinear dynamics alone cannot fully describe. 

In this paper, we focus on identifying systems with linear dynamics and bilinear observation models, where the system matrices, initial state, and noise distributions, including their means and covariances, are treated as unknown parameters to be estimated. 
%
To address this problem, we first propose a scheme using the ML approach. 
%
More precisely, we formulate a likelihood function to identify the optimal parameters that maximize the likelihood of the observed data. 
%
Subsequently, we develop a practical first-order method to efficiently solve the optimization problem associated with the proposed ML approach. 
%
Additionally, to further improve the tractability and computational efficiency, we extends our previous work \cite{liu2024system} and propose an alternative identification scheme based on the RTS smoother \cite{rauch1965maximum} and the EM approach \cite{dempster1977maximum}.
The developed scheme consists of two iterative steps. 
%
In the first step, state estimates are obtained using the estimated parameters in the current iteration. 
%
To this end, one may employ the Kalman filter \cite{kalman1960new}, a well-established method for state estimation that recursively updates state estimates by combining information from previous and current measurements. In the literature, various extensions of the Kalman filter have been proposed to improve estimation accuracy \cite{julier2004unscented,julier1995new}. 
%
In this work, we utilize the RTS smoother, further refining the state estimates by incorporating information from future measurements alongside past and current data.
%
Notably, the RTS smoothing procedure produces more accurate state estimates compared to the Kalman filter estimation results.
%
In the second step, we use the estimated states to define a suitable cost function based on the expected log-likelihood of the parameters and subsequently minimize it to update the parameter estimates.
We show the invexity of the cost function, guaranteeing the existence and uniqueness of the optimal solution. Additionally, we derive a closed-form expression for the optimal parameters and prove the recursive feasibility property for the EM procedure. 
The main contributions of this work are outlined below. 
\begin{itemize}
    \item 
    We develop a system identification framework based on the maximum likelihood approach and formulate accordingly an optimization problem to identify systems with linear dynamics and bilinear observation models.
    \item  
    To facilitate the practical implementation of the introduced maximum likelihood-based approach, we propose first-order methods to solve the associated optimization problem and compute the necessary derivatives.
    \item 
    To further improve tractability and computational efficiency of the estimation of the parameters describing the system, an alternative framework combining the Rauch–Tung–Striebel smoother and the Expectation--Maximization approach is proposed to address the identification problem, with estimated parameters iteratively updated by solving an optimization problem.
    \item 
    We show the invexity of the EM cost function and prove the existence and uniqueness of the optimal solution for the EM optimization problem. Furthermore, we derive the closed-form solution for the optimal parameters and also verify the recursive feasibility of the EM procedure.    
    \item 
    To facilitate the practical application of the proposed methods and for ease of numerical implementation, we have included detailed algorithms outlining all the steps.
    \item
    We present the results of extensive numerical experiments conducted to demonstrate and evaluate the performance of the proposed schemes, followed by a rigorous discussion of the observed features and phenomena, offering valuable insights into their underlying behavior.
\end{itemize}

The remainder of this paper is organized as follows. Section \ref{sec:notation} lists the main notations used throughout the paper. The identification problem is discussed in detail in Section \ref{sec:pf}. A scheme employing the ML approach is proposed in Section \ref{sec:ML}. In Section \ref{sec:An Expectation-Maximization Approach}, we develop an alternative scheme based on the EM approach. Section \ref{sec:numerical experiments} presents numerical examples to evaluate the performance of the proposed approaches. Finally, Section~\ref{sec:conclusion} concludes the paper.
\section{Notation} \label{sec:notation}
In this paper, the set of integers, the set of positive integers, the set of real numbers, the extended set of real numbers, and the set of $n$ by $m$ matrices with real value are denoted by $\Zbb$, $\Zbb_+$, $\Rbb$, $\extendedR$, and $\Rbb^{n\times m}$, respectively. 
%
The sets of $n$ by $n$ symmetric matrices, positive semi-definite matrices and positive definite matrices are denoted as $\Sbb^{n\times n}$, $\Sbb_{+}^{n\times n}$ and $\Sbb_{++}^{n\times n}$, respectively. 
%
For any symmetric matrix $\mxA \in \Sbb^{n \times n}$, the notation $\mxA \succ 0$ indicates that $\mxA$ is positive definite.
For any matrix $\mxA \in \Rbb^{m\times n}$, the entry at the $i$-th row and $j$-th column is denoted by $[\mxA]_{i,j}$, for $i=1,\ldots,m$, and $j=1,\ldots,n$. 
%
The $i$-th column of $\mxA$ is represented by $[\mxA]_i$. Furthermore, $\vec(\mxA)$ refers to the column vector in $\Rbb^{mn}$ obtained from stacking the columns of the matrix $\mxA$
%
Additionally, the Kronecker product is denoted by $\otimes$. 
For any vector $\vcx \in \Rbb^n$, the Euclidean norm of $\vcx$ is denoted by $\|\vcx\|$. 
%
Finally,  $p(A|B)$ represents the conditional probability of $A$ given $B$.
%

\section{Identification of Linear Dynamics with Bilinear Observation Models} \label{sec:pf}
Consider an \emph{unknown} time-invariant random dynamical system $\Scal$ with linear dynamics and a bilinear observation model.  
More precisely, let the process model describing the dynamics of $\Scal$ be as
\begin{equation}\label{eqn:dynamics1}
\vcx_{k+1} = \mxA\vcx_k+\mxB\vcu_k+\vcw_k,    \qquad \forall k\in\Zbb_+,
\end{equation}
where $\vcx_{k}\in\Rbb^{\Nx}$, $\vcu_k\in\Rbb^{\Nu}$, and $\vcw_k\in\Rbb^{\Nx}$ are respectively the vectors of state variables, input, and process noise, at time instant $k\in\Zbb_+$, and,
$\mxA\in\Rbb^{\Nx\times\Nx}$ and $\mxB\in\Rbb^{\Nx\times\Nu}$ are \emph{unknown} matrices characterizing the dynamics of system.
Also, let the observation model of the system have a bilinear form as 
\begin{equation}\label{eqn:dynamics2}
\vcy_{k+1} = \big(\mxC_0+\sum_{i=1}^{\Nu}\mxC_iu_{k,i}\big)\vcx_k+\mxD\vcu_k+\vcv_k,       \qquad \forall k\in\Zbb_+,
\end{equation}
where $u_{k, i}$ denotes the $i^\nth$ entry of $\vcu_k$, for $i=1,\ldots,\Nu$ and $k\in\Zbb_+$, $\vcy_{k}\in\Rbb^{\Ny}$ and $\vcv_k\in\Rbb^{\Ny}$  are respectively the vectors of output observations and measurement noise, and, $\mxC_0,\mxC_1,\ldots,\mxC_{\Nu}\in\Rbb^{\Ny\times\Nx}$ and $\mxD\in\Rbb^{\Ny\times\Nu}$ are \emph{unknown} matrices describing the observation model of the system. Suppose that the initial state $\vcx_0$, the process noise $\big(\vcw_k\big)_{k\in\Zbb_+}$, and the measurement noise $\big(\vcv_k\big)_{k\in\Zbb_+}$ are mutually independent Gaussian random variables such as 
\begin{align}
    &\vcx_0\sim\Ncal(\mux,\mxSx),
    \label{eqn:x_dist_Gaussion}
    \\
    &\vcw_k\sim\Ncal(0,\mxSw), \qquad \forall k\in\Zbb_+,
    \label{eqn:w_dist_Gaussion}
    \\
    &\vcv_k\sim\Ncal(0,\mxSv), \qquad \forall k\in\Zbb_+,
    \label{eqn:v_dist_Gaussion}
\end{align}
where $\mux\in\Rbb^{\Nx}$ is  an \emph{unknown} vector and $\mxSx\in\Rbb^{\Nx\times\Nx}$ is an \emph{unknown} positive definite matrix  respectively denoting the mean and covariance of $\vcx_0$, and, $\mxSw\in\Rbb^{\Nx\times\Nx}$  and  $\mxSv\in\Rbb^{\Ny\times\Ny}$ are \emph{unknown} positive definite matrices representing respectively the covariance of vector $\vcw_k$ and  the covariance of vector $\vcv_k$, for any $k\in\Zbb_+$.

Assume a set of $\nD\in\Nbb$ input-output pairs of measurement data is given as
\begin{equation}
    \Dcal:=\big\{(\vcu_k,\vcy_k)\,|\, k=0,\ldots,\nD-1\big\}
\end{equation}
Accordingly, we introduce the main problem discussed in this paper as identifying system $\Scal$ through estimating the unknown vector and matrices mentioned above using the set of data $\Dcal$. 
\begin{problem*}[\textbf{Identification Problem for Linear Dynamics with Bilinear Observation Models}] \label{pr:Bilinear Observation}
Given the measurement set of data $\Dcal$, 
estimate $\mxA$, $\mxB$, $\mxC_0$, $\mxC_1$, $\ldots$, $\mxC_{\Nu}$, $\mxD$, $\mux$, $\mxSx$, $\mxSw$, and $\mxSv$.
\end{problem*}
To address the estimation problem introduced, we utilize probabilistic approaches, specifically the maximum likelihood (ML) method \cite{lehmann2006theory} and the expectation--maximization (EM) algorithm \cite{theodoridis2006pattern}. In the following sections, we derive tractable procedures for implementing these methods. 
Detailed discussions can be found in Section~\ref{sec:ML} and Section~\ref{sec:An Expectation-Maximization Approach}.
\section{Maximum Likelihood Approach}\label{sec:ML}

To address the identification problem introduced in Section~\ref{sec:pf}, a natural choice is the maximum likelihood (ML) estimation method \cite{lehmann2006theory}, which identifies the parameters that maximize the likelihood of observed measurements under the assumed model. 
In this section, we outline how to apply the ML method to our identification problem.
More precisely, we derive the likelihood function based on the observed data and the parameterized model, and then, determine the values of the optimal parameters by maximizing the likelihood function. 

To derive the likelihood function and formulate the maximum likelihood scheme, we first define multi-tuple of parameters, denoted by  $\vctheta$, as
\begin{equation}
    {\vctheta} 
    := 
    \,\,(
    \mxA, 
    \mxB, 
    \mxC_0,\mxC_1,...,\mxC_{\Nu}, 
    \mxD, 
    \vcmu_{\vcx_0},
    \mxS_{\vcx_0},
    \mxS_\vcw,\mxS_\vcv),
    \end{equation}
which belongs to the vector space $\Vbb$ defined as
\begin{equation}
    \Vbb:=
    \Rbb^{\Nx\times\Nx} 
    \times
    \Rbb^{\Nx\times\Nu} 
    \times
    \Rbb^{\Ny\times\Nx} 
    \times
    \cdots
    \times 
    \Rbb^{\Ny\times\Nx} 
    \times
    \Rbb^{\Ny\times\Nu} 
    \times
    \Rbb^{\Nx} 
    \times
    \Sbb^{\Nx\times\Nx} 
    \times
    \Sbb^{\Nx\times\Nx} 
    \times
    \Sbb^{\Ny\times\Ny}. 
\end{equation}
Also, let $\Theta \subseteq \Vbb$ be the feasible set for $\vctheta$, specified based on the structure of the its entries and the corresponding constraints, i.e.,
\begin{equation}
    \Theta :=
    \Rbb^{\Nx\times\Nx} 
    \times
    \Rbb^{\Nx\times\Nu} 
    \times
    \Rbb^{\Ny\times\Nx} 
    \times
    \cdots
    \times 
    \Rbb^{\Ny\times\Nx} 
    \times
    \Rbb^{\Ny\times\Nu} 
    \times
    \Rbb^{\Nx} 
    \times
    \Sbb_{++}^{\Nx\times\Nx} 
    \times
    \Sbb_{++}^{\Nx\times\Nx} 
    \times
    \Sbb_{++}^{\Ny\times\Ny}. 
\end{equation}
Let $\mathbf{u}$ and $\mathbf{y}$ denote the vectors of input and output measurements, respectively, defined as $\mathbf{y} := [\vcy_0^\tr,\vcy_1^\tr,..,\vcy_{{\nsD-1}}^\tr]^\tr\in\Rbb^{\Ny\nsD}$ and $\mathbf{u} := [\vcu_0^\tr,\vcu_1^\tr,..,\vcu_{\nsD-1}^\tr]^\tr\in\Rbb^{\Nu\nsD}$. The ML estimation approach identifies $\vctheta$ by determining the 
parameters $\vcthetaML$ that maximizes the probability of observing the given measurements $\mathbf{u}$ and $\mathbf{y}$. 
More precisely, we have 
\begin{equation} \label{eqn:ML_generic}
    \vcthetaML  := \argmaxOp_{\vctheta \in \Theta}\, p(\mathbf{y}|\vctheta, \mathbf{u}).
\end{equation} 
To solve the optimization problem \eqref{eqn:ML_generic}, we need to simplify the corresponding objective function
and express it in a tractable closed form characterized directly in terms of the optimization variable $\vctheta$.

\subsection{Maximum Likelihood Optimization Problem}\label{sec:ML formulation}
Let $\mathbf{x}$ denote the vector of state trajectory, defined as $\mathbf{x} := [\vcx_0^\tr,\vcx_1^\tr,..,\vcx_{{\nsD}-1}^\tr]^\tr$. 
According to \eqref{eqn:dynamics1}--\eqref{eqn:v_dist_Gaussion}, one can easily see that $\mathbf{x}\big|\vctheta,\mathbf{u}$ and $\mathbf{y}\big|\mathbf{x},\vctheta,\mathbf{u}$ follow Gaussian distributions, i.e., we have
\begin{equation}
  \mathbf{x}\big|\vctheta,\mathbf{u} 
  \ \sim \ 
  \mathcal{N}
  \big(
  \mu^{(\mathbf{x})}(\vctheta), \mxS^{(\mathbf{x})}(\vctheta)
  \big),
\end{equation}
and 
\begin{equation}
  \mathbf{y}\big|\mathbf{x},\vctheta,\mathbf{u} 
  \ \sim \
  \mathcal{N}
  \big(
  \mu^{(\mathbf{y})}(\vctheta,\mathbf{x}), \mxS^{(\mathbf{y})}(\vctheta,\mathbf{x})
  \big),  
\end{equation}
where 
$\mu^{(\mathbf{x})}(\vctheta)$, 
$\mu^{(\mathbf{y})}(\vctheta,\mathbf{x})$, 
$\mxS^{(\mathbf{x})}(\vctheta)$, and
$\mxS^{(\mathbf{y})}(\vctheta,\mathbf{x})$
are the corresponding mean vectors and covariance matrices.
For ease of notation, we have omitted the dependence on the input measurement vector $\mathbf{u}$, assuming that it is implicitly known.
Using the law of total probability, we know that 
\begin{equation}\label{eq:law_of_total_probability}
\begin{split}
p(\mathbf{y}|\vctheta, \mathbf{u}) 
= 
\int_{\Xbb}
\,
p(\mathbf{y}|\mathbf{x}, \vctheta, \mathbf{u})p(\mathbf{x}|\vctheta, \mathbf{u}) \drm \mathbf{x}, 
\end{split}
\end{equation}
where 
$\Xbb$ 
is the integration domain, defined as $\Rbb^{\Nx\nsD}$.
Thus, the maximization problem \eqref{eqn:ML_generic} can be written as 
\begin{equation}\label{eq:likelihood0}
\begin{split}
\vcthetaML 
= 
\argmaxOp_{\vctheta\in\Theta} 
\int_{\Xbb}\, p(\mathbf{y}|\mathbf{x}, \vctheta, \mathbf{u})p(\mathbf{x}|\vctheta, \mathbf{u}) \drm \mathbf{x}. \\
\end{split}
\end{equation}
To formulate a tractable optimization problem for the ML estimation of parameters based on \eqref{eq:likelihood0}, we need to evaluate the integral in \eqref{eq:law_of_total_probability}, which requires
determining the values of $\mu^{(\mathbf{x})}(\vctheta)$, $\mu^{(\mathbf{y})}(\vctheta,\mathbf{x})$, $\mxS^{(\mathbf{x})}(\vctheta)$, and $\mxS^{(\mathbf{y})}(\vctheta,\mathbf{x})$.
%
%
According to \eqref{eqn:dynamics1}, we have
\begin{equation}\label{eqn:x_t}
    \vcx_t = \mxA^t\vcx_0 + \sum_{k=0}^{t-1}\mxA^{t-1-k}\mxB \vcu_k + \sum_{k=0}^{t-1} \mxA^{t-1-k}\vcw_k,
\end{equation}
for each $t = 1,2,\cdots, {\nsmD-1}$.
One can easily see the middle term in \eqref{eqn:x_t} is determined when $\vctheta$ and $\mathbf{u}$ are known.
Additionally, since we assume that $\vcx_0$ and $\vcw_k$ are independent Gaussian random vectors, the first and last terms also follow independent Gaussian distributions given $\vctheta$ and $\mathbf{u}$. Consequently, $\vcx_t$ also follows a Gaussian distribution. Thus, regarding $p(\mathbf{x}|\vctheta, \mathbf{u})$, 
we can obtain the joint Gaussian distribution $\mathcal{N}\left(\mu^{(\mathbf{x})}(\vctheta), \mxS^{(\mathbf{x})}(\vctheta)\right)$, where
\begin{equation}\label{eq:mean x}
    \begin{split}
    \mu^{(\mathbf{x})}(\vctheta) := 
    \Exp\big[\mathbf{x}\big| \vctheta, \mathbf{u}\big] = 
    \begin{bmatrix}
        \mu_{\vcx_0} \\
        \mxA\mu_{\vcx_0} + \mxB\vcu_0\\
        \vdots\\
        \mxA^t\mu_{\vcx_0} + \sum_{k=0}^{t-1}\mxA^{t-1-k}\mxB \vcu_k\\
        \vdots\\
        \mxA^{\nsD-1}\mu_{\vcx_0} + \sum_{k=0}^{\nsD-2}\mxA^{\nsD-2-k}\mxB \vcu_k\\
    \end{bmatrix},
    \end{split}
\end{equation}
and
\begin{equation}\label{eq:covariance x}
    \begin{split}
    \mxS^{(\mathbf{x})}{(\vctheta)} := 
    \Exp
    \big[
    \big(\mathbf{x}- \mu^{(\mathbf{x})}(\vctheta)\big)
    \big(\mathbf{x}- \mu^{(\mathbf{x})}(\vctheta)\big)^\tr    
    \big| \vctheta, \mathbf{u}
    \big] = 
    \begin{bmatrix}
        \mxS^{(\mathbf{x})}_{0,0}(\vctheta) & \, \mxS^{(\mathbf{x})}_{0,1}(\vctheta) & \cdots & \mxS^{(\mathbf{x})}_{0,\nsD-1}(\vctheta)\\[10pt]
        \mxS^{(\mathbf{x})}_{1,0}(\vctheta) & \, \mxS^{(\mathbf{x})}_{1,1}(\vctheta) & \cdots & \mxS^{(\mathbf{x})}_{1,\nsD-1}(\vctheta)\\
        \vdots & \, \vdots & \ddots & \vdots\\
        \mxS^{(\mathbf{x})}_{\nsD-1,0}(\vctheta) & \, \mxS^{(\mathbf{x})}_{\nsD-1,1}(\vctheta) & \cdots & \mxS^{(\mathbf{x})}_{\nsD-1,\nsD-1}(\vctheta)\\
    \end{bmatrix},
    \end{split}
\end{equation}
given that $\mxS^{(\mathbf{x})}_{t,s}(\vctheta)$ is defined as
\begin{equation}
    \begin{split}
        \mxS_{t,s}^{(\mathbf{x})}(\vctheta) 
        :=&\ 
        \mathbb{E}\big[(\vcx_t - \mu_t^{(\mathbf{x})})(\vcx_s - \mu_s^{(\mathbf{x})})^\tr\big| \vctheta, \mathbf{u}
        \big]
        =
        \mxA^t\mxS_{\vcx_0}(\mxA^s)^\tr + \sum_{k=0}^{\min(s,t) - 1}\mxA^{t-1-k}\mxS_\vcw(\mxA^{s-1-k})^\tr,
    \end{split}
\end{equation}
for any $t,s = 1,2,\cdots, \nD-1$. 
One can easily see that for $p(\mathbf{y}|\mathbf{x},\vctheta, \mathbf{u})$, we have similar arguments. For $t = 1,2,\cdots, \nD-1$, let $\Xi_t(\vctheta)$ be define as
\begin{equation}\label{eqn:Xi_t}
    \Xi_t(\vctheta) := {\mxC}_0 + \sum_{i=1}^{\Nu} {\mxC}_i \vcu_{t,i},
\end{equation}
where the dependence of $\Xi_t(\vctheta)$ on the input measurement vector $\mathbf{u}$ is omitted for ease of notation, and assumed to be implicitly known. 
Following the problem settings introduced in Section~\ref{sec:pf}, we know that $\vcv_0, \vcv_1, \cdots, \vcv_{\nsD-1}$ are zero-mean independent Gaussian random vectors with individual covariance $\mxS_\vcv$.
Therefore, according to \eqref{eqn:dynamics2}, one can see that
\begin{equation}
\begin{split}
    \mu_t^{(\mathbf{y})}(\vctheta,\mathbf{x}) 
    :=
    \Exp
    \big[
    \vcy_t 
    \big| 
    \mathbf{x}, \vctheta, \mathbf{u}
    \big]
     &=\ 
     \Xi_t(\vctheta) \vcx_t + \mxD \vcu_t,
\end{split}
\end{equation}
and
\begin{equation}
\begin{split}
    \mxS_{t,s}^{(\mathbf{y})}({\vctheta}) 
    :=
    \Exp
    \big[
    (\vcy_t - \mu_t^{(\mathbf{y})})(\vcy_s - \mu_s^{(\mathbf{y})})^\tr 
    \big| 
    \mathbf{x}, \vctheta, \mathbf{u}
    \big] 
    = 
    \begin{cases}
        \mxS_\vcv, & \text{if }\  t = s,\\
        \mathbf{0}, & \text{if }\  t \ne s,\\
    \end{cases}
\end{split}
\end{equation}
for any $t,s = 0,1,\cdots, \nsmD-1$.
Note that $\mxS_{t,s}^{(\mathbf{y})}({\vctheta})$  does not depend on the input measurement vector $\mathbf{u}$, and thus, the argument $\mathbf{u}$ is omitted for the ease of notation and brevity.
Hence, we know that $\mathbf{y}$ has  Gaussian distribution as $\mathcal{N}(\mu^{(\mathbf{y})}(\vctheta,\mathbf{x}), \mxS^{(\mathbf{y})}(\vctheta,\mathbf{x}))$,
where 
$\mu^{(\mathbf{y})}({\vctheta},\mathbf{x})$ and $\mxS^{(\mathbf{y})}({\vctheta})$ can be written as
\begin{equation}\label{eq:mean y}
        \mu^{(\mathbf{y})}({\vctheta},\mathbf{x}) 
        := 
        \big[
        \mu_t^{(\mathbf{y})}(\vctheta,\mathbf{x})
        \big]_{t\,=\,0}^{\nsD-1}
        =
        \begin{bmatrix}
        \Xi_0(\vctheta) \vcx_0 + \mxD \vcu_0 \\
        \Xi_1(\vctheta) \vcx_1 + \mxD \vcu_1 \\
        \vdots \\
        \Xi_{\nsD-1}(\vctheta) \vcx_{\nsD-1} + \mxD \vcu_{\nsD-1}\\
        \end{bmatrix},
\end{equation}
and
\begin{equation}\label{eq:covariance y}
    \mxS^{(\mathbf{y})}({\vctheta})
    := 
    \big[
    \mxS_{t,s}^{(\mathbf{y})}({\vctheta}) 
    \big]_{s,t\,=\,0}^{\nsD-1}
    =
    \begin{bmatrix}
        \mxS_\vcv & \mathbf{0} & \cdots & \mathbf{0}\\
        \mathbf{0} & \mxS_\vcv & \cdots & \mathbf{0}\\
        \vdots & \vdots & \ddots & \vdots\\
        \mathbf{0} & \mathbf{0} & \cdots & \mxS_\vcv\\
    \end{bmatrix}
    = 
    \mathbf{I}_{\nsD}\otimes \mxS_\vcv.
\end{equation}
For further simplification, we additionally 
define $\Xi(\vctheta)$ and $\mathbf{D}(\vctheta)$ as  
\begin{equation}
    \Xi(\vctheta) := \begin{bmatrix}
        \Xi_0(\vctheta) & \mathbf{0} & \cdots & \mathbf{0}\\
        \mathbf{0} & \Xi_1(\vctheta) & \cdots & \mathbf{0}\\
        \vdots & \vdots & \ddots & \vdots\\
        \mathbf{0} & \mathbf{0} & \cdots & \Xi_{\nsD-1}(\vctheta)\\
    \end{bmatrix},
\end{equation}
and
\begin{equation}
    \mathbf{D}(\vctheta) := 
    \begin{bmatrix}
        \mxD & \mathbf{0} & \cdots & \mathbf{0}\\
        \mathbf{0} & \mxD & \cdots & \mathbf{0}\\
        \vdots & \vdots & \ddots & \vdots\\
        \mathbf{0} & \mathbf{0} & \cdots & \mxD\\
    \end{bmatrix}
    = \mathbf{I}_\nsD \otimes \mxD,
\end{equation}
respectively.
Accordingly, due to \eqref{eq:mean y}, for $\mu^{(\mathbf{y})}({\vctheta},\mathbf{x})$, 
we have 
\begin{equation}
\mu^{(\mathbf{y})}({\vctheta},\mathbf{x}) = \Xi(\vctheta)\mathbf{x} + \mathbf{D}(\vctheta)\mathbf{u}.
\end{equation}

%
Applying  the law of total probability and using the Gaussian distributions derived above, we can express the likelihood function as
\begin{equation}\label{eq: likelihoog int}
\begin{split}
        p(\mathbf{y}|\vctheta, \mathbf{u})
        & =
        \int_{\Xbb}
        \,
        p(\mathbf{y}|\mathbf{x}, \vctheta, \mathbf{u})p(\mathbf{x}|\vctheta, \mathbf{u}) \drm \mathbf{x}, 
        \\&= 
        \frac{(2\pi)^{-\frac{1}{2}(\Nx+\Ny)\nsD}}{|\Sbfxtheta|^{\frac{1}{2}} |\Sbfytheta|^{\frac{1}{2}}} \int_{\Xbb} 
        \exp 
        \bigg( -\frac{1}{2} \Big[ 
        \big(\mathbf{y}-\mubfythetax \big)^\tr
        \mxS^{{(\mathbf{y})}}\!(\vctheta)^{-1}
        \big( \mathbf{y}-\mubfythetax \big) \\
        &
        \qquad\qquad \qquad\qquad \qquad\qquad \qquad\qquad \qquad 
        + 
        \big(\mathbf{x}-\mubfxtheta \big)^\tr
        \mxS^{{(\mathbf{x})}}\!(\vctheta)^{-1}
        \big( \mathbf{x}-\mubfxtheta \big) 
        \Big] \bigg) 
        \drm \mathbf{x}.
\end{split}
\end{equation}
%
%
One can see that the current form of the likelihood function described above is excessively complex to directly obtain $\vcthetaML$ using \eqref{eq:mean x}. Thus, we simplify it to formulate a tractable optimization problem, as outlined in the following proposition.

\begin{proposition}\label{Pro:loglikelihood}
Define function 
$\, J:\Vbb\to\extendedR \,$ 
as
\begin{equation}\label{eq:cost function ML}
    \begin{split}
    J({\vctheta}) :=&\,  
    \logdet\big(
        \Xitheta^\tr 
        \Sbfytheta^{-1} 
        \Xitheta 
        + 
        \Sbfxtheta^{-1}
    \big) 
    \, + \, 
    \logdet\big(\Sbfxtheta \big) 
    \, + \, 
    \logdet\big(\Sbfytheta \big)
     \\ &\quad
    +\Big(\big(\mathbf{u}^\tr 
    \mathbf{D}(\vctheta)^\tr 
    - \mathbf{y}^\tr + \mu^{(\mathbf{x})}({\vctheta})^\tr\Xi({\vctheta})^\tr\big)
    \big(\Xi({\vctheta})\Sbfxtheta \Xi({\vctheta})^\tr 
    + \Sbfytheta \big)^{-1}
    \big(\mathbf{D}(\vctheta) \mathbf{u} - \mathbf{y} + \Xi({\vctheta})\mubfxtheta \big)\Big). 
    \end{split}
\end{equation}
Then, the maximum likelihood problem 
\eqref{eqn:ML_generic}
is equivalent to minimizing $J$ over $\Theta$, i.e., we have
        \begin{equation}\label{eqn:min_J}
            \begin{split}
                \vcthetaML = \argminOp_{{\vctheta}\in\Theta} \ J({\vctheta}) 
            \end{split}
        \end{equation}
\end{proposition}
\begin{proof}
        The proof is provided in Appendix \ref{app:proof 1}.
\end{proof}


\subsection{Maximum Likelihood Estimation: Toward First-Order Approaches}\label{sec:Maximum Likelihood Estimation: Toward first-order approaches}
In the preceding discussion, we outlined the formulation of the log-likelihood estimation method and demonstrated that, to estimate the parameter vector $\vctheta$ using the ML approach, we need to solve the optimization problem presented in \eqref{eqn:min_J}. One can easily see that $J$ is a nonlinear function of the parameters $\vctheta$. 
Accordingly, to solve \eqref{eqn:min_J}, we can employ first-order optimization methods \cite{nocedal1999numerical}, such as \emph{Gradient Descent}.  
To implement these methods, it is necessary to obtain the gradient vector $\nabla_\vctheta J$ by deriving the derivatives of the objective function $J$ with respect to each parameter.

Note that $J$ depends on $\mu^{(\mathbf{x})}, \mxS^{(\mathbf{x})}, \mxS^{(\mathbf{y})}, \Xi, \mathbf{D}$, each of which is a function of $\vctheta$. Consequently, we can apply matrix calculus techniques \cite{petersen2008matrix}, particularly \emph{matrix chain rule} outlined in Lemma~\ref{lem:chain_rule} in Appendix~\ref{app:auxiliary lemma}, to derive the gradient vector $\nabla_\vctheta \, J$. More precisely, we first obtain the derivatives of $J$ with respect to the mentioned terms, as given by
\begin{align}
        \pdv{J(\vctheta)}{\mu^{(\mathbf{x})}} 
        =& \, 
        2 \Xitheta^\tr
        \big(\Sbfytheta + \Xitheta \Sbfxtheta \Xitheta^\tr\big)^{-1}
        \big(\mathbf{D}(\vctheta) \mathbf{u} - \mathbf{y}\big) 
        + 
        2\Xi^\tr(\vctheta) \big(\Sbfytheta + \Xitheta \Sbfxtheta \Xitheta^\tr\big)^{-1}\Xitheta \mubfxtheta ,\label{eq:derivative to mux}\\ 
        \pdv{J(\vctheta)}{\mxS^{(\mathbf{x})}} =& -\Xitheta^\tr\big(\Sbfytheta + \Xitheta \Sbfxtheta \Xitheta^\tr\big)^{-1} \big(\mathbf{D}(\vctheta)\mathbf{u} - \mathbf{y} + \Xitheta \mubfxtheta\big) \big(\mathbf{D}(\vctheta)\mathbf{u} - \mathbf{y} + \Xitheta \mubfxtheta\big)^\tr 
        \nonumber
        \\
        &\big(\Sbfytheta + \Xitheta \Sbfxtheta \Xitheta^\tr\big)^{-1}\Xitheta  - \Sbfxtheta^{-1}\big(\Sbfxtheta^{-1} + \Xitheta^\tr \Sbfytheta^{-1} \Xitheta\big)^{-1} \Sbfxtheta^{-1} + \Sbfxtheta^{-1},
        \label{eq:derivative to Sx}\\
        \pdv{J(\vctheta)}{\Xi} =& -2\big(\Sbfytheta + \Xitheta \Sbfxtheta \Xitheta^\tr\big)^{-1} \big(\mathbf{D}(\vctheta)\mathbf{u} - \mathbf{y} + \Xitheta \mubfxtheta\big) \big(\mathbf{D}(\vctheta)\mathbf{u} - \mathbf{y} + \Xitheta \mubfxtheta\big)^\tr (\Sbfytheta 
        \nonumber\\
        &+\Xitheta \Sbfxtheta \Xitheta^\tr)^{-1}\Xitheta \Sbfxtheta +2\big(\Sbfytheta + \Xitheta \Sbfxtheta \Xitheta^\tr\big)^{-1}\big(\mathbf{D}(\vctheta)\mathbf{u}-\mathbf{y} + \Xitheta \mubfxtheta\big) \mubfxtheta^\tr 
        \nonumber\\
        &+ 2\Sbfytheta^{-1} \Xitheta \big(\Xitheta^\tr \Sbfytheta \Xitheta + \Sbfxtheta^{-1}\big)^{-1},
        \label{eq:derivative to Xi}\\
        \pdv{J(\vctheta)}{\mathbf{D}} =& \, 2\big(\Sbfytheta + \Xitheta \Sbfxtheta \Xitheta^\tr\big)^{-1}\mathbf{D}(\vctheta)\mathbf{u}\mathbf{u}^\tr - 2\big(\Sbfytheta + \Xitheta \Sbfxtheta \Xitheta^\tr\big)^{-1}(\mathbf{y}+\Xitheta \mubfxtheta)\mathbf{u}^\tr
        \nonumber\\
        &+ 4\big(\Sbfytheta + \Xitheta \Sbfxtheta \Xitheta^\tr\big)^{-1}\Xitheta \mubfxtheta \mathbf{u}^\tr,
        \label{eq:derivative to D}\\
        \pdv{J(\vctheta)}{\mxS^{(\mathbf{y})}} =& -\big(\Sbfytheta + \Xitheta \Sbfxtheta \Xitheta^\tr\big)^{-1} \big(\mathbf{D}(\vctheta)\mathbf{u} - \mathbf{y} + \Xitheta \mubfxtheta\big) \big(\mathbf{D}(\vctheta)\mathbf{u} - \mathbf{y} + \Xitheta \mubfxtheta\big)^\tr \big(\Sbfytheta 
        \nonumber\\
        &+ \Xitheta \Sbfxtheta \Xitheta^\tr\big)^{-1}  - \Sbfytheta^{-1} \Xitheta\big(\Xitheta^\tr \Sbfytheta^{-1} \Xitheta + \Sbfxtheta^{-1}\big)^{-1}\Xitheta^\tr\Sbfytheta^{-1} + \Sbfytheta^{-1}.
        \label{eq:derivative to Sy}
\end{align}
Subsequently, we need to determine the derivatives of $J$ with respect to each parameter. To this end, we define the matrices $\mxE_{i,j}$ and $f_{i,j}^{\ n}(\mxA)$ respectively as $\mxE_{i,j} = \vce_i\vce_j^\tr$ and $f_{i,j}^{\ n}(\mxA) := \sum_{r=0}^{n-1} \mxA^{r}\mxE_{i,j}\mxA^{n-1}$, for $i,j=1,\ldots,\Nx$. 
Similar to the the previous step, we apply matrix calculus techniques. 
Thus, 
for $i,j=1,\ldots,\Nx$, $\, k = 1,\ldots,\Nu$, $\, l,m=1,\ldots,\Ny$, and $\, t,s = 0,1,\cdots,\nD-1$, we obtain
\begin{align}
        \pdv{\mu^{(\mathbf{x})}_t(\vctheta)}{\left[\mxA \right]_{i,j}} &= \begin{cases}
            0 &\text{if }\  t=0,\\
            \mxE_{i,j} \mu_{\vcx_0} &\text{if }\  t=1,\\
            \sum_{r=0}^{t-1} \mxA^{r}\mxE_{i,j}\mxA^{t-1-r}\mu_{\vcx_0} + \sum_{k=0}^{t-2}f_{i,j}^{\ t-1-k}(\mxA)\mxB\vcu_k&\text{if }\  t\geq2,
        \end{cases}
        \label{eq:J derivative A2}\\
        \pdv{\mxS^{(\mathbf{x})}_{t,s}(\vctheta)}{\left[ \mxA \right]_{i,j}} &= \begin{cases}
            0 &\text{if }\  t,s=0,\\
            \mxS_{\vcx_0}\left[ f^{\ s}_{i,j}(\mxA) \right]^\tr &\text{if }\  t=0,s\geq 1,\\
            f^{\ t}_{i,j}(\mxA)\mxS_{\vcx_0} &\text{if }\  t \geq 1, s=0,\\
            \mxA^{t}\mxS_{\vcx_0} \left[ f^{\ s}_{i,j}(\mxA)\right]^\tr + f^{\ t}_{i,j}(\mxA)\mxS_{\vcx_0}\mxA^{s^\tr}  &\text{if }\ 
             t,s=1,\\
            \mxA^{t}\mxS_{\vcx_0} \left[ f^{\ s}_{i,j}(\mxA)\right]^\tr + f^{\ t}_{i,j}(\mxA)\mxS_{\vcx_0}\mxA^{s^\tr} + \mxS_\vcw \left[ f_{i,j}^{\ s-1}(\mxA) \right]^\tr &\text{if }\ 
             t=1,s\geq 2,\\
            \mxA^{t}\mxS_{\vcx_0} \left[ f^{\ s}_{i,j}(\mxA)\right]^\tr + f^{\ t}_{i,j}(\mxA)\mxS_{\vcx_0}\mxA^{s^\tr} + f_{i,j}^{\ t-1}(\mxA) \mxS_\vcw  &\text{if }\ 
             t \geq 2,s = 1,\\
            \mxA^{t}\mxS_{\vcx_0} \left[ f^{\ s}_{i,j}(\mxA)\right]^\tr + f^{\ t}_{i,j}(\mxA)\mxS_{\vcx_0}\mxA^{s^\tr} &\\
            \qquad+ \sum_{k=0}^{\min (s,t)-1} \left( \mxA^{t-1-k}\mxS_\vcw \left[ f_{i,j}^{\ s-k-1}(\mxA) \right]^\tr + f_{i,j}^{\ t-k-1}(\mxA)  \mxS_\vcw \mxA^{(s-k-1)^\tr}  \right) &\text{if }\ 
             t\geq 2,s\geq 2, \label{eq:J derivative A3}\\
        \end{cases}\\
        \pdv{\mu^{(\mathbf{x})}_t(\vctheta)}{\left[\mxB \right]_{i,k}} 
        &= 
        \sum_{\tau=0}^{t-1} \left[\mxA^{t-1-\tau} \right]_i \mathbf{u}_{\tau,k}, 
        \label{eq:J derivative B}\\
        \pdv{\mxS^{(\mathbf{x})}_{t,s}(\vctheta)}{\left[ \mxS_\vcw \right]_{i,j}} 
        &= 
        \sum_{\tau=0}^{\min(s,t)-1} \left[\mxA^{t-1-\tau}\right]_i \left[\mxA^{s-1-\tau}\right]_j^\tr,
        \label{eq:J derivative Sw}\\
        \pdv{\mxS^{(\mathbf{x})}_{t,s}(\vctheta)}{\left[ \mxS_{x_0} \right]_{i,j}} &= \left[ \mxA^t \right]_i \left[ \mxA^s \right]_j^\tr, 
        \label{eq:J derivative Sx0}
\end{align}
To bridge \eqref{eq:derivative to mux}--\eqref{eq:derivative to Sy}  
and \eqref{eq:J derivative A2}--\eqref{eq:J derivative Sx0}, we need to employ the matrix chain rule, i.e., using Lemma~\ref{lem:chain_rule}, we have
\begin{align}
        \pdv{J(\vctheta)}{\left[ \mxA \right]_{i,j}} &= \pdv{J(\vctheta)}{\mu^{(\mathbf{x})}(\vctheta)}{}^\tr\, \pdv{\mu^{(\mathbf{x})}(\vctheta)}{\left[\mxA \right]_{i,j}} 
        + 
        \trace\left(\pdv{J(\vctheta)}{\mxS^{(\mathbf{x})}(\vctheta)}{}^\tr\, \pdv{\mxS^{(\mathbf{x})}(\vctheta)}{\left[ \mxA \right]_{i,j}} \right), 
        \label{eq:J derivative A1}\\
        \pdv{J(\vctheta)}{\left[ \mxB \right]_{i,k}} &=  \pdv{J(\vctheta)}{\mu^{(\mathbf{x})}(\vctheta)}^\tr \pdv{\mu^{(\mathbf{x})}(\vctheta)}{\left[\mxB \right]_{i,k}}, 
        \\
        \pdv{J(\vctheta)}{[\mxC_0]_{l,j}} &= \trace \left( \pdv{J(\vctheta)}{\Xi(\vctheta)}^\tr \left( \mathbf{I}_{\nsD} \otimes \mxE_{l,j} \right) \right)
        = \sum_{t=1}^\nsD \left[ \pdv{J(\vctheta)}{\Xi(\vctheta)} \right]_{tl,tj}, 
        \label{eq:J derivative C0}\\
        \pdv{J(\vctheta)}{[\mxC_k]_{l,j}} &= \trace \left( \pdv{J(\vctheta)}{\Xi(\vctheta)}^\tr \left( \mathbf{I}_{\nsD} \otimes \mxE_{l,j} \right) \right)
        = \sum_{t=1}^\nsD \vcu_{k,t-1} \left[ \pdv{J(\vctheta)}{\Xi(\vctheta)} \right]_{tl,tj}, 
        \label{eq:J derivative Ci}\\
        \pdv{J(\vctheta)}{[\mxD]_{l,k}} &= \trace\left(\pdv{J(\vctheta)}{\mathbf{D}(\vctheta)}^\tr \left( \mathbf{I}_{\nsD} \otimes \mxE_{l,k} \right) \right)
         = \sum_{t=1}^\nsD \left[ \pdv{J(\vctheta)}{\mathbf{D}(\vctheta)} \right]_{tl,tk}, 
         \label{eq:J derivative D}\\
        \pdv{J(\vctheta)}{\left[\mxS_\vcw \right]_{i,j}} &= \trace \left( \pdv{J(\vctheta)}{\mxS^{(\mathbf{x})}(\vctheta)}^\tr \pdv{\mxS^{(\mathbf{x})}(\vctheta)}{\left[ \mxS_\vcw \right]_{i,j}} \right), 
        \\
        \pdv{J(\vctheta)}{[\mxS_\vcv]_{l,m}} &= \trace\left(\pdv{J(\vctheta)}{\mxS^{(\mathbf{y})}(\vctheta)}^\tr \left( \mathbf{I}_{\nsD} \otimes \mxE_{l,m} \right) \right)
         = \sum_{t=1}^\nsD \left[ \pdv{J(\vctheta)}{\mxS^{(\mathbf{y})}(\vctheta)} \right]_{tl,tm}, 
        \label{eq:J derivative Sv}\\
        \pdv{J(\vctheta)}{ \mu_{x_0}(\vctheta) } &= \begin{bmatrix}
            (\mxA^{0})^{\tr} & (\mxA^{1})^\tr & \cdots (\mxA^{\nsD-1})^\tr
        \end{bmatrix} \pdv{J(\vctheta)}{\mu^{(\mathbf{x}})(\vctheta)},
        \label{eq:J derivative mux0}\\
        \pdv{J(\vctheta)}{[\mxS_{x_0}]_{i,j}} &= \trace\left(\pdv{J(\vctheta)}{\mxS^{(\mathbf{x})}(\vctheta)}^\tr \pdv{\mxS^{(\mathbf{x})}(\vctheta)}{\left[ \mxS_{x_0} \right]_{i,j}} \right), 
        \label{eq:J derivative muS0}
\end{align}
for $i,j=1,\ldots,\Nx$, $k = 1,\ldots,\Nu$, and $l,m=1,\ldots,\Ny$.
The proposed ML approach for estimating the parameter vector $\vctheta$ is outlined in Algorithm~\ref{Al:ML ID}.

\begin{remark}
    Given $\nabla_\vctheta J$, , alternative first-order optimization methods \cite{nocedal1999numerical} can be employed instead of gradient descent. For instance, the \emph{Broyden-Fletcher-Goldfarb-Shanno} (BFGS) algorithm or its \emph{limited-memory} variant (L-BFGS) can be used to achieve faster convergence. 
\end{remark}
\begin{remark}
One can easily see that the cost function $J$ is nonconvex. Consequently, the gradient descent algorithm, or any other first-order method, may converge to the local optima of $J$ rather than its global optimum, which can potentially affect the estimation performance and the accuracy of parameter identification.
\end{remark}
\begin{remark}\label{rem:comp_comp}
It can be shown that the derivative calculation used in the proposed ML method has a computational complexity of $O(\nD^3)$. One should note that this can lead to significant computational time and require substantial resources, particularly when the set of measurement data is considerably large.
\end{remark}

In the remainder of this section, we present a Monte Carlo empirical analysis of the computational complexity of Algorithm \ref{Al:ML ID}.

\begin{algorithm}[t]
    \caption{ML Estimation Method for Identification of Linear Dynamics with Bilinear Observation Models}\label{Al:ML ID}
    \begin{algorithmic}
        \Statex \textbf{Input:} $\Dcal$.
        \Statex \textbf{Output:} $\vctheta$.
        \State \textbf{Initial guess:} $\hat{\vctheta}_0$
        \State $k \gets 0$
        \While{$\mathbf{1}$}
            \State Current parameters estimates: $\hat{\vctheta} \gets \hat{\vctheta}_k$
            \State Compute $\mxS^{(\mathbf{x})}{(\hat{\vctheta})}, \mxS^{(\mathbf{y})}{(\hat{\vctheta})}, \Xi{(\hat{\vctheta})}, \mathbf{D}{(\hat{\vctheta})}, \mu^{(\mathbf{x})}{(\hat{\vctheta})}$.
            \State Compute $\pdv{J(\vctheta)}{\mu^{(\mathbf{x})}(\vctheta)}|_{\vctheta = \hat{\vctheta}}, \pdv{J(\vctheta)}{\mxS^{(\mathbf{x})}(\vctheta)}|_{\vctheta = \hat{\vctheta}}, \pdv{J(\vctheta)}{\Xi(\vctheta)}|_{\vctheta = \hat{\vctheta}}, \pdv{J(\vctheta)}{\mathbf{D}(\vctheta)}|_{\vctheta = \hat{\vctheta}}$ and $\pdv{J(\vctheta)}{\mxS^{(\mathbf{y})}(\vctheta)}|_{\vctheta = \hat{\vctheta}}$ as in \eqref{eq:derivative to mux}--\eqref{eq:derivative to Sy}. 
            \State Find the derivatives $\pdv{J(\vctheta)}{\vctheta}|_{\vctheta = \hat{\vctheta}}$ 
            \eqref{eq:J derivative A2} -- \eqref{eq:J derivative Sx0} and \eqref{eq:J derivative A1} -- \eqref{eq:J derivative muS0}.
            \State Update $\hat{\vctheta}_{k+1}$ using gradient descent.
            \If{$\lVert \hat{\vctheta}_{k+1} -- \hat{\vctheta}_k \rVert< \epsilon$,}
            \State break
            \Else
            \State $k \gets k+1$
            \EndIf
        \EndWhile
        \State $\vctheta \gets \hat{\vctheta}_k$
    \end{algorithmic}
    \label{alg_1}
\end{algorithm}

\subsection{Computational Complexity Analysis of Maximum Likelihood: An Empirical Evaluation}
\label{sec:Complexity Analysis}
To further elaborate on the computational complexity of Algorithm \ref{Al:ML ID}, we perform Monte Carlo numerical analysis on simple example. To this end, we consider a one-dimensional system governed by the equations
\begin{equation}\label{eqn:simple_system}
\begin{split}
    \vcx_{t+1} &= 0.6\,\vcx_t + 0.45\,\vcu_t + \vcw_t,\\
    \vcy_t &= (0.3 + 0.1\vcu_t)\, \vcx_t + \vcv_t,
\end{split}
\end{equation}
with a mean initial state of $\mu_{\vcx_0} = 1$. 
For clarity and simplicity in the discussion, the variance of the initial state is set to a negligible value.
We conduct Monte Carlo numerical experiments considering different lengths for the set of data, specifically $\nD = 20, 40, \ldots, 200$. With respect to each of the considered lengths for the dataset, $100$ different realizations of the initial state, random binary input sequences, process noise, and measurement noise are generated. Subsequently, the state and output trajectories are obtained according to \eqref{eqn:simple_system}. 
For the sake of clarity in the discussion, the focus is on moderately noisy conditions, i.e., SNR levels ranging from $15$\,dB to $20$\,dB, and the noise variances are adjusted accordingly.
Algorithm~\ref{Al:ML ID} is employed to estimate the system parameters, initialized to half of their true values, i.e., $\hat{\vctheta}_0 = \frac{1}{2}\vctheta$.

Figure \ref{img:time complexity} shows the results of the performed Monte Carlo experiments. In the figure, $\mu$ and $\sigma$ denote the mean and standard deviation values of the running time, respectively. The shaded region indicates the range $(\mu -\sigma, \mu + \sigma)$, i.e., 68\%-confidence range. The dark blue curve corresponds to the mean running time. The figure highlights how the computational time increases with the length of the dataset, which is essentially follows the point mentioned in Remark~\ref{rem:comp_comp} on $O(\nD^3)$ computational complexity due to the calculation of the derivatives with respect to $\mxA$.
%

It is important to note that the proposed approach may impose high computational demands, particularly when dealing with large datasets, potentially making it impractical for real-world applications. To address this limitation, in the remainder of this paper, we introduce an alternative estimation method based on the Expectation--Maximization (EM) algorithm, which is significantly more computationally efficient and better suited for large datasets.


\begin{figure}[t!]
   \centering
   \includegraphics[width = 0.6\textwidth,trim={2cm 8.7cm 3cm 9cm},clip]{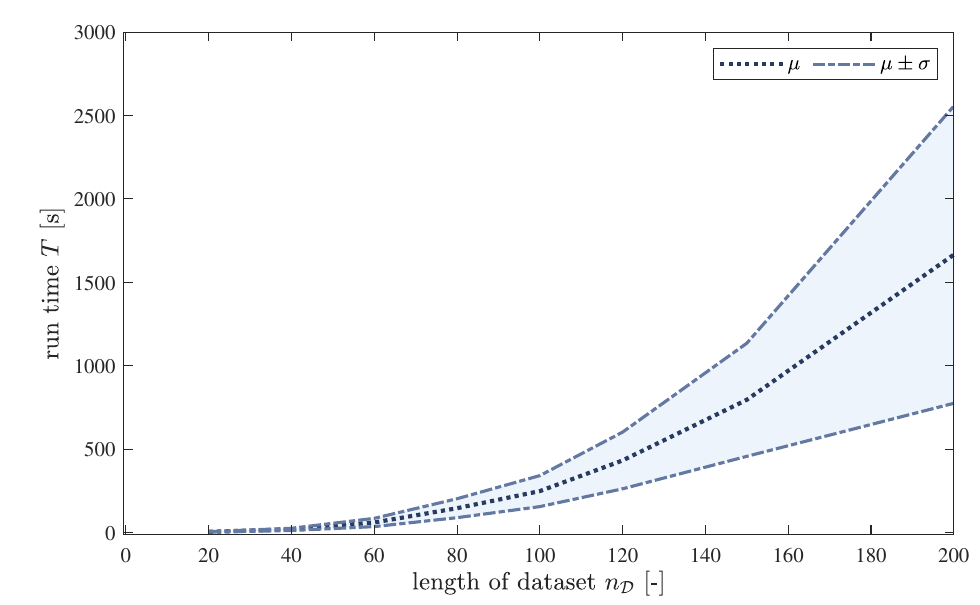}
   \caption{Time complexity analysis of algorithm \ref{Al:ML ID}. The red curve line is the mean running time. The red shaded region contains the interval $(\mu - \sigma, \mu + \sigma)$.}
   \label{img:time complexity}
\end{figure}
\section{Expectation-Maximization Approach}\label{sec:An Expectation-Maximization Approach}
In the previous section, we proposed an ML estimation approach for the system identification problem outlined in Section~\ref{sec:pf}. As noted earlier, obtaining the derivative of the likelihood function can be computationally intensive, particularly in scenarios involving large datasets, making the discussed ML approach limited and less suitable for real-world applications. To address this limitation, we propose an alternative method in this section based on the expectation-maximization (EM) framework. Compared to the introduced ML estimation method, the EM approach provides a more numerically efficient iterative framework, particularly when solving the maximum likelihood optimization problem is computationally challenging.
The EM algorithm proceeds iteratively, with each iteration consisting of two main steps, namely the expectation (E) step and the maximization (M) step. In the E step, a probability distribution for the state trajectories is determined using the measurement data and an initial parameter estimate, e.g., through employing a suitably adapted version of the Rauch-Tung-Striebel smoother \cite{rauch1965maximum}. In the subsequent M step, the estimation of system parameters is then updated by maximizing the expected log-likelihood function derived from the state estimates and distributions obtained in the E step. The resulting updated parameter estimates are then employed in the next EM iteration. 
In the following, we describe the proposed EM approach for addressing the identification problem outlined in Section~\ref{sec:pf}. More precisely, we describe the E step, namely details of the appropriate adaptation of the Rauch-Tung-Striebel smoother to the considered dynamics. Additionally, we outline how to solve the optimization problem in the M step, i.e., showing that the expected log-likelihood maximization admits a tractable closed-form solution.

\subsection{Rauch-Tung-Striebel Smoother}
\label{sec:RTS smoother}

%
In the E step, we obtain an estimation for the state trajectory. To this end, an appropriate variant of the Rauch-Tung-Striebel (RTS) smoother can be employed to derive a probability distribution for the state variables and, subsequently, estimate them.
The RTS smoother employs the Kalman filter \cite{kalman1960new}, which utilizes an estimate of parameters, the measurement data and the statistics of process and measurement noise to provide an initial estimation for the state variables of the system. It then applies a smoothing step to improve the accuracy of the estimated state variables.
Given the measurement data and an estimation of the parameters, the Kalman filter provides a Gaussian distribution for the state variables. In particular, for any time step $t = 0, 1, \ldots, \nD-1$, the Kalman filter computes the mean and covariance for the state variables estimates, denoted by $\hat{\vcx}_{t|t}$ and $\mxP_{t|t}$, respectively, as 
\begin{equation}
    \hat{\vcx}_{t|t} := \mathbb{E}\big[\vcx_t | \vcy_0, \ldots, \vcy_t, \hat{\vctheta}\big], 
\end{equation}
and
\begin{equation}
    \mxP_{t|t} := \mathbb{E}\big[(\vcx_t - \hat{\vcx}_{t|t})(\vcx_t - \hat{\vcx}_{t|t})^\tr | \vcy_0, \ldots, \vcy_t, \hat{\vctheta}\big],
\end{equation}
where $\hat{\vctheta}$ represents the given estimates of the parameters. 
%
%
%
%
Note that, for $t=0$, $\hat{\vcx}_{t|t}$ and $\mxP_{t|t}$ corresponds to $\hat{\mu}_{\vcx_0}$ and $\hat{\mxS}_{\vcx_0}$, respectively, and can thus be obtained from $\hat{\vctheta}$.
To estimate the state variables for time steps $t = 1, \ldots, \nD-1$, an additional inner iteration needs to be performed. More precisely, based on the dynamics described in \eqref{eqn:dynamics1} and \eqref{eqn:dynamics2}, the Kalman filter employs an iterative scheme as
\begin{equation}\label{eq:kalman filter}
    \begin{split}
        \hat{\vcx}_{t|t}   =&\,\, \hat{\vcx}_{t|t-1} + \mxK_t(\vcy_t - \hat{\Xi}_t \hat{\vcx}_{t|t-1}),\\
        \mxP_{t|t}   =&\,\, (\mathbb{I}_{\Nx} - \mxK_t \hat{\Xi}_t)\mxP_{t|t-1},\\
        \mxK_t   =&\,\, \mxP_{t|t-1}\hat{\Xi}_t^\tr[\hat{\Xi}_t \mxP_{t|t-1}\hat{\Xi}_t^\tr + \hat{\mxS}_\vcv]^{-1},\\
        \hat{\vcx}_{t+1|t}   =&\,\, \hat{\mxA}\hat{\vcx}_{t|t} + \hat{\mxB} \vcu_{t},\\
        \mxP_{t+1|t}   =&\,\, \hat{\mxA}\mxP_{t|t}\hat{\mxA}^\tr + \hat{\mxS}_\vcw,\\
    \end{split}
\end{equation}
where $\hat{\Xi}_t$ is defined as
\begin{equation}
\hat{\Xi}_t = \hat{\mxC}_0 + \sum_{i=1}^{\Nu} \hat{\mxC}_iu_{t,i}.    
\end{equation}
Subsequently, the RTS smoother improves the accuracy of the state estimates following an additional backward procedure, recursively smoothing the estimates obtained from the Kalman filter. More specifically, we obtain the mean and covariance of the state variables estimation, namely  $\hat{\vcx}_{t|{\nD}}$ and $\mxP_{t|{\nD}}$ defined respectively as $\hat{\vcx}_{t|{\nD}} := \mathbb{E}[\vcx_t|\mathbf{y},\hat{\vctheta}]$ and $\mxP_{t|{\nD}} := \mathbb{E}[(\vcx_t - \hat{\vcx}_{t|{\nD}})(\vcx_t - \hat{\vcx}_{t|{\nD}})^\tr|\mathbf{y},\hat{\vctheta}]$. To this end, for $t = 0,1, \ldots, \nD-2$, we employ a recursive procedure as 
\begin{equation}\label{eq:mean and covariance}
    \begin{split}
        \hat{\vcx}_{t|{\nD}}   =&\,\, \hat{\vcx}_{t|t} + {\mxL}_t(\hat{\vcx}_{t+1|\nD} - \hat{\vcx}_{t+1|t}),\\
        \mxP_{t|{\nD}}   =&\,\, \mxP_{t|t} + {\mxL}_t(\mxP_{t+1|{\nD}} - \mxP_{t+1|t}){\mxL}_t^\tr,\\
        \mxL_t   =&\,\, \mxP_{t|t}\hat{\mxA}^\tr\mxP_{t+1|t}^{-1},\\
    \end{split}
\end{equation}
%
Moreover, one can additionally obtain 
\begin{equation} \label{eq:correlation}
    \begin{split}
        \mathbb{E}[\vcx_t \vcx_t^{\tr}|\mathbf{y}, \hat{\vctheta}]   =&\,\, \hat{\vcx}_{t|{\nD}}\hat{\vcx}_{t|{\nD}}^\tr + \mxP_{t|{\nD}},
        \qquad\qquad\quad \ \;\!
        \forall\,t = 0,1, \ldots, \nD-1,
        \\
        \mathbb{E}[\vcx_{t+1} \vcx_t^{\tr}|\mathbf{y}, \hat{\vctheta}]   =&\,\, \hat{\vcx}_{t+1|{\nD}}\hat{\vcx}_{t|{\nD}}^\tr + \mxP_{t+1|{\nD}}\mxL_t^\tr,
        \qquad
        \forall\,t = 0,1, \ldots, \nD-2,
        \\
    \end{split}
\end{equation}
which will be utilized later in the EM algorithm, as discussed in the remainder of this section.

\subsection{Parameters Estimation using EM Algorithm} 
\label{sec:Parameters Estimation using EM algorithm}

In the M step, the estimated parameter, denoted by $\hat{\theta}$ in the E step, is updated by maximizing an expected log-likelihood function. Specifically, let  $\hat{\vctheta}_k$ denote the estimate of $\vctheta$ at iteration $k$ of the EM algorithm, and define function $Q(\cdot|\hat{\vctheta}_k):\Theta\to\Rbb$ as
\begin{equation} \label{eqn:Q_theta_theta_k}
    Q(\vctheta|\hat{\vctheta}_k) 
    := 
    \mathbb{E}_{p(\mathbf{x} |\mathbf{y},\hat{\vctheta}_k,\mathbf{u})} \big[\log p(\mathbf{x},\mathbf{y}|\vctheta,\mathbf{u})\big],    
\end{equation}
for any $\vctheta\in\Theta$.
Rather than directly maximizing the likelihood $p(\mathbf{y}|\vctheta,\mathbf{u})$ as discussed in Section~\ref{sec:ML}, the EM algorithm employs an iterative scheme that maximizes the auxiliary function $Q(\vctheta|\hat{\vctheta}_k)$ \cite{little2019statistical}.
More precisely, given $\vctheta_0\in\Theta$, an initial guess for the parameters $\vctheta$, EM iteratively generates a sequence of parameter estimations given by
\begin{equation}\label{eq:EM optimization}
    \hat{\vctheta}_{k+1}
    =
    \argmaxOp_{\vctheta\in\Theta}\,
    Q(\vctheta|\hat{\vctheta}_k),
\end{equation}
for $k\ge 0$.
The resulting iterative procedure leads to the improvement of $\log p(\mathbf{y}|\vctheta,\mathbf{u})$. Unlike the ML approach, which aims to directly identify the parameters that maximize the likelihood of the observed data, the EM algorithm indirectly increases the likelihood by successively refining $Q(\vctheta|\hat{\vctheta}_k)$.
More precisely, one can see that
\begin{equation} \label{eq:log likelihoood}
    \begin{split}
        \log p(\mathbf{y}|\vctheta,\mathbf{u}) \, \, 
        &=\,\, 
        \int_{\Xbb} p(\mathbf{x}|\mathbf{y},\hat{\vctheta}_k,\mathbf{u})
        \big[
        \log p(\mathbf{x},\mathbf{y}|\vctheta,\mathbf{u}) - \log p(\mathbf{x}|\mathbf{y},\vctheta,\mathbf{u})
        \big]\,\drm \mathbf{x}\\
        &= Q(\vctheta|\hat{\vctheta}_k) + I(\vctheta|\hat{\vctheta}_k),
    \end{split} 
\end{equation}
where function $I(\cdot|\hat{\vctheta}_k):\Theta\to\Rbb$ is defined as 
\begin{equation}
    I(\vctheta|\hat{\vctheta}_k) 
    := 
    -\, \mathbb{E}_{p(\mathbf{x} | \mathbf{y},\hat{\vctheta}_k,\mathbf{u})} 
    \big[
    \log p(\mathbf{x}|\mathbf{y},\vctheta,\mathbf{u})
    \big], 
\end{equation}
for any $\vctheta\in\Theta$.
Comparing $\log p(\mathbf{y}|\vctheta,\mathbf{u})$ with $\log p(\mathbf{y}|\hat{\vctheta}_k,\mathbf{u})$, the log-likelihood in the current iteration, we have
\begin{equation} \label{eq:Q convengence}
    \begin{split}
        \log p(\mathbf{y}|\vctheta,\mathbf{u}) 
        - 
        \log p(\mathbf{y}|\hat{\vctheta}_k,\mathbf{u}) =  Q(\vctheta|\hat{\vctheta}_k) - Q(\hat{\vctheta}_k|\hat{\vctheta}_k) + I(\vctheta|\hat{\vctheta}_k) - I(\hat{\vctheta}_k|\hat{\vctheta}_k).
    \end{split}
\end{equation}
From the Gibbs' inequality \cite{bremaud2012introduction}, we know that  $I(\vctheta|\hat{\vctheta}_k) \ge I(\hat{\vctheta}_k|\hat{\vctheta}_k)$, for any $\vctheta\in\Theta$. Accordingly, from \eqref{eq:Q convengence}, it is implied that
\begin{equation} \label{eq:Q convengence_2}
    \begin{split}
        \log p(\mathbf{y}|\vctheta,\mathbf{u}) 
        - 
        \log p(\mathbf{y}|\hat{\vctheta}_k,\mathbf{u}) 
        \ge 
        Q(\vctheta|\hat{\vctheta}_k) - Q(\hat{\vctheta}_k|\hat{\vctheta}_k).
    \end{split}
\end{equation}
for any $\vctheta\in\Theta$. 
Therefore, $Q(\vctheta|\hat{\vctheta}_k)$ provides a lower bound for $\log p(\mathbf{y}|\vctheta,\mathbf{u})$, and consequently, its successive improvement will ultimately result in improving the log-likelihood function. Hence, the EM algorithm does not find the maximizer of $\log p(\mathbf{y}|\vctheta,\mathbf{u})$ directly but improves the value through each iteration until it reaches convergence. 
%
As discussed in Section~\ref{sec:Complexity Analysis}, the ML approach offers a direct and intuitive approach to finding optimal parameters, albeit at a high computational cost. In contrast, as detailed below, the EM algorithm is better suited for large datasets, as maximizing $Q(\vctheta|\hat{\vctheta}_k)$ is computationally less expensive than maximizing the log-likelihood $\log p(\mathbf{y}|\vctheta,\mathbf{u})$. 
In the following, we derive $Q(\vctheta|\hat{\vctheta}_k)$ and subsequently demonstrate how to solve \eqref{eq:EM optimization}.

Recall that we have
\begin{equation}\label{eq:likelihood function}
    p(\mathbf{x},\mathbf{y}|\vctheta,\mathbf{u}) = p(\mathbf{y}|\mathbf{x},\vctheta)p(\mathbf{x}|\vctheta,\mathbf{u}).
\end{equation}
Using Markov properties \cite{ross2014introduction}, $p(\mathbf{y}|\mathbf{x},\vctheta)$ 
can be decomposed as
\begin{equation}
    p(\mathbf{y}|\mathbf{x},\vctheta) 
    = 
    \prod_{t=0}^{\nD-1} p(\vcy_t|\mathbf{x},\vctheta)
    = 
    \prod_{t=0}^{\nD-1} p(\vcy_t|\vcx_t,\vctheta).
\end{equation}
Furthermore, for $p(\mathbf{x}|\vctheta,\mathbf{u})$, we have
\begin{equation}
    p(\mathbf{x}|\vctheta,\mathbf{u}) 
    = 
    \prod_{t=0}^{\nsmD-1} p(\vcx_t|\vctheta,\mathbf{u})
    = 
    p(\vcx_0|\vctheta,\mathbf{u}) \prod_{t=0}^{\nD-1} p(\vcx_{t+1}|\vcx_{t},\vctheta).
\end{equation}
Thus, the log function of \eqref{eq:likelihood function} can be derived as 
\begin{equation}\label{eq:log-likelihood function}
    \begin{split}
        \log p(\mathbf{x},\mathbf{y}|\vctheta,\mathbf{u}) = \log p(\vcx_0|\vctheta,\mathbf{u}) + \sum_{t=0}^{\nD-1}\log p(\vcy_t|\vcx_t,\vctheta) + \sum_{t=0}^{\nD-1}\log p(\vcx_{t+1}|\vcx_{t},\vctheta).
    \end{split}
\end{equation} 
Following the same line of arguments as in Section~\ref{sec:ML}, it can be seen that $\vcy_t|\vcx_t,\vctheta$ and $\vcx_{t+1}|\vcx_t,\vctheta$ follow Gaussian distributions, given respectively by $\mathcal{N}(\Xi_t(\vctheta) \vcx_t + \mxD \vcu_t, \mxS_\vcv)$ and $\mathcal{N}(\mxA \vcx_t + \mxB \vcu_t, \mxS_\vcw)$, for $t=0,\ldots,\nD-1$, where $\Xi_t(\vctheta)$ is defined in \eqref{eqn:Xi_t}. 
For $t=0,\ldots,\nD-1$, 
it is worth noting that $\Xi_t$ is indeed function of matrix $\mxC$ defined as $\mxC := [\mxC_0,\, \mxC_1,\, ... \,,\,{\mxC}_{\Nu}]$, and thus, one can equivalently write $\Xi_t(\mxC)$ instead of $\Xi_t(\vctheta)$.
Accordingly, from \eqref{eqn:Q_theta_theta_k} and \eqref{eq:log-likelihood function}, it is implied that
\begin{equation}\label{eq:Q function}
    \begin{split}\!\!\!
        Q(\vctheta|\hat{\vctheta}_k) = &\mathbb{E}_{p(\mathbf{x}|\mathbf{y},\hat{\vctheta}_k,\mathbf{u})}
        \Big[
        -\frac{1}{2}   \logdet(\mxS_{\vcx_0}) 
        -\frac{\nD}{2} \logdet(\mxS_\vcv) 
        -\frac{\nD}{2} \logdet(\mxS_\vcw) 
        -\frac{1}{2}(\vcx_0-\vcmu_{\vcx_0})^\tr \mxS_{\vcx_0}^{-1} (\vcx_0-\vcmu_{\vcx_0})
        \\
        & \!\!\!\! \!\!\!\! 
        -\frac{1}{2} \sum_{t=0}^{\nD-1} (\vcy_t - \Xi_t(\mxC) \vcx_t - \mxD \vcu_t)^\tr \mxS_\vcv^{-1} (\vcy_t - \Xi_t(\mxC) \vcx_t - \mxD \vcu_t) 
         \\
        & \!\!\!\! \!\!\!\! 
        -\frac{1}{2} \sum_{t=0}^{\nD-1} (\vcx_{t+1} - \mxA \vcx_t - \mxB \vcu_t)^\tr \mxS_\vcw^{-1} (\vcx_{t+1} - \mxA \vcx_t - \mxB \vcu_t) 
        \Big].
    \end{split}\!\!\!\!
\end{equation}
Note that, for any matrix $\mxR\in\Rbb^{n\times n}$ and any vector $\vcx\in\Rbb^n$, we have $\vcx^\tr\mxR\vcx = \trace(\mxR\vcx\vcx^\tr)$. 
From this property, it follows for $Q(\vctheta|\hat{\vctheta}_k)$ that
\begin{equation}\label{eq:Q function2}
    \begin{split}\!\!\!\!\!\!\!\!
        Q(\vctheta|\hat{\vctheta}_k) 
        = & \,\,
        \mathbb{E}_{p(\mathbf{x}|\mathbf{y},\hat{\vctheta}_k,\mathbf{u})}
        \Big[
        -\frac{\nD}{2} \logdet(\mxS_\vcv) 
        -\frac{1}{2}   \logdet(\mxS_{\vcx_0})
        -\frac{\nD}{2} \logdet(\mxS_\vcw) 
        -\frac{1}{2}\trace\left(\mxS_{\vcx_0}^{-1}\, (\vcx_0-\vcmu_{\vcx_0}) (\vcx_0-\vcmu_{\vcx_0})^\tr\right)
        \\ &
        \qquad \qquad \qquad \qquad \qquad 
        -\frac{1}{2} \sum_{t=0}^{\nD-1} \trace\left(\mxS_\vcv^{-1} (\vcy_t - \Xi_t(\mxC) \vcx_t - \mxD \vcu_t) (\vcy_t - \Xi_t(\mxC) \vcx_t - \mxD \vcu_t)^\tr\right) 
        \\ &
        \qquad \qquad \qquad \qquad \qquad 
        -\frac{1}{2} \sum_{t=0}^{\nD-1}\trace\left(\mxS_\vcw^{-1} (\vcx_{t+1} - \mxA \vcx_t - \mxB \vcu_t) (\vcx_{t+1} - \mxA \vcx_t - \mxB \vcu_t)^\tr\right)
        \Big].
    \end{split}\!\!\!\!\!\!
\end{equation}
To simplify the notation, let the matrices $\mxM$ and $\mxN$ be defined as $\mxM := [\mxA,\, \mxB]$ and $\mxN := [\mxC, \, \mxD]$, respectively. 
Furthermore, we define matrix-valued functions 
$\mxF_k:\Rbb^{\Nx}\to\Sbb^{\Nx\times\Nx}$,
$\mxG_k: \Rbb^{\Ny \times (\Nu + \Nx + \Nx\Nu)} \to\Sbb^{\Ny\times\Ny}$, and
$\mxH_k: \Rbb^{\Nx \times (\Nx + \Nu)} \to\Sbb^{\Nx\times\Nx}$ respectively as
\begin{align}
        \mxF_k(\mu_{\vcx_0})\!\, &:= \,
        \mathbb{E}_{p(\mathbf{x}|\mathbf{y},\hat{\vctheta}_k,\mathbf{u})} \left[(\vcx_0-\vcmu_{\vcx_0}) (\vcx_0-\vcmu_{\vcx_0})^\tr\right], \label{eq:F_k}\\
        \mxG_k(\mxN)\, &:= \,
        \mathbb{E}_{p(\mathbf{x}|\mathbf{y},\hat{\vctheta}_k,\mathbf{u})}
        \left[
        \sum_{t=0}^{\nD-1}  (\vcy_t - \Xi_t(\mxC) \vcx_t - \mxD \vcu_t) (\vcy_t - \Xi_t(\mxC) \vcx_t - \mxD \vcu_t)^\tr\right],\label{eq:G_k}
\end{align}
and, 
\begin{align}
    \mxH_k(\mxM) &:= \mathbb{E}_{p(\mathbf{x}|\mathbf{y},\hat{\vctheta}_k,\mathbf{u})}
    \left[
    \sum_{t=0}^{\nD-1}(\vcx_{t+1} - \mxA \vcx_t - \mxB \vcu_t) (\vcx_{t+1} - \mxA \vcx_t - \mxB \vcu_t)^\tr\right],\label{eq:H_k}
\end{align}
for any 
$\mu_{\vcx_0} \in \Rbb^{\Nx}$, 
$\mxM \in \Rbb^{\Nx \times (\Nx + \Nu)}$, and $\mxN \in \Rbb^{\Ny \times (\Nu + \Nx + \Nx\Nu)}$.
Additionally, for $k \in \Nbb$, let the function $J_k: \Vbb \to \extendedR$ be defined as 
\begin{equation}\label{eq:J function}
    J_k(\vctheta) := 
        \frac{\nD}{2} \logdet(\mxS_\vcv) 
        +\frac{1}{2}   \logdet(\mxS_{\vcx_0})
        +\frac{\nD}{2} \logdet(\mxS_\vcw) 
        +\frac{1}{2} \trace \left(\mxS_{\vcx_0}^{-1}\, \mxF_k(\mu_{\vcx_0}) \right)
        +\frac{1}{2} \trace \left(\mxS_\vcv^{-1}\, \mxG_k(\mxN)\right) 
        +\frac{1}{2} \trace \left(\mxS_\vcw^{-1}\, \mxH_k(\mxM)\right),
\end{equation}
for any $\theta\in\Vbb$.
According to the linearity of trace and expectation, it follows from \eqref{eq:Q function2} that
$J_k(\vctheta) := -Q(\vctheta|\hat{\vctheta}_k)$,
for any $\theta\in\Theta$.
Thus, one can see that the EM iterative generation of the sequence of parameter estimation, as formulated through the optimization problem \eqref{eq:EM optimization}, is equivalent to
\begin{equation}\label{eq:EM optimization J_k}
\begin{split}
    \hat{\vctheta}_{k+1}
    &=
    \argminOp_{\vctheta \in \Theta}\,
    J_k(\vctheta),\\
\end{split}
\end{equation}
for $k \in \Nbb$. 
To proceed further and elaborate on the properties of the introduced functions and the optimization problem \eqref{eq:EM optimization J_k}, suitable assumptions are required.

\begin{assumption}\label{assum:input}
For the input sequence     
$\vcu_0,\vcu_1,\ldots,\vcu_{\nsmD-1}$, we have 
\begin{equation} \label{eqn:rank_1_u_i_s}
    \rank\left(
    \begin{bmatrix} 
    1       & 1      & \cdots & 1 \\
    \vcu_0  & \vcu_1 & \cdots & \vcu_{\nsmD-1}
    \end{bmatrix}
    \right) = \Nu+1.
\end{equation}
\end{assumption}

\begin{remark}
    Assumption~\ref{assum:input} is equivalent to the property that the variation of input sequence spans $\Rbb^{\Nu}$. More precisely, \eqref{eqn:rank_1_u_i_s} holds if and only if
    \begin{equation}
        \Rbb^{\Nu}
        =
        \linspan
        \big\{
        \vcu_1 - \vcu_0,
        \vcu_2 - \vcu_1, 
        \ldots,
        \vcu_{\nsmD-1} - \vcu_{\nsmD-2}
        \big\}.
    \end{equation}
    Accordingly, Assumption~\ref{assum:input} can be interpreted as a weak form of persistent excitation for the input sequence 
    $\vcu_0\!,\!\ldots\!,\!\vcu_{\nsmD-1}$.
\end{remark}



\begin{assumption}\label{assum:output}
For the output measurements    
$\vcy_0,\vcy_1,\ldots,\vcy_{\nsmD-1}$, 
we have 
\begin{equation} \label{eqn:rank_y_i_s}
    \rank\left(
    \begin{bmatrix} 
    \vcu_0  & \vcu_1 & \cdots & \vcu_{\nsmD-1}\\
    \vcy_0  & \vcy_1 & \cdots & \vcy_{\nsmD-1}
    \end{bmatrix}
    \right) = \Nu + \Ny.
\end{equation}
\end{assumption}

\begin{remark}
Assumption~\ref{assum:output} is equivalent to the condition that $[\vcu_0^\tr, \vcy_0^\tr]^\tr, [\vcu_1^\tr, \vcy_1^\tr]^\tr, \ldots, [\vcu_\nsmD^\tr, \vcy_\nsmD^\tr]^\tr$ span $\Rbb^{\Nu+\Ny}$.
\end{remark}

\begin{proposition}\label{pro: PD for expectation}
    Under Assumption~\ref{assum:input} and \ref{assum:output}, the matrices $\mxF_k(\mu_{\vcx_0})$, $\mxG_k(\mxN)$, and $\mxH_k(\mxM)$ are positive definite for any $\mu_{\vcx_0} \in \Rbb^{\Nx}, \mxM \in \Rbb^{\Nx \times (\Nx + \Nu)}$, and $\mxN \in \Rbb^{\Ny \times (\Nu + \Nx + \Nx\Nu)}$.
\end{proposition}
\begin{proof}
    The proof is provided in Appendix~\ref{app:proof PD of Expectation}.
\end{proof}
Through the next propositions, we show that \eqref{eq:EM optimization J_k} is well-defined by demonstrating the existence and uniqueness of its solution. 
\begin{proposition}[Stationary Point Uniqueness]\label{pro:local minimum}
Let Assumption~\ref{assum:input} hold. Then, $J_k$ has a \emph{unique stationary point} in $\Theta$, meaning the system of equations $\nabla J_k(\vctheta) = 0$ admits a unique solution in $\Theta$, which can be expressed in closed form as 
\begin{align} 
        \hat{\mxM}_{k+1} =&  \sum_{t=0}^{\nD-1} \begin{bmatrix} \hat{\vcx}_{t+1|{\nD}}\hat{\vcx}_{t|{\nD}}^\tr + \mxP_{t+1|{\nD}}\mxL_t^\tr & \quad \hat{\vcx}_{t+1|\nD}\vcu_t^\tr  \end{bmatrix} 
        \left(\sum_{t=0}^{\nD-1} \begin{bmatrix}\hat{\vcx}_{t|\nD}\hat{\vcx}_{t|\nD}^\tr + \mxP_{t|\nD}  &\hat{\vcx}_{t|\nD}\vcu_t^\tr 
        \\ \vcu_t\hat{\vcx}_{t|\nD}^\tr &\vcu_t \vcu_t^\tr \end{bmatrix}\right)^{-1}, 
        \label{eq: closed form M}
        \\
        \hat{\mxN}_{k+1}   =&\,\, \sum_{t=0}^{\nD-1} \vcy_t \begin{bmatrix} 
        \hat{\vcx}_{t|\nD}\\ \vcu_t \otimes \hat{\vcx}_{t|\nD} \\ \vcu_t\end{bmatrix}^\tr
        \left(\sum_{t=0}^{\nD-1}
        \begin{bmatrix} 
        \hat{\vcx}_{t|\nD} \hat{\vcx}_{t|\nD}^\tr + \mxP_{t|\nD} 
        & \vcu_t^\tr \otimes \left(\hat{\vcx}_{t|\nD} \hat{\vcx}_{t|\nD}^\tr + \mxP_{t|\nD} \right) 
        & \hat{\vcx}_{t|\nD} \vcu_t^\tr \\ 
        \vcu_t \otimes \left(\hat{\vcx}_{t|\nD} \hat{\vcx}_{t|\nD}^\tr +    \mxP_{t|\nD} \right)
        & \vcu_t \vcu_t^\tr \otimes \left(\hat{\vcx}_{t|\nD} \hat{\vcx}_{t|\nD}^\tr     + \mxP_{t|\nD} \right)
        & \vcu_t\vcu_t^\tr \otimes \hat{\vcx}_{t|\nD}\\
        \vcu_t \hat{\vcx}_{t|\nD}^\tr 
        & \vcu_t \vcu_t^\tr \otimes \hat{\vcx}_{t|\nD}^\tr
        & \vcu_t \vcu_t^\tr 
        \end{bmatrix}\right)^{-1}\!\!\!\!\!\!, 
        \label{eq: closed form C}
        \\
        \hat{\mxS}_{\vcw,k+1}   =&\,\, \frac{1}{\nD}\sum_{t=0}^{\nD-1} (\hat{\vcx}_{t+1|\nD}\hat{\vcx}_{t+1|\nD}^\tr + \mxP_{t+1|\nD}) - \frac{1}{\nD}\sum_{t=0}^{\nD-1} \begin{bmatrix}\hat{\vcx}_{t+1|\nD}\hat{\vcx}_{t|\nD}^\tr + \mxP_{t+1|\nD}\mxL_t^\tr &\,\, \hat{\vcx}_{t+1|\nD} \vcu_t^\tr  \end{bmatrix} 
        \nonumber
        \\
        &
        \qquad
        \left(\sum_{t=0}^{\nD-1} \begin{bmatrix}\hat{\vcx}_{t|\nD}\hat{\vcx}_{t|\nD}^\tr + \mxP_{t|\nD} & \hat{\vcx}_{t|\nD}\vcu_t^\tr\\ \vcu_t\hat{\vcx}_{t|\nD}^\tr &\vcu_t \vcu_t^\tr\end{bmatrix}\right)^{-1}
        \sum_{t=0}^{\nD-1}\begin{bmatrix}\hat{\vcx}_{t+1|\nD}\hat{\vcx}_{t|\nD}^\tr + \mxP_{t+1|\nD}\mxH_t^\tr &\, \,\hat{\vcx}_{t+1|\nD} \vcu_t^\tr  \end{bmatrix}^\tr, 
        \label{eq: closed form Sw}
        \\
        \hat{\mxS}_{\vcv,k+1}   =&\,\, \frac{1}{\nD}\sum_{t=0}^{\nD-1} \vcy_t  \vcy_t^\tr - \frac{1}{\nD}  \sum_{t=0}^{\nD-1} \vcy_t\begin{bmatrix} 
        \hat{\vcx}_{t|\nD}\\ \vcu_t \otimes \hat{\vcx}_{t|\nD} \\ \vcu_t\end{bmatrix}^\tr 
        \nonumber
        \\
        &\left(\sum_{t=0}^{\nD-1}
        \begin{bmatrix} 
        \hat{\vcx}_{t|\nD} \hat{\vcx}_{t|\nD}^\tr + \mxP_{t|\nD} 
        & \vcu_t^\tr \otimes \left(\hat{\vcx}_{t|\nD} \hat{\vcx}_{t|\nD}^\tr + \mxP_{t|\nD} \right) 
        & \hat{\vcx}_{t|\nD} \vcu_t^\tr \\ 
        \vcu_t \otimes \left(\hat{\vcx}_{t|\nD} \hat{\vcx}_{t|\nD}^\tr +    \mxP_{t|\nD} \right)
        & \vcu_t \vcu_t^\tr \otimes \left(\hat{\vcx}_{t|\nD} \hat{\vcx}_{t|\nD}^\tr     + \mxP_{t|\nD} \right)
        & \vcu_t\vcu_t^\tr \otimes \hat{\vcx}_{t|\nD}\\
        \vcu_t \hat{\vcx}_{t|\nD}^\tr 
        & \vcu_t \vcu_t^\tr \otimes \hat{\vcx}_{t|\nD}^\tr
        & \vcu_t \vcu_t^\tr 
        \end{bmatrix}\right)^{-1}
        \sum_{t=0}^{\nD-1} \begin{bmatrix} 
        \hat{\vcx}_{t|\nD}\\ \vcu_t \otimes \hat{\vcx}_{t|\nD} \\ \vcu_t\end{bmatrix} \vcy_t^\tr, 
        \label{eq: closed form Sv}\\
        \hat{\vcmu}_{x_0,k+1}   =&\,\, \hat{\vcx}_{0|\nD}, 
        \label{eq: closed form mux0}\\
        \hat{\mxS}_{\vcx_0,k+1}   =&\,\, \mxP_{0|\nD}. \label{eq: closed form Sx0} 
\end{align}
\end{proposition}
\begin{proof}
    The proof is provided in Appendix~\ref{app:Stationary Point Uniqueness}.
\end{proof}



\begin{corollary}\label{cor: invex function}
     The function $J_k$ is invex, that is, $J_k$ has exactly one local minima, which is the unique optimizer of $J_k$.
\end{corollary}
\begin{proof}
    The proof is provided in Appendix~\ref{app:Proof of Corollary 1}.
\end{proof}

\begin{proposition}[Recursive Feasibility]\label{pro: postive definite}
Let Assumption~\ref{assum:input} and Assumption~\ref{assum:output} hold, and suppose $\hat{\mxS}_{\vcw,0} \in \Sbb_{++}^{\Nx\times \Nx} $, $\hat{\mxS}_{\vcv,0} \in \Sbb_{++}^{\Ny\times \Ny}$, and $\hat{\mxS}_{\vcx_0,0} \in \Sbb_{++}^{\Nx\times \Nx}$. Then, \eqref{eq: closed form M}-\eqref{eq: closed form Sx0} provides a feasible solution for the optimization problem \eqref{eq:EM optimization J_k}, meaning it lies within $\Theta$. This result ensures that the EM procedure admits the \emph{recursive feasibility} property.
\end{proposition}
\begin{proof}
    The proof is provided in Appendix~\ref{app:recursive feaisbility}.
\end{proof}


The Algorithm~\ref{Al:EM ID} summarizes the EM approach introduced for the identification of linear dynamics with bilinear observation models.
\begin{algorithm}[t]
    \caption{EM Estimation Approach for Identification of Linear Dynamics with Bilinear Observation Models}\label{Al:EM ID}
    \begin{algorithmic}
        \Statex \textbf{Input:} $\Dcal$.
        \Statex \textbf{Output:} $\vctheta$.
        \State \textbf{Initial guess:} $\hat{\vctheta}_0$
        \State $k \gets 0$
        \While{$\mathbf{1}$}
            \State Current parameters estimates: $\hat{\vctheta} \gets \hat{\vctheta}_k$
            \For{$t \gets 0$ to $\nD$}
               \State Kalman Filter: \eqref{eq:kalman filter}
            \EndFor
            \For{$t \gets \nD$ to $0$}
               \State RTS smoother: \eqref{eq:mean and covariance}
            \EndFor
            \State EM approach: compute \eqref{eq: closed form M}-\eqref{eq: closed form Sx0} to find a new parameters estimates. $\hat{\vctheta}_{k+1} \gets \argminOp_{\vctheta \in \Theta} J_k(\vctheta)$.
            \If{$\lVert \hat{\vctheta}_{k+1} - \hat{\vctheta}_k \rVert< \epsilon$}
            \State break
            \Else
            \State $k \gets k+1$
            \EndIf
        \EndWhile
        \State $\vctheta \gets \hat{\vctheta}_k$
    \end{algorithmic}
    \label{alg_1}
\end{algorithm}


\section{Numerical Experiments} \label{sec:numerical experiments}
To evaluate the proposed algorithms and verify their effectiveness, we perform extensive numerical experiments, including Monte Carlo simulations under different noise levels and dataset sizes. Accordingly, we employ suitable evaluation metrics, providing quantitative measures of the accuracy of the identified parameters relative to the ground truth. Furthermore, cross-validation is used to ensure the reliability of the results and enable a proper and unbiased evaluation of the proposed approaches. In addition to accuracy, the computational complexity of the methods is analyzed, focusing on their scalability and runtime behavior as the size of the dataset increases. The convergence performance of the EM approach is also empirically demonstrated to evaluate the efficiency of the algorithm further. The results show the efficacy of the proposed methods, their practical effectiveness in estimating system parameters, and their potential usage in scenarios involving bilinear observational structures. 
%

\textbf{Example 1.} 
In this example, we evaluate our proposed algorithms considering a time-invariant system as introduced in \eqref{eqn:dynamics1}-\eqref{eqn:dynamics2} with matrices $\mxA$, $\mxB$, $\mxC_0$, $\mxC_1$ and $\mxD$ defined as 
\begin{subequations}\label{eq:system example1_1}
\begin{align}
        \mxA &= \begin{bmatrix}
            0.6 & -0.28 \\ 0.25 & 0.45
        \end{bmatrix},\\
        \mxB &= \begin{bmatrix}
            0.5 \\ -0.5
        \end{bmatrix},\\
        \mxC_0 &= \begin{bmatrix}
            0.5 & -0.15
        \end{bmatrix},\\
        \mxC_1 &= \begin{bmatrix}
            0.15 & 0.1
        \end{bmatrix},\\
        \mxD &= 0.
\end{align}
\end{subequations}
The initial state is set to $\mu_{\mathbf{x}_0} = \begin{bmatrix}
    1 & 1 
\end{bmatrix}^\tr$, and realizations of random binary signals are generated as the control input sequence under different dataset length. 
Following the problem setting introduced in Section \ref{sec:pf}, random sequences of process noise and measurement noise are sampled from Gaussian distributions $\Ncal(\mathbf{0}, \mxS_\vcw)$ and $\Ncal(\mathbf{0}, \mxS_\vcv)$, respectively, where $\mxS_\vcw$ and $\mxS_\vcv$ are designed for different SNR levels. 

To evaluate the performance and robustness of the two proposed algorithms, we perform a Monte Carlo experiment. For each dataset length, given the same input sequences, $N_{\text{MC}} = 100$ different initializations and noise sequences are generated under four different SNR levels, namely $5\,$dB, $10\,$dB, $15\,$dB, and  $20\,$dB.  For each realization, the proposed Algorithm \ref{Al:ML ID} and Algorithm \ref{Al:EM ID} are applied to identify the system. After identification, another random input sequence with length $T = 100$ is generated for validation. The performance is evaluated based on the normalized relative error of output as following,
\begin{equation}\label{eq:Normalized relative error}
    \vcy_{\text{error}} = \sum_{t=0}^{T} \frac{\lVert \vcy_t - \hat{\vcy}_t \rVert} {\lVert \vcy_t \rVert},  
\end{equation}
where $\vcy_t$ and $\hat{\vcy}_t$ are the output trajectories computed using the real system and the identified system, respectively. 

\begin{figure}[t!]
   \centering
   \includegraphics[width=0.7\textwidth,,trim={2cm 6cm 2cm 9.5cm},clip]{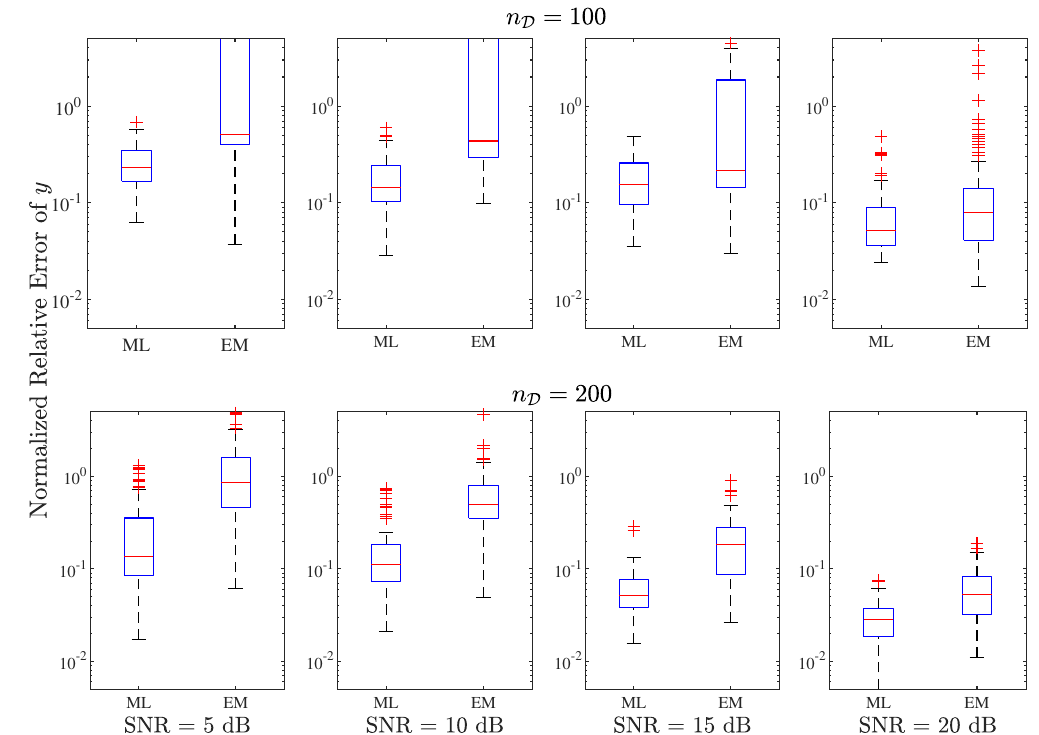}
   \caption{Comparison of normalized relative error of output trajectory of systems identified by ML and EM approaches.}
   \label{img:MC_comparison}
\end{figure}

\begin{figure}[t!]
   \centering
   \includegraphics[width=0.7\textwidth,trim={2.5cm 10cm 3.1cm 10.5cm},clip]{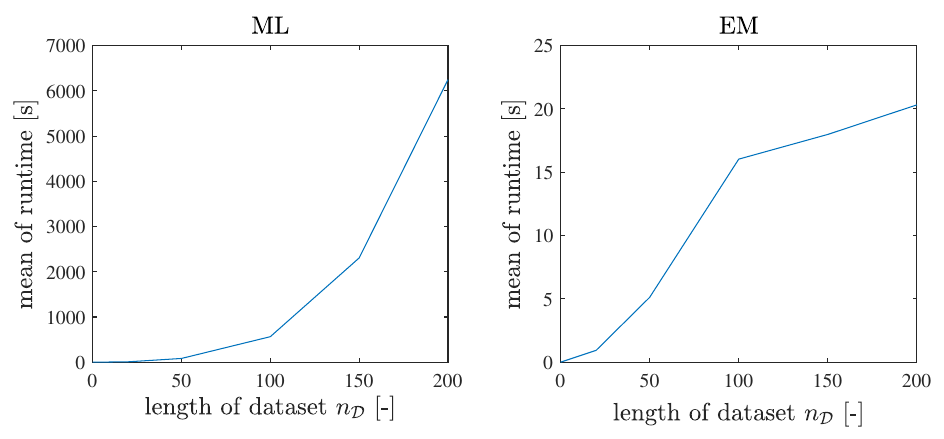}
   \caption{Runtime of ML and EM approach with different length of dataset.}
   \label{img:ML_com}
\end{figure}

The results are shown in Figure~\ref{img:MC_comparison}, where  data sets of varying sizes ($\nD=
100,200$) are used for both approaches. When comparing the results within each row, we can see that increasing the SNR levels leads to improved performance for both algorithms. Similarly, when comparing the results in each column, increasing the dataset length also enhances the performance of both algorithms. Moreover, when comparing the two algorithms in each subplot, the ML approach outperforms the EM approach for the same data set length and SNR level. However, this improved performance comes at the cost of increased runtime. As shown in Figure~\ref{img:ML_com}, the runtime of the ML approach increases significantly with increasing data set length. In contrast, the runtime of the EM approach remains around 20 seconds for a data set length of 200, which is considerably lower than that of the ML approach. As discussed in Section~\ref{sec:Complexity Analysis}, this difference arises due to the high computational demand to calculate derivatives in the ML approach. Consequently, although ML performs better for shorter data sets, solving the optimization problem for longer data sets becomes challenging, where the EM approach proves to be more effective. In Figure~\ref{img:EM_1000}, the EM approach is applied to a data set of 1000 length, achieving a run-time of approximately 34 seconds and showing better performance than the ML approach with a data set length of 200. 

In conclusion, the ML approach outperforms the EM approach for small data sets. The performance gap becomes smaller for longer length of the data set and higher SNR level. Furthermore, as the size of the data set increases, the ML approach becomes computationally challenging, whereas the EM approach is more efficient and better suited for large data sets.


To further investigate the performance of the Algorithm~\ref{Al:EM ID}, we analyze the relative error 
\begin{equation}\label{eq: relative error parameters}
\begin{split}
    \mxC_{\text{error}} =&  \frac{\lVert \mxC - \hat{\mxC} \rVert} {\lVert \mxC \rVert},\\
    \mxM_{\text{error}} =&  \frac{\lVert \mxM - \hat{\mxM} \rVert} {\lVert \mxM \rVert},
\end{split}
\end{equation}
of the estimated system matrices, where $\mxC = \begin{bmatrix}\mxC_0 & \mxC_1 \end{bmatrix}$ and $\mxM = \begin{bmatrix}\mxA &\mxB \end{bmatrix}$ as defined after \eqref{eq:Q function2}.
\begin{figure}[t!]
   \centering
   \includegraphics[width=0.5\textwidth,trim={1cm 10.5cm 10cm 10.5cm},clip]{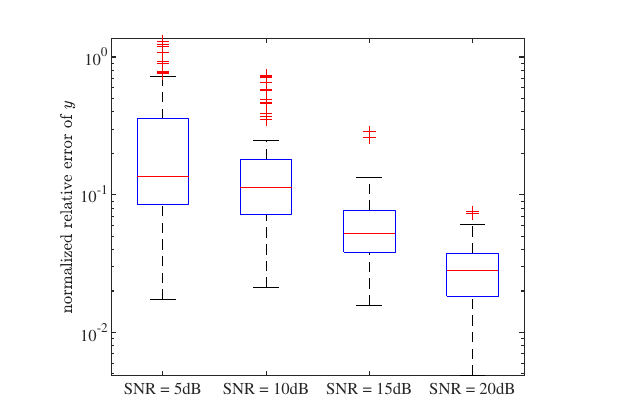}
   \caption{Normalized relative error of output trajectory of systems identified by EM approach with 1000 length of data set.}
   \label{img:EM_1000}
\end{figure}

\begin{figure}[t!]
   \centering
   \includegraphics[width=0.8\textwidth,trim={1cm 10.5cm 1cm 10.6cm},clip]{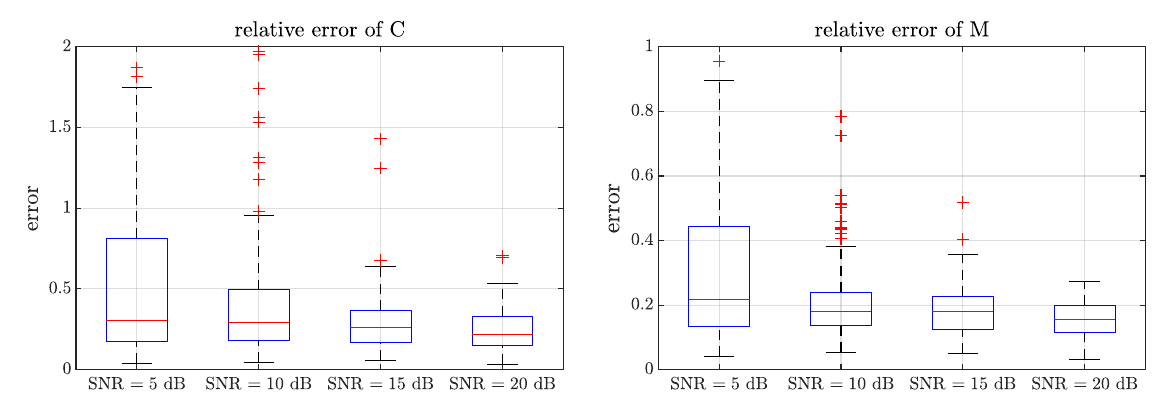}
   \caption{Relative error of system matrices under four different SNRs using EM approach with $\nD=1000$.}
   \label{img:MonteCarlo}
\end{figure}

Figure~\ref{img:MonteCarlo} shows the results when we apply Algorithm~\ref{Al:EM ID} and set the length of the data set $\nD=1000$. Generally, the mean relative errors of both $\mxC$ and $\mxM$ are lower for scenarios of higher SNR levels, indicating improved precision in the identified parameters. At lower SNR levels, the performance is more sensitive to initialization. Specifically, while some realizations achieve low relative errors, others exhibit significantly higher errors, resulting in increased variance. However, for larger SNR cases, the variance of the relative errors is relatively smaller, indicating that the algorithm is less dependent on initialization. 

\begin{figure}[t!]
   \centering
   \includegraphics[width=0.9\textwidth,trim={1cm 10cm 1cm 10.8cm},clip]{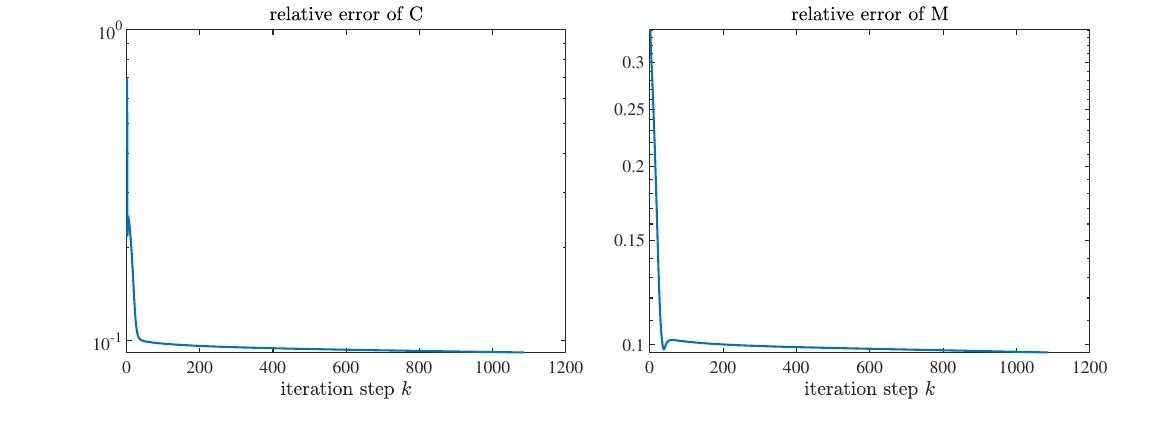}
   \caption{The relative error of system matrices estimates with respect to iteration steps.}
   \label{img:convergence}
\end{figure}

Figure~\ref{img:convergence} illustrates a single test of Algorithm~\ref{Al:EM ID}. In this figure, we can  observe the convergence of the relative errors $\mxC_\text{error}$ and $\mxM_\text{error}$. It can be seen that, initially, the relative error is large because of the difference between the initialized parameters and the real parameters. However, as the number of iterations increases, the relative error decreases. The convergence rate is fast in the initial iterations; however, it gets slower later. For $\epsilon = 10^{-5}$, the algorithm stops after about $1100$ EM steps. Figure~\ref{img:validation} shows the performance of the proposed Algorithm~\ref{Al:EM ID}. To this end, another random input sequence of length $T = 80$ is generated to validate the identified system. We can note that the predicted trajectory of the outputs for the identified system is close to that of the real system. The relatively error $\vcy_{\text{error}}$, computed as in \eqref{eq:Normalized relative error} is approximately $10.3\%$. 

\begin{figure}[t!]
   \centering
   \includegraphics[width=0.7\textwidth,trim={5cm 10.5cm 5cm 10.5cm},clip]{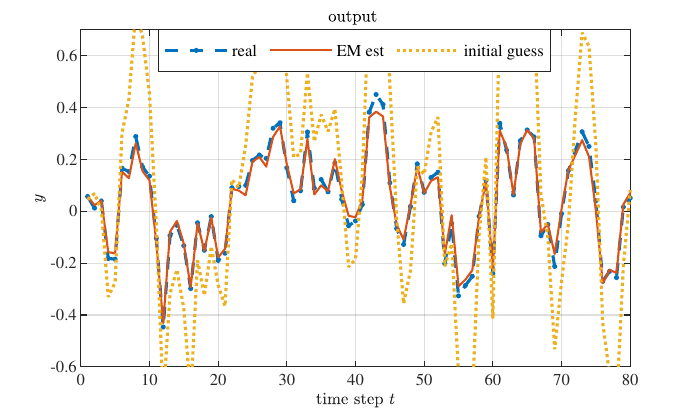}
   \caption{Comparison of real system outputs and identified system outputs. The blue line represents the output trajectory of the real system, the red line corresponds to the trajectory obtained from the identified system, and the yellow line depicts the trajectory generated using the system with initial guess parameters. }
   \label{img:validation}
\end{figure}

\begin{remark}
    As shown in Figure~\ref{img:convergence}, the convergence speed of the EM algorithm is relatively slow. As discussed after \eqref{eq:Q convengence}, with each iteration step, improving $Q(\vctheta | \hat{\vctheta}_k)$ will improve the log-likelihood $\log p(\mathbf{\mathbf{y}}|\vctheta,\mathbf{u})$. However, the rate of improvement can be slow. Furthermore, there is no guarantee that the log-likelihood will converge to the global optimum. More precisely, the EM algorithm may converge to a local optimum. Consequently, in simulations, it may be necessary to restart the algorithm with different initialization to obtain more accurate parameter estimates.
\end{remark}

\textbf{Example 2.} In this example, we consider a resistor–capacitor circuit (RC circuit), where the capacitor is non-ideal as shown in Figure~\ref{fig:Parallel RC}a). The non-ideal capacitor is modeled as in Figure~\ref{fig:Parallel RC}b) \cite{bisquert2000role}. In this model, $R_s$ represents the equivalent series resistance (ESR), which captures the resistive losses due to the internal structure. The parasite inductance, denoted as $L$, accounts for the inductive effects caused by the leads or internal connections. Leakage resistance $R_p$ models the imperfection of the dielectric material, which causes current leakage through the capacitor. 

\begin{figure}[t]
    \begin{circuitikz}
    \node at (-0.5,2.3) {\textbf{a)}};
    \draw (1.5,2)
    to[short] (3,2)
    to[C=$C_0$] (3,0) 
    to[short] (0,0);

    \draw (0,0)
    to[V,v=$V$] (0,2) 
    to[short, i=$I$] (1.5,2)
    to[R=$R_0$] (1.5,0)
    to[short] (0,0);
    
    \node at (6.5,2.3) {\textbf{b)}};
    
    \draw (7,0)
    to[V,v=$V$] (7,2) 
    to[short,  i=$I$] (8.5,2)
    to[R=$R_0$] (8.5,0)
    to[short] (7,0);

    \draw (8.5,2)
    to[short, i=$I_L$](9.5, 2)
    to[R=$R_s$] (10.5,2) 
    to[L=$L$] (12.5,2)
    to[C=$C$, i=$I_c$] (12.5,0)
    to[short] (7,0);

    \draw (12.5,2)
    to[short] (14,2)
    to[R=$R_p$] (14,0)
    to[short] (7,0);

    \node at (15,2) {$+$};
    \node at (15,0) {$-$};
    \node at (15,1) {$V_c$};
    
\end{circuitikz}
    \caption{a) Parallel RC circuit with a non-ideal capacitor and b) Parallel RC circuit with equivalent circuit of a real capacitor.}
    \label{fig:Parallel RC}
\end{figure}
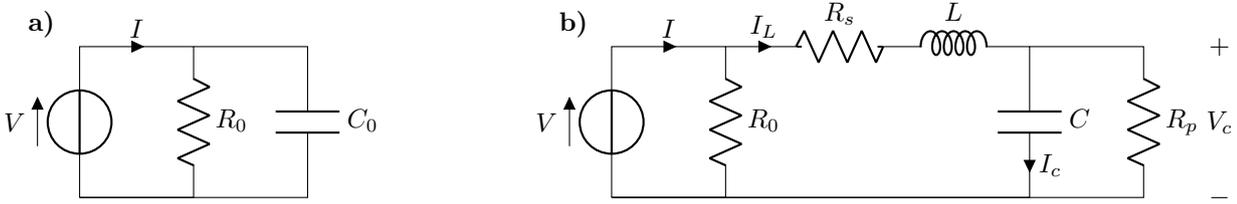

As shown in Figure~\ref{fig:Parallel RC}b), let $I_L$ denotes the current passing through the inductance $L$,  $I_c$ the current passing capacitor $C$, and $V_c$ the voltage across the capacitor. The above RC circuit can be seen as a simplified circuit of a load. To protect the capacitor, limitations are imposed on the thermal effects of the non-ideal capacitor, which can be modeled as $Q = \alpha I_L V$, where $\alpha$ is a heat transfer coefficient and $Q$ is proportional to the temperature and can be measured through a temperature sensor. Define 
\begin{equation}
\begin{split}
    \vcx &= \begin{bmatrix}
    I_L \\ V_c \end{bmatrix},\\
    \vcu  &= V, \\
    \vcy &= Q.\\
\end{split}
\end{equation}
We have 
\begin{equation}\label{eq:RC circuit}
\begin{split}
    \dot{\vcx} &= \begin{bmatrix}
        -\frac{R_L}{L} & -\frac{1}{L} \\
        \frac{1}{C} & -\frac{1}{CR_c} \\
    \end{bmatrix} \vcx + \begin{bmatrix}
        \frac{1}{L} \\ 0
    \end{bmatrix} \vcu, \\
    \vcy &= \begin{bmatrix}
        \alpha & 0
    \end{bmatrix} \vcx \vcu.\\
\end{split}
\end{equation}
In this experiment, the data set $\Dcal$ of length $\nD = 1000$ is generated by an input sequence $\vcu(t) = 12\sin(\omega t)$, where the angular frequency $\omega = 2\pi\times10^{8} \text{rad/s}$ and the sample time is $1\times10^{-9} \text{s}$. The covariance matrices of process noise, measurement noise, and initial states are set as
\begin{subequations}\label{eq:RC covariance}
\begin{align} 
        \mxS_\vcw &= \begin{bmatrix}
            1\times10^{-3} & 0 \\ 0 & 1\times10^{-3}
        \end{bmatrix},
        \\
        \mxS_\vcv &= 1\times10^{-4}, \\
        \mxS_{\vcx_0} &= \begin{bmatrix}
            1\times10^{-3} & 0 \\ 0 & 1\times10^{-3}
        \end{bmatrix},
\end{align}
\end{subequations}
which results in an SNR level of $10\,\text{dB}$. We use the proposed Algorithm~\ref{Al:EM ID} to identify the above system \eqref{eq:RC circuit}. After identification, another input signal sequence $u(t) = 6\sin(\omega t) + 6\sin(0.3\omega t)$, with the same sample time is used for validation.

Figure~\ref{img:validation2} presents the validation results. The estimated thermal dissipation trajectory closely follows the actual trajectory, with some discrepancies observed near the peaks. The mean normalized relative error of the output, calculated using \eqref{eq:Normalized relative error}, is approximately 9.9\%.

\begin{figure}[t!]
   \centering
   \includegraphics[width=0.7\textwidth,trim={5cm 10cm 5cm 10cm},clip]{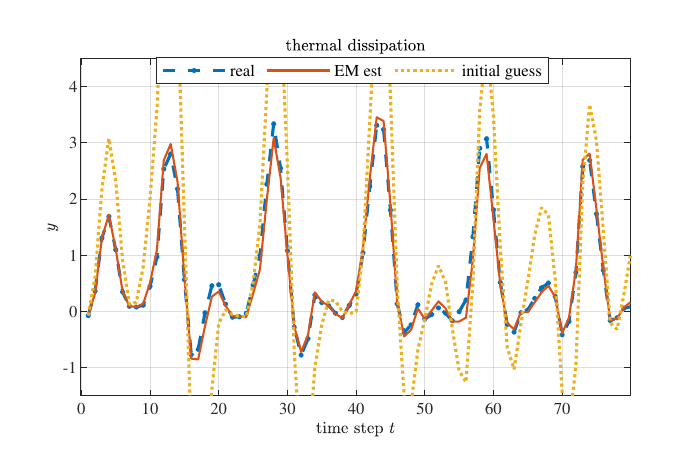}
   \caption{Closed-loop cross validation of thermal dissipation. The blue line represents the output trajectory of the real system, the red line corresponds to the trajectory obtained from the identified system, and the yellow line depicts the trajectory generated using the system with initial guess parameters.}
   \label{img:validation2}
\end{figure}
\section{Conclusion} \label{sec:conclusion}
In this paper, we have addressed the identification problem for unknown dynamical systems with linear dynamics and bilinear observation models. 
First, we have proposed a tractable identification procedure based on maximum likelihood (ML) to maximize the probability of observing the measurements. 
Given the considerable computational demand associated with the developed ML-based method, particularly for large measurement datasets, we have introduced an alternative approach combining the Rauch–Tung–Striebel (RTS) smoother with the Expectation-Maximization (EM) algorithm.
In the resulting framework, the RTS smoother estimates the system states, while the EM algorithm is used to estimate the unknown parameters. We have evaluated the performance of both approaches through numerical examples, comparing their effectiveness under various signal-to-noise ratios (SNRs) via onte Carlo numerical experiments. Furthermore, we applied the EM algorithm to a non-ideal capacitor example, demonstrating that the estimated output trajectory closely matches the real one. These results highlight the potential of the proposed methods for addressing the identification problem in practical applications.

\appendix



\label{sec:Appendix}
\section{Auxiliary Lemmas}\label{app:auxiliary lemma}
\begin{lemma}[Matrix Inversion Lemma, 
\cite{petersen2008matrix}]
\label{lem:matrix_inversion_lemma}
Given matrices $\mxA \in \Rbb^{n\times n}$, $\mxB \in \Rbb^{n\times m}$, $\mxC \in \Rbb^{m\times m}$ and $\mxD \in \Rbb^{m\times n}$. If $\mxA$ and $\mxC$ are invertible, then $\mxA+\mxB\mxC\mxD$ has an inverse as
\begin{equation}\label{eq:matrix inversion lemma}
    (\mxA+\mxB\mxC\mxD)^{-1} = \mxA^{-1} -\mxA^{-1}\mxB(\mxC^{-1} + \mxD\mxA^{-1}\mxB)^{-1}\mxD\mxA^{-1}.
\end{equation}
\end{lemma}
\begin{lemma}[Chain Rule -- Matrix Version, \cite{petersen2008matrix}] 
\label{lem:chain_rule}
Consider a matrix $\mxX \in \Rbb^{n\times m}$. Let $\mxU = f(\mxX)$,
where $f:\Rbb^{n\times m}\to\Rbb^{p \times q}$. Let $g:\Rbb^{p\times q}\to\Rbb$. Then, for the derivative of $g(\mxU)$ with respect to $\mxX$, we have
\begin{equation}\label{eq:chain rule}
    \pdv{g(\mxU)}{\mxX_{i,j}} = \trace\left[\left(\pdv{g(\mxU)}{\mxU}\right)^\tr \pdv{\mxU}{\mxX_{i,j}} \right],
\end{equation}
for any $i=1,\ldots,n$ and $j=1,\ldots,m$.
\end{lemma}
\begin{lemma}[Multidimensional Gaussian Integral Formula, \cite{petersen2008matrix}]  
\label{lem:Gaussian_integration}
For any vector $\mu\in \mathbb{R}^{n}$ and any positive definite matrix $\mxS \in \mathbb{R}^{n \times n}$, we have
\begin{equation}\label{eq:Gaussian integration}
    \int_{\mathbb{R}^{n}} 
    \exp\big(-\frac{1}{2} 
        (\vcx - \mu)^\tr\mxS^{-1} (\vcx - \mu) 
        \big) 
    \drm \vcx 
    = 
    (2\pi)^{\frac{1}{2}n}|\mxS|^{\frac{1}{2}}.
\end{equation}
\end{lemma}
\begin{lemma}\label{lem:invex}
    Define function $f:\Rbb_+\to \Rbb$ as    
    \begin{equation}
        f(x) := \log(x) + \frac{a}{x}, \qquad \forall\, x\in\Rbb_+,
    \end{equation}
    where $a\in\Rbb^+$. Then, the function $f$ is invex.
\end{lemma}
\begin{proof}
    The function $f(x)$ is radially unbounded since
    \begin{equation}
        \lim_{x \to \infty} \log(x) + \frac{a} {x} = \infty,
    \end{equation}
    and
    \begin{equation}
        \lim_{x \to 0} \log(x) + \frac{a} {x} = \infty.
    \end{equation}
    Furthermore, there is a unique stationary point at $x=a$. Thus, $f(x)$ is invex and has a lower bound $\log(a) + 1$.
\end{proof}

\begin{lemma}\label{lem:invexity_logdet_trinv}
    Given $a,b \in \Rbb_{+}$, $\mxP\in \Sbb_{++}^{n\times n}$, 
    define function $f:\Sbb_{++}^{n \times n}\to \Rbb$ as    
    \begin{equation}
        f(\mxX) := a\logdet(\mxX) + b\trace(\mxX^{-1}\mxP), \qquad \forall\, \mxX\in\Sbb_{++}^{n\times n}.
    \end{equation}
    Then, the function $f(\mxX)$ diverges to infinity as $\mxX$ approaches the boundary of $\Sbb_{++}^{n \times n}$ or infinity. 
    Moreover, $f$ admits a unique minimizer in $\Sbb_{++}^{n \times n}$, which implies that $f$ is an invex function on $\Sbb_{++}^{n \times n}$.    
\end{lemma}
\begin{proof}
    Since $\mxP$ is positive definite, there exists positive scalar $\varepsilon$, such that $\mxP > \varepsilon^2 \mathbf{I}$. Thus, we have
    \begin{equation}
        \trace(\mxX^{-1} \mxP) = \trace(\mxP^{\frac{1}{2}} \mxX^{-1} \mxP^{\frac{1}{2}}) \geq \varepsilon^2 \trace(\mxX^{-1}),
    \end{equation}
    which implies that
    \begin{equation}\label{eqn:f(X)_lower_bound}
    \begin{split}
        f(\mxX) &\geq a \log\det \mxX + \varepsilon^2 b \trace(\mxX^{-1})
        \\ &
        = \sum_{i=1}^n \left( a \log \lambda_i(\mxX) + \frac{b \varepsilon^2}{\lambda_i(\mxX)} \right),
    \end{split}
    \end{equation}
    for any $\mxX$.
    As $\mxX$ approaches the boundary of $\Sbb_{++}^{n \times n}$, the smallest eigenvalue of $\mxX$ goes to zero. Consequently, by Lemma~\ref{lem:invex} and \eqref{eqn:f(X)_lower_bound}, $f(\mxX)$ diverges to infinity. Similarly, as $\mxX$ grows unbounded, its largest eigenvalue diverges to infinity. Applying Lemma~\ref{lem:invex} and \eqref{eqn:f(X)_lower_bound} again, we conclude that $f(\mxX)$ also diverges to infinity.
    Now, we show $f$ has a unique minimizer in $\Sbb_{++}^{n \times n}$.  To find all the local minimum, we use the first-order necessary condition for optimality. Accordingly, we obtain
    \begin{equation}
        \begin{split}
            \pdv{f(\mxX)}{\mxX} &= a \mxX^{-1} - b\varepsilon^2 \mxX^{-1}\mxP\mxX^{-1} = 0,
        \end{split}
    \end{equation}
    which implies
    \begin{equation}
        \mxX = \frac{b\varepsilon^2}{a} \mxP^{-1}.
    \end{equation}
    Thus, $f$ has only one stationary point, and, due to the discussion above,  it needs to be the unique minimizer of $f$. Therefore, $f$ is an invex function on $\Sbb_{++}^{n\times n}$, which concludes the proof. 
    %
\end{proof}

\begin{lemma}\label{lem:PD}
    For any positive definite matrices $\mxP, \mxQ  \in \Sbb_{++}^{n\times n}$ and any matrices $\mxR,\mxX\in\Rbb^{n\times n}$, we have 
    \begin{equation}
            (\mathbf{I} - \mxX \mxR)
        \mxP 
            (\mathbf{I} - \mxX \mxR)
        ^\tr
        +
        \mxX\mxQ\mxX^\tr
        \succ 0.
    \end{equation}
    Moreover, one can show that the lower bound is 
\begin{equation}
    (\mxP^{-1} + \mxR^\tr \mxQ^{-1}\mxR)^{-1}.
\end{equation}
\end{lemma}

\begin{lemma}\label{lem:PD for summarion u_t}
    Let Assumption~\ref{assum:input} holds. Then, the following matrix 
    \begin{equation}
        \sum_{t=0}^{\nD-1}
        \begin{bmatrix}
            1 \\ \vcu_t
        \end{bmatrix}
        \begin{bmatrix}
            1 \\ \vcu_t
        \end{bmatrix}^\tr
    \end{equation}
    is positive definite.
\end{lemma}
\begin{proof}
    We prove this by contradiction. Assume there exists a non zero vector $\vcz$, such that
    \begin{equation}
        \vcz
        \left(
        \sum_{t=0}^{\nD-1}
        \begin{bmatrix}
            1 \\ \vcu_t
        \end{bmatrix}
        \begin{bmatrix}
            1 \\ \vcu_t
        \end{bmatrix}^\tr
        \right)
        \vcz = 0.
    \end{equation}
    The above equation implies that
    \begin{equation}
        \begin{bmatrix}
        1 \\ \vcu_t
        \end{bmatrix}^\tr
        \vcz = 0, \quad \forall t \in 0,1,\cdots,\nD-1.
    \end{equation}
    Thus, we have
    \begin{equation}
        \begin{bmatrix}
            1 & 1 & \cdots & 1\\
            \vcu_0 & \vcu_1 &\cdots & \vcu_{\nD-1}
        \end{bmatrix}^\tr
        \vcz = 0.
    \end{equation}
    From Assumption~\ref{assum:input} the transposed matrix in the above equation is full column rank, which implies $\vcz = 0$. Since we assume $\vcz$ is a non zero vector, we arrive at a contradiction, which concludes the proof.
\end{proof}

\begin{lemma}\label{lem:PD for summarion y_t}
    Let Assumption~\ref{assum:output} holds. Then, the matrices 
    \begin{equation}
        \sum_{t=0}^{\nD-1}
        \begin{bmatrix}
            \vcu_t \\ \vcy_t
        \end{bmatrix}
        \begin{bmatrix}
            \vcu_t \\ \vcy_t
        \end{bmatrix}^\tr,
    \end{equation}
    and
    \begin{equation}
        \sum_{t=0}^{\nD-1}
            \vcy_t \vcy_t^\tr
    \end{equation}
    are positive definite.
\end{lemma}
\begin{proof}
    The proof is similar to Lemma~\ref{lem:PD for summarion u_t}.
\end{proof}

\begin{lemma}\label{lem:PD_phi_psi}
Let Assumption~\ref{assum:input} holds. Define matrices $\Phi$ and $\Psi$ as
\begin{equation}
    \Phi :=    \mathbb{E}_{p(\mathbf{x}|\mathbf{y},\hat{\vctheta}_k,\mathbf{u})} \left[ \sum_{t=0}^{\nD - 1} \left(\begin{bmatrix}
        \vcx_t \\ \vcu_t
    \end{bmatrix}
    \begin{bmatrix}
        \vcx_t  \\ \vcu_t
    \end{bmatrix}^\tr \right)\right],
\end{equation}
and
\begin{equation}
    \Psi :=    \mathbb{E}_{p(\mathbf{x}|\mathbf{y},\hat{\vctheta}_k,\mathbf{u})} \left[ \sum_{t=0}^{\nD - 1} \left(\begin{bmatrix}
        \vcx_t \\ \vcu_t \otimes \vcx_t \\ \vcu_t
    \end{bmatrix}
    \begin{bmatrix}
        \vcx_t \\ \vcu_t \otimes \vcx_t \\ \vcu_t
    \end{bmatrix}^\tr \right)\right],
\end{equation}
respectively. Then $\Phi$ and $\Psi$ are positive definite. 
\end{lemma}
\begin{proof}
    We show that $\Psi$ is positive definite, which implies that $\Phi$ is positive definite by Schur complement. 
    To show that $\Psi$ is positive definite, we use contradiction. Assume there exists a non zero vector $\vcq = \begin{bmatrix}
        \vcq_1^\tr & \vcq_2^\tr & \vcq_3^\tr
    \end{bmatrix}$
    such that $\vcq^\tr\Psi\vcq = 0$, where $\vcq \in \Rbb^{\Nx+\Nu+\Nx\Nu}$, $\vcq_1 \in \Rbb^{\Nx}$, $\vcq_2 \in \Rbb^{\Nx\Nu}$, and $\vcq_3 \in \Rbb^{\Nu}$. Then, we have
    \begin{equation}\label{eq:Psi=0}
        \sum_{t=0}^{\nD-1} \mathbb{E}_{p(\mathbf{x}|\mathbf{y},\hat{\vctheta}_k,\mathbf{u})} 
        \left[ \vcq^\tr  \begin{bmatrix}
        \vcx_t \\ \vcu_t \otimes \vcx_t \\ \vcu_t
    \end{bmatrix}
    \begin{bmatrix}
        \vcx_t \\ \vcu_t \otimes \vcx_t \\ \vcu_t
    \end{bmatrix}^\tr \vcq  \right] = 0.
    \end{equation}
    Define 
    \begin{equation}
        \xi_t := \begin{bmatrix}
        \vcx_t \\ \vcu_t \otimes \vcx_t \\ \vcu_t
        \end{bmatrix}^\tr \vcq.
    \end{equation} 
    Notice that $\xi_t$ is a scalar. Thus, from \eqref{eq:Psi=0}, we have $\mathbb{E}_{p(\mathbf{x}|\mathbf{y},\hat{\vctheta}_k,\mathbf{u})} 
        \left[\xi_t^2 \right] = 0, \forall t = 1,\cdots,\nD-1$,
    which implies that 
    \begin{equation}\label{eq:mu_xi_t}
        \mu_{\xi_t} := \mathbb{E}_{p(\mathbf{x}|\mathbf{y},\hat{\vctheta}_k,\mathbf{u})} 
        \left[\xi_t \right] = 0, \quad \forall t = 1,\cdots,\nD-1,
    \end{equation}
    and
    \begin{equation}\label{eq:var_xi_t}
        \mathbb{E}_{p(\mathbf{x}|\mathbf{y},\hat{\vctheta}_k,\mathbf{u})} 
        \left[(\xi_t - \mu_{\xi_t})^2 \right] = 0, \quad \forall t = 1,\cdots,\nD-1.
    \end{equation}
On the one hand, \eqref{eq:var_xi_t} implies that
\begin{equation}
    \sum_{t=0}^{\nD-1}
    \begin{bmatrix}
        \vcq_1^\tr & \vcq_2^\tr
    \end{bmatrix}
    \mathbb{E}_{p(\mathbf{x}|\mathbf{y},\hat{\vctheta}_k,\mathbf{u})}
    \left[ 
    \left(
    \begin{bmatrix}
        \vcx_t \\ \vcu_t \otimes \vcx_t
    \end{bmatrix}
    -
    \begin{bmatrix}
        \hat{\vcx}_{t|\nD} \\ \vcu_t \otimes \hat{\vcx}_{t|\nD}
    \end{bmatrix}
    \right)
        \left(
    \begin{bmatrix}
        \vcx_t \\ \vcu_t \otimes \vcx_t
    \end{bmatrix}
    -
    \begin{bmatrix}
        \hat{\vcx}_{t|\nD} \\ \vcu_t \otimes \hat{\vcx}_{t|\nD}
    \end{bmatrix}
    \right)^\tr
    \right]
    \begin{bmatrix}
        \vcq_1 \\ \vcq_2
    \end{bmatrix}
    = 0,
\end{equation}
which is equivalent to
\begin{equation}\label{eq:q_1_q_2}
    \begin{bmatrix}
        \vcq_1^\tr & \vcq_2^\tr
    \end{bmatrix}
    \left[
    \sum_{t=0}^{\nD-1}
    \left(
    \mxP_{t|\nD}
    \otimes
    \begin{bmatrix}
         1 \\ \vcu_t   
    \end{bmatrix}
    \begin{bmatrix}
         1 \\ \vcu_t   
    \end{bmatrix}^\tr
    \right)
    \right]
    \begin{bmatrix}
        \vcq_1 \\ \vcq_2
    \end{bmatrix}
    = 0.
\end{equation}
From Lemma~\ref{lem:PD for summarion u_t} and \eqref{eq:P_t|nD}, we know $\sum_{t=0}^{\nD-1}\left(\begin{bmatrix}
         1 \\ \vcu_t   
    \end{bmatrix}
    \begin{bmatrix}
         1 \\ \vcu_t   
    \end{bmatrix}^\tr \right)$ and $\mxP_{t|\nD}$ are positive definite. Thus, \eqref{eq:q_1_q_2} holds if and only if $\vcq_1 = 0$ and $\vcq_2=0$. On the other hand, \eqref{eq:mu_xi_t} implies that
    \begin{equation}
        \hat{x}_{t|\nD}^\tr \vcq_1 + \left(\vcu_t^\tr \otimes \hat{\vcx}_{t|\nD}^\tr \right) \vcq_2 + \vcu_t^\tr \vcq_3 = 0, \quad \forall t = 1,\cdots,\nD-1,
    \end{equation}
    which implies that 
    \begin{equation}
        \begin{bmatrix}
            \vcu_0^\tr \\ \vcu_1^\tr \\ \cdots \\ \vcu_{\nD-1}^\tr
        \end{bmatrix}\vcq_3 = 0.
    \end{equation}
    From Assumption~\ref{assum:input}, the above equation implies that $\vcq_3 = 0$. Thus, we arrive at $\vcq = 0$, which contradicts the assumption that $\vcq$ is a non zero vector. Therefore, $\Psi$ is positive definite and thus invertible, which concludes the proof.
\end{proof}
\section{Proof of Proposition 1} \label{app:proof 1}
One can see that the integrand in \eqref{eq: likelihoog int} is the exponential of a quadratic function with respect to $\mathbf{x}$. 
Accordingly, we can employ Lemma~\ref{lem:Gaussian_integration} to evaluate and simplify the integral.
To this end, we define the matrix $\Sqtheta$ and the vector $\vcq(\vctheta)$ as
\begin{equation}\label{eqn:Sq}
    \Sqtheta 
    := 
    \Big(
    \Xi(\vctheta)^\tr\,
    \Sbfytheta^{-1}\, 
    \Xi(\vctheta)
    + 
    \Sbfxtheta^{-1}
    \Big)^{-1}\!,
    %
    %
\end{equation}
and 
\begin{equation}\label{eq:q function}
\begin{split}
    \vcq(\vctheta)
    &:= 
    -\,
    \Sqtheta
    \Big[
    \Xi(\vctheta)^\tr
    \Sbfytheta^{-1}\,
    ( \mathbf{D}(\vctheta)\mathbf{u}-\mathbf{y})
    -
    \Sbfxtheta^{-1}
    \mubfxtheta
    \Big],
\end{split}
\end{equation}
respectively.
By expanding the exponent term in \eqref{eq: likelihoog int}, we have
\begin{align}
    \big( \mathbf{y}-\mubfythetax & \big)^\tr
    \Sbfytheta^{-1}
    \big( \mathbf{y}-\mubfythetax \big) 
    + 
    \big( \mathbf{x}-\mubfxtheta \big)^\tr
    \Sbfxtheta^{-1} 
    \big( \mathbf{x}-\mubfxtheta \big)
    \nonumber\\
    =& \, 
    \big\| \Xi(\vctheta) \mathbf{x} + \mathbf{D}(\vctheta) \mathbf{u} -  \mathbf{y} \big\|^2_{\Sbfytheta^{-1}} 
    + 
    \big\| \mathbf{x} - \mubfxtheta \big\|^2_{\Sbfxtheta^{-1}} 
    \label{eq: S_q Gaussian}\\
    =& \, 
    \big\| \mathbf{x} - \vcq(\vctheta) \big\|^2_{\Sqtheta^{-1}} 
    - 
    \vcq(\vctheta)^\tr \Sqtheta^{-1} \vcq(\vctheta) 
    + 
    \big(\mathbf{D}(\vctheta)\mathbf{u} - \mathbf{y}\big)^\tr
    \Sbfytheta^{-1}
    \big(\mathbf{D}(\vctheta)\mathbf{u} - \mathbf{y}\big) 
    + 
    \mubfxtheta^\tr\Sbfxtheta^{-1}\mubfxtheta.
    \nonumber
\end{align}
Note that only the first term in the last line of \eqref{eq: S_q Gaussian} depends on $\mathbf{x}$.
Thus, substituting \eqref{eq: S_q Gaussian} into \eqref{eq: likelihoog int} and using the property \eqref{eq:Gaussian integration} introduced in Lemma~\ref{lem:Gaussian_integration}, we can simplify the likelihood function as
\begin{equation}\label{eq:likelihood integrated}
\begin{split}
     p(\mathbf{y}|\vctheta, \mathbf{u}) 
     = 
     (2\pi)^{-\frac{1}{2}\Ny\nsD}\,
     \frac{ |\Sqtheta|^{\frac{1}{2}}}
           {|\Sbfxtheta|^{\frac{1}{2}} |\Sbfytheta|^{\frac{1}{2}}} 
    \exp \bigg( &\frac{1}{2} \Big[ 
        \vcq(\vctheta)^\tr \Sqtheta^{-1} \vcq(\vctheta) 
        - 
        \big(\mathbf{D}(\vctheta)\mathbf{u} - \mathbf{y}\big)^\tr
        \Sbfytheta^{-1}
        \big(\mathbf{D}(\vctheta)\mathbf{u} - \mathbf{y}\big) 
     \\&\qquad \qquad \qquad \qquad \qquad \qquad \quad 
     - \mubfxtheta^\tr \Sbfxtheta^{-1} \mubfxtheta
     \Big] \bigg).
\end{split}
\end{equation}
For further simplification, we define the vector $\vcz(\vctheta)$ and the matrix $\mxF(\vctheta)$ as
\begin{align}\label{eqn:z_theta}
    \vcz(\vctheta) := \mathbf{D}(\vctheta)\mathbf{u} - \mathbf{y},
\end{align}
and
\begin{align}
    \mxF(\vctheta) := \Xi(\vctheta)\Sbfxtheta\Xi(\vctheta)^\tr + \Sbfytheta,
\end{align}
respectively.
Using Lemma \eqref{eq:matrix inversion lemma}, i.e., 
\emph{Matrix Inversion Lemma} or \emph{Sherman--Morrison--Woodbury Formula},  
we have
\begin{align}
        \mxF(\vctheta) ^{-1}
        &= 
        \big(\Sbfytheta + \Xi(\vctheta)\Sbfxtheta\Xi(\vctheta)^\tr\big)^{-1} 
        \nonumber\\
        &= 
        \Sbfytheta^{-1} -\Sbfytheta^{-1}\Xi(\vctheta)
        \big(\Sbfxtheta^{-1} + \Xi(\vctheta)^\tr \Sbfytheta^{-1} \Xi(\vctheta)\big)^{-1}
        \Xi(\vctheta)^\tr\Sbfytheta^{-1} 
        \label{eq:F-1}\\
        &=
        \Sbfytheta^{-1} -\Sbfytheta^{-1}\Xi(\vctheta)
        \Sqtheta
        \Xi(\vctheta)^\tr\Sbfytheta^{-1},
        \nonumber
\end{align}
where the last equality is due to  the definition of matrix $\Sqtheta$ given in \eqref{eqn:Sq}.
Similarly, one can see that
\begin{align}        
        \Sqtheta 
        &\,= 
        \Sbfxtheta - 
        \Sbfxtheta\Xi(\vctheta)^\tr
        \big(\Sbfytheta + \Xi(\vctheta) \Sbfxtheta \Xi(\vctheta)^\tr\big)^{-1} \Xi(\vctheta)\Sbfxtheta,
        \label{eq:Sq-1}
\end{align}
We employ \eqref{eq:q function} and \eqref{eqn:z_theta} to simplify the exponent in \eqref{eq:likelihood integrated}, thereby obtaining
\begin{align}
        \vcq(\vctheta)^\tr\Sqtheta^{-1}&\vcq(\vctheta)\, - \, (\mathbf{D}(\vctheta)\mathbf{u}-\mathbf{y})^\tr\Sbfytheta^{-1}(\mathbf{D}(\vctheta)\mathbf{u}-\mathbf{y}) - \mubfxtheta^\tr\Sbfxtheta^{-1}\mubfxtheta
        \\
        =\,& 
        \big[ \vcz(\vctheta)^\tr\Sbfytheta^{-1}\Xi(\vctheta) - \mubfxtheta^\tr\Sbfxtheta^{-1}\big]
        \, \Sqtheta\, 
        \big[ \vcz(\vctheta)^\tr\Sbfytheta^{-1}\Xi(\vctheta) - \mubfxtheta^\tr\Sbfxtheta^{-1} \big]^\tr
        \nonumber
        \\&
        \qquad\qquad\qquad
        - \vcz(\vctheta)^\tr\Sbfytheta^{-1}\vcz(\vctheta) - \mubfxtheta^\tr\Sbfxtheta^{-1}\mubfxtheta.
\end{align}
By expanding the above equation further and rearranging the terms, we arrive at
\begin{align}
        \vcq(\vctheta)^\tr\Sqtheta^{-1}&\vcq(\vctheta)\, - \, (\mathbf{D}(\vctheta)\mathbf{u}-\mathbf{y})^\tr\Sbfytheta^{-1}(\mathbf{D}(\vctheta)\mathbf{u}-\mathbf{y}) - \mubfxtheta^\tr\Sbfxtheta^{-1}\mubfxtheta
        \\
        =\,& 
        \big[ 
            \vcz(\vctheta)^\tr \Sbfytheta^{-1} \Xi(\vctheta) \Sqtheta 
            - 
            \mubfxtheta^\tr\Sbfxtheta^{-1}\Sqtheta 
        \big]
        \big[\Xi(\vctheta)^\tr\Sbfytheta^{-1}\vcz(\vctheta) \big] 
        \nonumber\\&\qquad
        -
        \big[ 
            \vcz(\vctheta)^\tr \Sbfytheta^{-1} \Xi(\vctheta) \Sqtheta 
            - 
            \mubfxtheta^\tr\Sbfxtheta^{-1}\Sqtheta
        \big]
        \big[ \Sbfxtheta^{-1} \mubfxtheta \big] 
        \nonumber
        \\&\qquad
        - 
        \vcz(\vctheta)^\tr \Sbfytheta^{-1} \vcz(\vctheta) 
        - 
        \mubfxtheta^\tr \Sbfxtheta^{-1} \mubfxtheta
        \\
        =\,& 
        \vcz(\vctheta)^\tr\Sbfytheta^{-1}\Xi(\vctheta)\Sqtheta \Xi(\vctheta)^\tr\Sbfytheta^{-1}\vcz(\vctheta)  
        - 
        \mubfxtheta^\tr\Sbfxtheta^{-1}\Sqtheta \Xi(\vctheta)^\tr\Sbfytheta^{-1}\vcz(\vctheta) 
        \nonumber\\&\qquad
        -
        \vcz(\vctheta)^\tr\Sbfytheta^{-1}\Xi(\vctheta)\Sqtheta \Sbfxtheta^{-1}\mubfxtheta 
        + 
        \mubfxtheta^\tr\Sbfxtheta^{-1}\Sqtheta  \Sbfxtheta^{-1}\mubfxtheta
        \nonumber\\&\qquad
        - \vcz(\vctheta)^\tr\Sbfytheta^{-1}\vcz(\vctheta) - \mubfxtheta^\tr\Sbfxtheta^{-1}\mubfxtheta
        \\
       =\,& 
       \vcz(\vctheta)^\tr\Sbfytheta^{-1}\Xi(\vctheta)\mxS_{\vcq}(\vctheta)\Xi(\vctheta)^\tr\Sbfytheta^{-1}\vcz(\vctheta) + \mubfxtheta^\tr\Sbfxtheta^{-1}\mxS_{\vcq}(\vctheta)\Sbfxtheta^{-1}\mubfxtheta 
       \nonumber\\&\qquad
       - 2\vcz(\vctheta)^\tr\Sbfytheta^{-1}\Xi(\vctheta)\Sqtheta\Sbfxtheta^{-1}\mubfxtheta - \vcz(\vctheta)^\tr\Sbfytheta^{-1}\vcz(\vctheta) - \mubfxtheta^\tr\Sbfxtheta^{-1}\mubfxtheta.
\end{align}
Additionally, by applying \eqref{eq:F-1} and \eqref{eq:Sq-1}, the above equation reduces to
\begin{align}
        \vcq(\vctheta)^\tr\Sqtheta^{-1}&\vcq(\vctheta)\, - \, (\mathbf{D}(\vctheta)\mathbf{u}-\mathbf{y})^\tr\Sbfytheta^{-1}(\mathbf{D}(\vctheta)\mathbf{u}-\mathbf{y}) - \mubfxtheta^\tr\Sbfxtheta^{-1}\mubfxtheta
        \\
        = \,& -\vcz(\vctheta)^\tr\mxF(\vctheta)^{-1}\vcz(\vctheta) + \mubfxtheta^\tr\Sbfxtheta^{-1}\mxS_{\vcq}(\vctheta)\Sbfxtheta^{-1}\mubfxtheta
        \nonumber\\&\qquad
        - 2\vcz(\vctheta)^\tr\Sbfytheta^{-1}\Xi(\vctheta)\Sqtheta\Sbfxtheta^{-1}\mubfxtheta - \mubfxtheta^\tr\Sbfxtheta^{-1}\mubfxtheta
        \\
        =\,& 
        -\vcz(\vctheta)^\tr\mxF(\vctheta)^{-1}\vcz(\vctheta) - \mubfxtheta^\tr\Xi(\vctheta)^\tr\mxF(\vctheta)^{-1}\Xi(\vctheta)\mubfxtheta - 2\vcz(\vctheta)^\tr\Sbfytheta^{-1}\Xi(\vctheta)\Sqtheta\Sbfxtheta^{-1}\mubfxtheta,
        \label{eq:exponent}
\end{align}        
%
On the other hand, it can be shown that the last term in the final line of the above equation is equivalent to $2\vcz(\vctheta)^\tr\mxF(\vctheta)^{-1}\Xi(\vctheta)\mubfxtheta$. More precisely, we have
\begin{align}
    \vcz(\vctheta)^\tr\Sbfytheta^{-1}
    &
    \Xi(\vctheta) \Sqtheta^{-1} \Sbfxtheta^{-1} \mubfxtheta 
    - 
    \vcz(\vctheta)^\tr \mxF(\vctheta)^{-1} \Xi(\vctheta) \mubfxtheta
    \\
    =\,& 
    \vcz(\vctheta)^\tr 
    \Sbfytheta^{-1} \Xi(\vctheta) \Sqtheta^{-1}\Sbfxtheta^{-1}
    \mubfxtheta 
    - 
    \vcz(\vctheta)^\tr
    \Sbfytheta^{-1}
    \Xi(\vctheta)
    \mubfxtheta
    \nonumber
    \\&
    \qquad
    + \vcz(\vctheta)^\tr\Sbfytheta^{-1}\Xi(\vctheta)\Sqtheta^{-1}\Xi(\vctheta)^\tr\Sbfytheta^{-1}\Xi(\vctheta)\mubfxtheta\\
    =\,& 
    \vcz(\vctheta)^\tr\Sbfytheta^{-1}\Xi(\vctheta)\Sqtheta^{-1}(\Sbfxtheta^{-1} + \Xi(\vctheta)^\tr\Sbfytheta^{-1}\Xi(\vctheta))\mubfxtheta - \vcz(\vctheta)^\tr\Sbfytheta^{-1}\Xi(\vctheta)\mubfxtheta
    \\
    =\,& 
    \vcz(\vctheta)^\tr\Sbfytheta^{-1}\Xi(\vctheta)\mubfxtheta -\vcz(\vctheta)^\tr\Sbfytheta^{-1}\Xi(\vctheta)\mubfxtheta
    \\
    =\,& 0,
\end{align}
where the first equality is due to \eqref{eq:F-1}. Thus, \eqref{eq:exponent} can be furthered simplified to 
\begin{equation}\label{eq: simlified exp}
    \begin{split}
        -\,
        \vcz(\vctheta)^\tr \mxF(\vctheta)^{-1} \vcz(\vctheta) 
        \,-\, 
        &
        \mubfxtheta^\tr\Xi(\vctheta)^\tr \mxF(\vctheta)^{-1}  \Xi(\vctheta) \mubfxtheta 
        - 
        2 \vcz(\vctheta)^\tr \mxF(\vctheta)^{-1} \Xi(\vctheta) \mubfxtheta
        \\
        & = 
        -
        \big(\vcz(\vctheta)^\tr + \mubfxtheta^\tr\Xi(\vctheta)^\tr\big)
        \mxF(\vctheta)^{-1}
        \big(\vcz(\vctheta) + \Xi(\vctheta)\mubfxtheta\big).
    \end{split}
\end{equation}
Substituting \eqref{eq: simlified exp} into \eqref{eq:likelihood integrated} and applying the logarithm, we obtain the log-likelihood function as
\begin{equation}
\begin{split}
    \log\, p(\mathbf{y}|\vctheta,\mathbf{u}) 
    &= 
    \frac{1}{2} 
    \Big[
        \logdet\big(\Sqtheta\big)
        - 
        \logdet\big(\Sbfxtheta\big)
        - 
        \logdet\big(\Sbfytheta\big) 
    \\&
    \qquad \qquad \qquad 
        - 
        \big(\vcz(\vctheta)^\tr + \mubfxtheta^\tr\Xi(\vctheta)^\tr\big)
        \mxF(\vctheta)^{-1}
        \big(\vcz(\vctheta) + \Xi(\vctheta)\mubfxtheta\big)
    \Big] + \frac12 \Ny\nD\log(2\pi).
\end{split}
\end{equation}
Thus, due to definition of matrix $\mxF(\vctheta)$ and the strictly increasing monotonicity of the logarithm function, 
the maximum likelihood problem \eqref{eqn:ML_generic}, which entails maximizing \eqref{eq: likelihoog int} over $\Theta$, is equivalent to minimizing the function $J$, as defined in \eqref{eq:cost function ML}, over $\Theta$, i.e., the optimization problem \eqref{eqn:min_J}.
This concludes the proof.
\qed
\section{Proof of Proposition 2}\label{app:proof PD of Expectation}
To show \eqref{eq:F_k}-\eqref{eq:H_k} are positive definite, we first expand the right half side of the equations. To this end, we use the following expectations,
\begin{align}
\mathbb{E}_{p(\mathbf{x}|\mathbf{y},\hat{\vctheta}_k,\mathbf{u})}[\vcx_t|\mathbf{y}, \hat{\vctheta}_k] &= \hat{\vcx}_{t|\nD}, \label{eq:expectation 1}\\
\mathbb{E}_{p(\mathbf{x}|\mathbf{y},\hat{\vctheta}_k,\mathbf{u})}[\vcx_t \vcx_t^\tr|\mathbf{y}, \hat{\vctheta}_k] &= \hat{\vcx}_{t|\nD}\hat{\vcx}_{t|\nD}^\tr + \mxP_{t|\nD}, \label{eq:expectation 2} \\
\mathbb{E}_{p(\mathbf{x}|\mathbf{y},\hat{\vctheta}_k,\mathbf{u})}[\vcx_{t+1} \vcx_t^\tr|\mathbf{y}, \hat{\vctheta}_k] &= \hat{\vcx}_{t+1|{\nD}}\hat{\vcx}_{t|{\nD}}^\tr + \mxP_{t+1|{\nD}}\mxL_{t,k}^\tr, \label{eq:expectation 3}
\end{align}
which can be estimated using RTS smoother as in \eqref{eq:mean and covariance} and \eqref{eq:correlation}.
Here, the subscript in $\mxL_{t,k}$ means the recursive updates $\mxL_t$ using the estimated parameter $\hat{\vctheta}_k$. Consider $\mxH_k(\mxM)$, using the above expectations, it can be expanded as
\begin{align}
    \mxH_k(\mxM) &= 
        \frac{1}{\nD} \sum_{t=0}^{\nD-1}  \Big[(\hat{\vcx}_{t+1|\nD}\hat{\vcx}_{t+1|\nD}^\tr + \mxP_{t+1|\nD}) -  \begin{bmatrix}\hat{\vcx}_{t+1|{\nD}}\hat{\vcx}_{t|{\nD}}^\tr + \mxP_{t+1|{\nD}}\mxL_{t,k}^\tr &\quad \hat{\vcx}_{t+1|\nD} \vcu_t^\tr  \end{bmatrix} {\mxM}^\tr  \nonumber \\
        &\quad -  {\mxM}\begin{bmatrix}\hat{\vcx}_{t+1|\nD}\hat{\vcx}_{t|\nD}^\tr + \mxP_{t+1|\nD}\mxL_{t,k}^\tr &\quad \hat{\vcx}_{t+1|\nD} \vcu_t^\tr  \end{bmatrix}^\tr + {\mxM} \begin{bmatrix}\hat{\vcx}_{t|\nD}\hat{\vcx}_{t|\nD}^\tr \!+\! \mxP_{t|\nD} & \hat{\vcx}_{t|\nD}\vcu_t^\tr\\ \vcu_t\hat{\vcx}_{t|\nD}^\tr &\vcu_t \vcu_t^\tr\end{bmatrix} {\mxM}^\tr \Big] \\
        &= \frac{1}{\nD} \sum_{t=0}^{\nD-1} 
        \big[ (\hat{\vcx}_{t+1|\nD} - 
        {\mxM}
        \begin{bmatrix}
            \hat{\vcx}_{t+1|\nD} \\ \vcu_t
        \end{bmatrix}
        )
        (\hat{\vcx}_{t+1|\nD} - 
        {\mxM}
        \begin{bmatrix}
            \hat{\vcx}_{t+1|\nD} \\ \vcu_t
        \end{bmatrix}
        )^\tr
        \nonumber \\
        &\quad + \mxP_{t+1|\nD} -  \mxP_{t+1|\nD}\mxL_{t,k}^\tr {\mxA}^\tr - {\mxA} \mxL_{t,k} \mxP_{t+1|\nD} + {\mxA}\mxP_{t|\nD}{\mxA}^\tr\big]. \label{eq:H_k 2}
\end{align}
To show that the above equation is positive definite, we first expand $\mxP_{t|\nD}$. Using \eqref{eq:mean and covariance}, it can be written as 
\begin{align}
    \mxP_{t|\nD} 
    &= \mxP_{t|t} + \mxL_{t,k}(\mxP_{t+1|\nD} - \mxP_{t+1|t})\mxL_{t,k}^\tr\\
    &= \mxP_{t|t} - \mxL_{t,k}\mxP_{t+1|t}\mxL_{t,k}^\tr + \mxL_{t,k}\mxP_{t+1|\nD}\mxL_{t,k}^\tr\\
    &= \mxP_{t|t} - \mxP_{t|t}\hat{\mxA}_k^\tr\mxP_{t+1|t}^{-1}\hat{\mxA}_k\mxP_{t|t} +\mxL_{t,k}\mxP_{t+1|\nD}\mxL_{t,k}^\tr\\
    &=(\mxP_{t|t}^{-1} + \hat{\mxA}_k^\tr \hat{\mxS}_{\vcw}^{-1} \hat{\mxA}_k)^{-1} + \mxL_{t,k}\mxP_{t+1|\nD}\mxL_{t,k}^\tr.\label{eq:P_t|nD}
\end{align}
In the last equality, we use Lemma~\ref{lem:matrix_inversion_lemma}. Notice that $\hat{\mxA}_k$ is the parameter estimates from the last iteration and thus is known. Since $\mxP_{t|t}$ is positive definite, \eqref{eq:P_t|nD} implies that $\mxP_{t|\nD}$ is also positive definite. By substituting \eqref{eq:P_t|nD} into \eqref{eq:H_k 2}, we have
\begin{align}
    \mxH_k(\mxM) &= 
    \frac{1}{\nD} \sum_{t=0}^{\nD-1} 
        \big[ (\hat{\vcx}_{t+1|\nD} - 
        {\mxM}
        \begin{bmatrix}
            \hat{\vcx}_{t+1|\nD} \\ \vcu_t
        \end{bmatrix}
        )
        (\hat{\vcx}_{t+1|\nD} - 
        {\mxM}
        \begin{bmatrix}
            \hat{\vcx}_{t+1|\nD} \\ \vcu_t
        \end{bmatrix}
        )^\tr \nonumber
        \\& \quad
        + (\mxP_{t+1|\nD}^{\frac{1}{2}} - {\mxA}\mxL_{t,k}\mxP_{t+1|\nD}^{\frac{1}{2}})(\mxP_{t+1|\nD}^{\frac{1}{2}} - {\mxA}\mxL_{t,k}\mxP_{t+1|\nD}^{\frac{1}{2}})^\tr 
        + \mxA(\mxP_{t|t}^{-1} + \hat{\mxA}_k \hat{\mxS}_{\vcw,k} \hat{\mxA}_k)^{-1}{\mxA}^\tr \big] \label{eq:H_theta|hat_theta PD1}\\
    &= 
    \frac{1}{\nD} \sum_{t=0}^{\nD-1} 
        \big[ (\hat{\vcx}_{t+1|\nD} - 
        {\mxM}
        \begin{bmatrix}
            \hat{\vcx}_{t+1|\nD} \\ \vcu_t
        \end{bmatrix}
        )
        (\hat{\vcx}_{t+1|\nD} - 
        {\mxM}
        \begin{bmatrix}
            \hat{\vcx}_{t+1|\nD} \\ \vcu_t
        \end{bmatrix}
        )^\tr \nonumber
        \\& \quad
        + (\mathbf{I} - {\mxA}\mxL_{t,k})\mxP_{t+1|\nD}(\mathbf{I} - {\mxA}\mxL_{t,k})^\tr 
        + \mxA(\mxP_{t|t}^{-1} + \hat{\mxA}_k \hat{\mxS}_{\vcw,k} \hat{\mxA}_k)^{-1}{\mxA}^\tr \big]. \label{eq:H_theta|hat_theta PD2}
\end{align}
Consequently, since the first term in \eqref{eq:H_theta|hat_theta PD2} is positive semi-definite, from Lemma~\ref{lem:PD}, it follows that
\begin{align}
    \mxH_k(\mxM) 
        \succeq
        \left(\mxP_{t+1|\nD}^{-1} + \mxL_{t,k}^\tr (\mxP_{t|t}^{-1} + \hat{\mxA}_k \hat{\mxS}_{\vcw,k} \hat{\mxA}_k) \mxL_{t,k} \right)^{-1}.
\end{align}
Thus, $\mxH_k(\mxM)$ is positive definite for all $\mxM \in \Rbb^{\Nx \times (\Nx + \Nu)}$, with a lower bound which does not depend on $\mxM$.

Next, we show the positive definiteness of $\mxG_k(\mxN)$, which can be expanded as
\begin{align}\!\!\!\!\!\!
    \mxG_k(\mxN) &= \mathbb{E}_{p(\mathbf{x}|\mathbf{y},\hat{\vctheta}_k,\mathbf{u})}
    \Bigg[\sum_{t=0}^{\nD-1} \left(\vcy_t - \mxN \begin{bmatrix} \vcx_t\\ \vcu_t \otimes \vcx_t \\ \vcu_t\end{bmatrix}  \right)
    \left(\vcy_t - \mxN \begin{bmatrix} \vcx_t\\ \vcu_t \otimes \vcx_t \\ \vcu_t\end{bmatrix}  \right)^\tr\Bigg] \label{eq:G_theta|hat_theta PD 1}\\
    &= \sum_{t=0}^{\nD-1} 
    \Bigg[ \vcy_t \vcy_t^\tr 
    - \mxN \begin{bmatrix} 
    \hat{\vcx}_{t|\nD}\\ \vcu_t \otimes \hat{\vcx}_{t|\nD} \\ \vcu_t\end{bmatrix}\vcy_t^\tr
    - \vcy_t \begin{bmatrix} \hat{\vcx}_{t|\nD}\\ \vcu_t \otimes \hat{\vcx}_{t|\nD} \\ \vcu_t\end{bmatrix}^\tr \mxN^\tr
    \nonumber \\ &
    \quad + \mxN 
    \begin{bmatrix} 
    \hat{\vcx}_{t|\nD} \hat{\vcx}_{t|\nD}^\tr + \mxP_{t|\nD} 
    & \vcu_t^\tr \otimes \left(\hat{\vcx}_{t|\nD} \hat{\vcx}_{t|\nD}^\tr + \mxP_{t|\nD} \right) 
    & \hat{\vcx}_{t|\nD} \vcu_t^\tr \\ 
    \vcu_t \otimes \left(\hat{\vcx}_{t|\nD} \hat{\vcx}_{t|\nD}^\tr + \mxP_{t|\nD} \right)
    & \vcu_t \vcu_t^\tr \otimes \left(\hat{\vcx}_{t|\nD} \hat{\vcx}_{t|\nD}^\tr + \mxP_{t|\nD} \right)
    & \vcu_t\vcu_t^\tr \otimes \hat{\vcx}_{t|\nD}\\
    \vcu_t \hat{\vcx}_{t|\nD}^\tr 
    & \vcu_t \vcu_t^\tr \otimes \hat{\vcx}_{t|\nD}^\tr
    & \vcu_t \vcu_t^\tr 
    \end{bmatrix}
    \mxN^\tr
    \Bigg] \label{eq:G_theta|hat_theta PD 2}\\
    &= \sum_{t=0}^{\nD-1} \left[
    \left( \mxN \begin{bmatrix} \hat{\vcx}_{t|\nD}\\ \vcu_t \otimes \hat{\vcx}_{t|\nD} \\ \vcu_t\end{bmatrix} - \vcy_t \right) \left( \mxN \begin{bmatrix} \hat{\vcx}_{t|\nD}\\ \vcu_t \otimes \hat{\vcx}_{t|\nD} \\ \vcu_t\end{bmatrix} - \vcy_t \right)^\tr
    + \mxC 
    \begin{bmatrix}  
    \mxP_{t|\nD} 
    & \vcu_t^\tr \otimes\mxP_{t|\nD}\\
    \vcu_t \otimes \mxP_{t|\nD}
    & \vcu_t \vcu_t^\tr \otimes \mxP_{t|\nD} 
    \end{bmatrix}
    \mxC^\tr
    \right] \label{eq:G_theta|hat_theta PD 3}\\
    &= \sum_{t=0}^{\nD-1} \left[
    \left( \mxC \begin{bmatrix} \hat{\vcx}_{t|\nD}\\ \vcu_t \otimes \hat{\vcx}_{t|\nD} \end{bmatrix} + \begin{bmatrix}\mxD & -\mathbf{I}\end{bmatrix}
    \begin{bmatrix} \vcu_t \\ \vcy_t \end{bmatrix} \right) 
    \left( \mxC \begin{bmatrix} \hat{\vcx}_{t|\nD}\\ \vcu_t \otimes \hat{\vcx}_{t|\nD} \end{bmatrix} 
    + \begin{bmatrix}\mxD & -\mathbf{I}\end{bmatrix}
    \begin{bmatrix} \vcu_t \\ \vcy_t \end{bmatrix}
    \right)^\tr\right.
    \nonumber \\ &
    \quad + \mxC \left.
    \begin{bmatrix}  
    \mxP_{t|\nD} 
    & \vcu_t^\tr \otimes\mxP_{t|\nD}\\
    \vcu_t \otimes \mxP_{t|\nD}
    & \vcu_t \vcu_t^\tr \otimes \mxP_{t|\nD} 
    \end{bmatrix}
    \mxC^\tr
    \right]. \label{eq:G_theta|hat_theta PD 4}
\end{align}
We show the above equation is positive definite by contradiction. Assume there exists a non zeros vector $\vcz \in \Rbb^{\Ny}$, such that $\vcz^\tr \mxG_k(\mxN) \vcz = 0$. Since both two terms in \eqref{eq:G_theta|hat_theta PD 4} are positive semi-definite, we have
\begin{equation}
\begin{split}
    &\qquad  \vcz^\tr \mxC 
    \sum_{t=0}^{\nD-1}
    \begin{bmatrix}  
    \mxP_{t|\nD} 
    & \vcu_t^\tr \otimes\mxP_{t|\nD}\\
    \vcu_t \otimes \mxP_{t|\nD}
    & \vcu_t \vcu_t^\tr \otimes \mxP_{t|\nD} 
    \end{bmatrix}
    \mxC^\tr \vcz = 0,\\
\end{split}
\end{equation}    
which is equivalent to
\begin{equation}
\begin{split}    
    \vcz^\tr \mxC 
    \sum_{t=0}^{\nD-1}
    \left( \begin{bmatrix}  1 \\ \vcu_t \end{bmatrix}
    \begin{bmatrix}  1 \\ \vcu_t \end{bmatrix}^\tr 
    \otimes \mxP_{t|\nD} \right)
    \mxC^\tr \vcz = 0. 
\end{split}
\end{equation}
As we show $\mxP_{t|\nD} \succ 0$ in \eqref{eq:P_t|nD} and from Lemma~\ref{lem:PD for summarion u_t}, the above equation implies $\mxC^\tr \vcz = 0$. On the other hand, from \eqref{eq:G_theta|hat_theta PD 4}, $\vcz^\tr \mxG_k(\mxN) \vcz = 0$ also implies that 
\begin{equation}
    \begin{split}
    &\qquad 
    \sum_{t=0}^{\nD-1} \left[    \vcz^\tr
    \left( \mxC \begin{bmatrix} \hat{\vcx}_{t|\nD}\\ \vcu_t \otimes \hat{\vcx}_{t|\nD} \end{bmatrix} + \begin{bmatrix}\mxD & -\mathbf{I}\end{bmatrix}
    \begin{bmatrix} \vcu_t \\ \vcy_t \end{bmatrix} \right) 
    \left( \mxC \begin{bmatrix} \hat{\vcx}_{t|\nD}\\ \vcu_t \otimes \hat{\vcx}_{t|\nD} \end{bmatrix} 
    + \begin{bmatrix}\mxD & -\mathbf{I}\end{bmatrix}
    \begin{bmatrix} \vcu_t \\ \vcy_t \end{bmatrix}
    \right)^\tr \vcz \right] = 0,\\
    \end{split}
\end{equation}
which is equivalent to
\begin{equation}
    \begin{split}
    \vcz^\tr
    \begin{bmatrix}\mxD & -\mathbf{I}\end{bmatrix}
    \sum_{t=0}^{\nD-1}
    \begin{bmatrix} 
        \vcu_t\vcu_t^\tr & \vcu_t\vcy_t^\tr\\
        \vcy_t\vcu_t^\tr & \vcy_t\vcy_t^\tr
    \end{bmatrix} 
    \begin{bmatrix}\mxD & -\mathbf{I}\end{bmatrix}^\tr
    \vcz  = 0. \\
    \end{split}
\end{equation}
From Lemma~\ref{lem:PD for summarion y_t}, we know
\begin{align}
  \sum_{t=0}^{\nD-1}
    \begin{bmatrix} 
        \vcu_t\vcu_t^\tr & \vcu_t\vcy_t^\tr\\
        \vcy_t\vcu_t^\tr & \vcy_t\vcy_t^\tr
    \end{bmatrix} \succ 0.  
\end{align}
Thus, we arrive at $\vcz = 0$, which contradicts our assumption that $\vcz$ is a non zero vector. Therefore, $\mxG_k(\mxN)$ is positive definite for all $\mxN \in \Rbb^{\Ny \times (\Nu + \Nx + \Nx\Nu)}$. Now, we give the lower bound of $\mxG_k(\mxN)$. Recall that
\begin{equation}
\begin{split}
    \Psi &=   \mathbb{E}_{p(\mathbf{x}|\mathbf{y},\hat{\vctheta}_k,\mathbf{u})} \left[ \sum_{t=0}^{\nD - 1} \left(\begin{bmatrix}
        \vcx_t \\ \vcu_t \otimes \vcx_t \\ \vcu_t
    \end{bmatrix}
    \begin{bmatrix}
        \vcx_t \\ \vcu_t \otimes \vcx_t \\ \vcu_t
    \end{bmatrix}^\tr \right)\right]
    \\ &
    =\sum_{t=0}^{\nD-1}
    \begin{bmatrix} 
    \hat{\vcx}_{t|\nD} \hat{\vcx}_{t|\nD}^\tr + \mxP_{t|\nD} 
    & \vcu_t^\tr \otimes \left(\hat{\vcx}_{t|\nD} \hat{\vcx}_{t|\nD}^\tr + \mxP_{t|\nD} \right) 
    & \hat{\vcx}_{t|\nD} \vcu_t^\tr \\ 
    \vcu_t \otimes \left(\hat{\vcx}_{t|\nD} \hat{\vcx}_{t|\nD}^\tr + \mxP_{t|\nD} \right)
    & \vcu_t \vcu_t^\tr \otimes \left(\hat{\vcx}_{t|\nD} \hat{\vcx}_{t|\nD}^\tr + \mxP_{t|\nD} \right)
    & \vcu_t\vcu_t^\tr \otimes \hat{\vcx}_{t|\nD}\\
    \vcu_t \hat{\vcx}_{t|\nD}^\tr 
    & \vcu_t \vcu_t^\tr \otimes \hat{\vcx}_{t|\nD}^\tr
    & \vcu_t \vcu_t^\tr 
    \end{bmatrix}
\end{split}
\end{equation}
as defined in Lemma~\ref{lem:PD_phi_psi}.
Additionally, define matrices $\Pi$ and $\mxY$ as
\begin{equation}
\begin{split}
    \Pi &:= \sum_{t=0}^{\nD-1} \vcy_t
    \begin{bmatrix}
        \vcx_t \\ \vcu_t \otimes \vcx_t \\ \vcu_t
    \end{bmatrix},
\end{split}
\end{equation}
and
\begin{equation}
\begin{split}
    \mxY &:= \sum_{t=0}^{\nD-1} \vcy_t \vcy_t^\tr, 
\end{split}
\end{equation}
respectively. From Lemma~\ref{lem:PD_phi_psi}, we know $\Psi$ is positive definite and invertible.  Using these notations, from \eqref{eq:G_theta|hat_theta PD 2}, we have
\begin{equation}
\begin{split}
    \mxG_k(\mxN) &= \mxN\Psi\mxN^\tr - \mxN \Pi^\tr - \Pi \mxN^\tr + \mxY
    \\ &
    = \left(\mxN\Psi^{\frac
    {1}{2}} - \Pi^\tr\Psi^{-\frac{1}{2}} \right) \left(\mxN\Psi^{\frac{1}{2}} - \Pi^\tr\Psi^{-\frac{1}{2}} \right)^\tr + \mxY - \Pi^\tr \Psi^{-1} \Pi,
\end{split}
\end{equation}
which implies that 
\begin{equation}
    \mxG_k(\mxN) \succeq \mxY - \Pi^\tr \Psi^{-1} \Pi.
\end{equation}
Moreover, we have
\begin{equation}
    \mxY - \Pi^\tr \Psi^{-1} \Pi = \mxG_k(\Pi^\tr \Psi^{-1}) \succ 0,
\end{equation}
which implies that $\mxG_k(\mxN) \succeq \lambda_0 \mathbf{I}$, where $\lambda_0 = \lambda_{\min} \left(\mxY - \Pi^\tr \Psi^{-1} \Pi \right) > 0$. Thus, $\mxG_k(\mxN)$ is positive definite with a lower bound which does not depend on $\mxN$.

Finally, for the $\mxF_k(\mu_{\vcx_0})$, we have
\begin{align}
   \mxF_k(\mu_{\vcx_0}) 
   &= 
   \mathbb{E}_{p(\mathbf{x}|\mathbf{y},\hat{\vctheta}_k,\mathbf{u})} \left[(\vcx_0-\vcmu_{\vcx_0}) (\vcx_0-\vcmu_{\vcx_0})^\tr\right]
   \\
   &= 
   \mathbb{E}_{p(\mathbf{x}|\mathbf{y},\hat{\vctheta}_k,\mathbf{u})} 
   \left[\vcx_0\vcx_0^\tr\right] 
   -
   \mathbb{E}_{p(\mathbf{x}|\mathbf{y},\hat{\vctheta}_k,\mathbf{u})} 
   \left[\vcmu_{\vcx_0}\vcx_0^\tr\right]
   -
   \mathbb{E}_{p(\mathbf{x}|\mathbf{y},\hat{\vctheta}_k,\mathbf{u})} 
   \left[\vcx_0\vcmu_{\vcx_0}^\tr\right]
   +
   \mathbb{E}_{p(\mathbf{x}|\mathbf{y},\hat{\vctheta}_k,\mathbf{u})} 
   \left[\vcmu_{\vcx_0}\vcmu_{\vcx_0}^\tr\right] 
    \\
   &= 
   \mathbb{E}_{p(\mathbf{x}|\mathbf{y},\hat{\vctheta}_k,\mathbf{u})} 
   \left[\vcx_0\vcx_0^\tr\right] 
   -
   \vcmu_{\vcx_0}\left(\mathbb{E}_{p(\mathbf{x}|\mathbf{y},\hat{\vctheta}_k,\mathbf{u})} 
   [\vcx_0]\right)^\tr
   -
   \mathbb{E}_{p(\mathbf{x}|\mathbf{y},\hat{\vctheta}_k,\mathbf{u})} 
    [\vcx_0]\vcmu_{\vcx_0}^\tr
   +
   \vcmu_{\vcx_0}\vcmu_{\vcx_0}^\tr 
    \\
   &= \hat{\vcx}_{0|\nD}\hat{\vcx}_{0|\nD}^\tr +
    \mxP_{0|\nD} - \vcmu_{\vcx_0}\hat{\vcx}_{0|\nD}^\tr - \hat{\vcx}_{0|\nD}\vcmu_{\vcx_0}^\tr + \vcmu_{\vcx_0}\vcmu_{\vcx_0}^\tr
    \\
   &=
   \mxP_{0|\nD} + (\hat{\vcx}_{0|\nD} - \vcmu_{\vcx_0})(\hat{\vcx}_{0|\nD} - \vcmu_{\vcx_0})^\tr 
   \\
   &\succeq \mxP_{0|\nD}.
\end{align}
Accordingly, $\mxP_{0|\nD} \succ 0$, it follows that $\mxF_k(\mu_{\vcx_0})$ is positive definite for all $\mu_{\vcx_0} \in \Rbb^{\Nx}$, with a lower bound which does not depend on $\vcmu_{\vcx_0}$. This concludes the proof.
\qed

\section{Proof of Proposition 3}\label{app:Stationary Point Uniqueness}
From \eqref{eq:J function}, we have
\begin{align}
        J_k(\vctheta) = \mathbb{E}_{p(\mathbf{x}|\mathbf{y},\hat{\vctheta}_k,\mathbf{u})}\Bigg[
        &\frac{\nD}{2}\logdet\big(\mxS_\vcv\big) 
        + 
        \frac{1}{2} \logdet\big(\mxS_{\vcx_0}\big)
        +
        \frac{\nD}{2}\logdet\big(\mxS_\vcw\big) 
        \\&
        +\frac{1}{2}\sum_{t=0}^{\nD-1} \trace\left(\mxS_\vcv^{-1} \left(\vcy_t - \mxN \begin{bmatrix} \vcx_t\\ \vcu_t \otimes \vcx_t \\ \vcu_t\end{bmatrix}  \right)
        \left(\vcy_t - \mxN \begin{bmatrix} \vcx_t\\ \vcu_t \otimes \vcx_t \\ \vcu_t\end{bmatrix}  \right)^\tr\right)
        \\ & 
        + \frac{1}{2}\trace\left(\mxS_{\vcx_0}^{-1} \left(\vcx_0-\vcmu_{\vcx_0}\right) \left(\vcx_0-\vcmu_{\vcx_0}\right)^\tr\right) 
        \\ &
        +\frac{1}{2}\sum_{t=0}^{\nD-1} \trace \left( \mxS_\vcw^{-1}\left(\vcx_{t+1} - \mxM \begin{bmatrix}\vcx_t\\ \vcu_t \end{bmatrix}\right)\left(\vcx_{t+1} - \mxM \begin{bmatrix}\vcx_t\\ \vcu_t \end{bmatrix}\right)^\tr\right)\Bigg].
\end{align}
To find all local minimums, let the partial derivative with respect to each parameter be zero. By deriving the derivatives, we obtain
\begin{align}
        \pdv{J_k(\vctheta)}{\mxM} =& \, \mxS_\vcw^{-1} \mathbb{E}_{p(\mathbf{x}|\mathbf{y},\hat{\vctheta}_k,\mathbf{u})} \left[\sum_{t=0}^{\nD-1} - \left(\vcx_{t+1} - \mxM \begin{bmatrix}\vcx_t\\ \vcu_t \end{bmatrix} \right)\begin{bmatrix}\vcx_t^\tr & \vcu_t^\tr \end{bmatrix}\right], \\
        \pdv{J_k(\vctheta)}{\mxN} =& \, \mxS_\vcv^{-1} \mathbb{E}_{p(\mathbf{x}|\mathbf{y},\hat{\vctheta}_k,\mathbf{u})} \left[-\sum_{t=0}^{\nD-1}  \vcy_t  \begin{bmatrix} \vcx_t\\ \vcu_t \otimes \vcx_t \\ \vcu_t\end{bmatrix}^\tr + \mxN \left(\sum_{t=0}^{\nD-1} ( \begin{bmatrix} \vcx_t\\ \vcu_t \otimes \vcx_t \\ \vcu_t\end{bmatrix} \begin{bmatrix} \vcx_t\\ \vcu_t \otimes \vcx_t \\ \vcu_t\end{bmatrix}^\tr\right)\right],\\
        \pdv{J_k(\vctheta)}{\mxS_\vcw} =& \, \frac{1}{2}\mathbb{E}_{p(\mathbf{x}|\mathbf{y},\hat{\vctheta}_k,\mathbf{u})} \sum_{t=0}^{\nD-1} \left[\nD\mxS_\vcw^{-1} - \mxS_\vcw^{-1} \left(\vcx_{t+1} - \mxM \begin{bmatrix}\vcx_t\\ \vcu_t \end{bmatrix}\right)\left(\vcx_{t+1} - \mxM \begin{bmatrix}\vcx_t\\ \vcu_t \end{bmatrix}\right)^\tr \mxS_\vcw^{-1}\right], \\ 
        \pdv{J_k(\vctheta)}{\mxS_{\vcv}} =& \, \frac{1}{2}\mathbb{E}_{p(\mathbf{x}|\mathbf{y},\hat{\vctheta}_k,\mathbf{u})} \sum_{t=0}^{\nD-1} \left[\nD\mxS_\vcv^{-1} - \mxS_\vcv^{-1} \left(\vcy_t  - \mxN \begin{bmatrix} \vcx_t\\ \vcu_t \otimes \vcx_t \\ \vcu_t\end{bmatrix}\right)\left(\vcy_t - \mxN \begin{bmatrix} \vcx_t\\ \vcu_t \otimes \vcx_t \\ \vcu_t\end{bmatrix}\right)^\tr \mxS_\vcv^{-1}\right], \\ 
        \pdv{J_k(\vctheta)}{\vcmu_{\vcx_0}} =& \, \mxS_{\vcx_0}^{-1} \mathbb{E}_{p(\mathbf{x}|\mathbf{y},\hat{\vctheta}_k,\mathbf{u})} \left[-\vcx_0+\vcmu_{\vcx_0}\right],\\
        \pdv{J_k(\vctheta)}{\mxS_{\vcx_0}} =& \, \frac{1}{2}\mathbb{E}_{p(\mathbf{x}|\mathbf{y},\hat{\vctheta}_k,\mathbf{u})}  \left[\mxS_{\vcx_0}^{-1} - \mxS_{\vcx_0}^{-1} (\vcx_0-\vcmu_{\vcx_0}) (\vcx_0-\vcmu_{\vcx_0})^\tr\mxS_{\vcx_0}^{-1}\right]. 
\end{align}
Using the expectations in \eqref{eq:expectation 1}-\eqref{eq:expectation 3} again, the partial derivatives can be further expanded as
\begin{align}
\label{eq:J derivative}
        \pdv{J_k(\vctheta)}{\mxM} =& \, \mxS_\vcw^{-1}  \left[ \sum_{t=0}^{\nD-1} (-\begin{bmatrix} \hat{\vcx}_{t+1|{\nD}}\hat{\vcx}_{t|{\nD}}^\tr + \mxP_{t+1|{\nD}}\mxL_{t,k}^\tr & \quad \hat{\vcx}_{t+1|\nD}\vcu_t^\tr  \end{bmatrix} + \mxM \begin{bmatrix}\hat{\vcx}_{t|\nD}\hat{\vcx}_{t|\nD}^\tr + \mxP_{t|\nD}  &\hat{\vcx}_{t|\nD}\vcu_t^\tr \\\vcu_t\hat{\vcx}_{t|\nD}^\tr &\vcu_t \vcu_t^\tr \end{bmatrix} )\right], \\
        \pdv{J_k(\vctheta)}{\mxN} =& \, \mxS_\vcv^{-1} \left[\sum_{t=0}^{\nD-1} -\vcy_t \begin{bmatrix} \hat{\vcx}_{t|\nD}\\ \vcu_t \otimes \hat{\vcx}_{t|\nD} \\ \vcu_t\end{bmatrix}^\tr\right. \!\!\!\! 
        \nonumber
        \\& + \mxN \left.\left( \sum_{t=0}^{\nD-1}
        \begin{bmatrix} 
        \hat{\vcx}_{t|\nD} \hat{\vcx}_{t|\nD}^\tr + \mxP_{t|\nD} 
        & \vcu_t^\tr \otimes \left(\hat{\vcx}_{t|\nD} \hat{\vcx}_{t|\nD}^\tr + \mxP_{t|\nD} \right) 
        & \hat{\vcx}_{t|\nD} \vcu_t^\tr \\ 
        \vcu_t \otimes \left(\hat{\vcx}_{t|\nD} \hat{\vcx}_{t|\nD}^\tr + \mxP_{t|\nD} \right)
        & \vcu_t \vcu_t^\tr \otimes \left(\hat{\vcx}_{t|\nD} \hat{\vcx}_{t|\nD}^\tr + \mxP_{t|\nD} \right)
        & \vcu_t\vcu_t^\tr \otimes \hat{\vcx}_{t|\nD}\\
        \vcu_t \hat{\vcx}_{t|\nD}^\tr 
        & \vcu_t \vcu_t^\tr \otimes \hat{\vcx}_{t|\nD}^\tr
        & \vcu_t \vcu_t^\tr 
        \end{bmatrix}
        \right)\right],\\
        \pdv{J_k(\vctheta)}{\mxS_\vcw} =& \, \frac{1}{2} \sum_{t=0}^{\nD-1} \Big[ \nD\mxS_\vcw^{-1} - \mxS_\vcw^{-1} \Big((\hat{\vcx}_{t+1|\nD}\hat{\vcx}_{t+1|\nD}^\tr + \mxP_{t+1|\nD}) -  \begin{bmatrix}\hat{\vcx}_{t+1|{\nD}}\hat{\vcx}_{t|{\nD}}^\tr + \mxP_{t+1|{\nD}}\mxL_{t,k}^\tr &\quad \hat{\vcx}_{t+1|\nD} \vcu_t^\tr  \end{bmatrix} \mxM^\tr  \nonumber \\
        &-  \mxM\begin{bmatrix}\hat{\vcx}_{t+1|\nD}\hat{\vcx}_{t|\nD}^\tr + \mxP_{t+1|\nD}\mxL_{t,k}^\tr &\quad \hat{\vcx}_{t+1|\nD} \vcu_t^\tr  \end{bmatrix}^\tr + \mxM \begin{bmatrix}\hat{\vcx}_{t|\nD}\hat{\vcx}_{t|\nD}^\tr \!+\! \mxP_{t|\nD} & \hat{\vcx}_{t|\nD}\vcu_t^\tr\\ \vcu_t\hat{\vcx}_{t|\nD}^\tr &\vcu_t \vcu_t^\tr\end{bmatrix} \mxM^\tr \Big) \mxS_\vcw^{-1} \Big],\\
        \pdv{J_k(\vctheta)}{\mxS_\vcv} =& \, \frac{1}{2} \sum_{t=0}^{\nD-1} \Bigg[ \nD\mxS_\vcv^{-1} - \mxS_\vcv^{-1} \Bigg( \vcy_t \vcy_t^\tr - \mxN \begin{bmatrix} \hat{\vcx}_{t|\nD}\\ \vcu_t \otimes \hat{\vcx}_{t|\nD} \\ \vcu_t\end{bmatrix} \vcy_t^\tr -  \vcy_t \begin{bmatrix} \hat{\vcx}_{t|\nD}\\ \vcu_t \otimes \hat{\vcx}_{t|\nD} \\ \vcu_t\end{bmatrix}^\tr\mxN^\tr \nonumber \\
        & + \mxN \sum_{t=0}^{\nD-1}
        \begin{bmatrix} 
        \hat{\vcx}_{t|\nD} \hat{\vcx}_{t|\nD}^\tr + \mxP_{t|\nD} 
        & \vcu_t^\tr \otimes \left(\hat{\vcx}_{t|\nD} \hat{\vcx}_{t|\nD}^\tr + \mxP_{t|\nD} \right) 
        & \hat{\vcx}_{t|\nD} \vcu_t^\tr \\ 
        \vcu_t \otimes \left(\hat{\vcx}_{t|\nD} \hat{\vcx}_{t|\nD}^\tr +    \mxP_{t|\nD} \right)
        & \vcu_t \vcu_t^\tr \otimes \left(\hat{\vcx}_{t|\nD} \hat{\vcx}_{t|\nD}^\tr     + \mxP_{t|\nD} \right)
        & \vcu_t\vcu_t^\tr \otimes \hat{\vcx}_{t|\nD}\\
        \vcu_t \hat{\vcx}_{t|\nD}^\tr 
        & \vcu_t \vcu_t^\tr \otimes \hat{\vcx}_{t|\nD}^\tr
        & \vcu_t \vcu_t^\tr 
        \end{bmatrix}
        \mxN^\tr \Bigg) \mxS_\vcv^{-1} \Bigg], \\
        \pdv{J_k(\vctheta)}{\vcmu_{\vcx_0}} =& \, \mxS_{\vcx_0}^{-1} [-\hat{\vcx}_{0|\nD}+\vcmu_{\vcx_0}],\\
        \pdv{J_k(\vctheta)}{\mxS_{\vcx_0}} =& \, \frac{1}{2}  [\mxS_{\vcx_0}^{-1} - \mxS_{\vcx_0}^{-1} (\hat{\vcx}_{0|\nD}\hat{\vcx}_{0|\nD}^\tr + \mxP_{0|\nD} - \hat{\vcx}_{0|\nD}\vcmu_{\vcx_0}^\tr - \vcmu_{\vcx_0}\hat{\vcx}_{0|\nD}^\tr + \vcmu_{\vcx_0}\vcmu_{\vcx_0}^\tr) \mxS_{\vcx_0}^{-1}]. 
\end{align}
From Lemma~\ref{lem:PD_phi_psi}, we know that
\begin{equation}
    \Phi = \sum_{t=0}^{\nD-1}\begin{bmatrix}\hat{\vcx}_{t|\nD}\hat{\vcx}_{t|\nD}^\tr + \mxP_{t|\nD}  &\hat{\vcx}_{t|\nD}\vcu_t^\tr \\\vcu_t\hat{\vcx}_{t|\nD}^\tr &\vcu_t \vcu_t^\tr \end{bmatrix},
\end{equation}
and
\begin{equation}
    \Psi =\sum_{t=0}^{\nD-1}\begin{bmatrix} 
        \hat{\vcx}_{t|\nD} \hat{\vcx}_{t|\nD}^\tr + \mxP_{t|\nD} 
        & \vcu_t^\tr \otimes \left(\hat{\vcx}_{t|\nD} \hat{\vcx}_{t|\nD}^\tr + \mxP_{t|\nD} \right) 
        & \hat{\vcx}_{t|\nD} \vcu_t^\tr \\ 
        \vcu_t \otimes \left(\hat{\vcx}_{t|\nD} \hat{\vcx}_{t|\nD}^\tr +    \mxP_{t|\nD} \right)
        & \vcu_t \vcu_t^\tr \otimes \left(\hat{\vcx}_{t|\nD} \hat{\vcx}_{t|\nD}^\tr     + \mxP_{t|\nD} \right)
        & \vcu_t\vcu_t^\tr \otimes \hat{\vcx}_{t|\nD}\\
        \vcu_t \hat{\vcx}_{t|\nD}^\tr 
        & \vcu_t \vcu_t^\tr \otimes \hat{\vcx}_{t|\nD}^\tr
        & \vcu_t \vcu_t^\tr 
        \end{bmatrix},
\end{equation}
are positive definite and thus invertible. Therefore, we can set the derived partial derivatives to zero and subsequently obtain \eqref{eq: closed form M}-\eqref{eq: closed form Sx0} by solving the resulting system of equations, which concludes the proof.
\qed
\section{Proof of Corollary 1}\label{app:Proof of Corollary 1}
In Proposition~\ref{pro:local minimum}, we have verified that $\nabla_\vctheta \,  J_k = 0$ has a unique solution. Thus, to conclude the invexity of $J_k$, we need to show that $J_k$ is radially unbounded, i.e., $\lim_{\|\vctheta\|\to\infty} \, J_k(\vctheta) = \infty$ and $\lim_{\vctheta\to\partial\Theta} \, J_k(\vctheta) = \infty$, where $\partial\Theta$ denotes the boundary of $\Theta$.
%
%
One can easily see that $J_k$ is a strongly convex quadratic function with respect to tuple $(\mxM,\mxN,\mu_{\vcx_0})$, 
which implies the mentioned property with respect to those variables. Therefore, we only need to consider $\mxS_\vcv, \mxS_\vcw$ and $\mxS_{\vcx_0}$. We prove the claim for $\mxS_\vcv$ considering that the other two cases can be shown similarly, i.e., we show that
\begin{equation}
\begin{split}
    V_k := \mathbb{E}_{p(\mathbf{x}|\mathbf{y},\hat{\vctheta}_k,\mathbf{u})} \left[\frac{\nD}{2}\logdet\mxS_\vcv  +\frac{1}{2}\sum_{t=0}^{\nD-1} \left( \left(\vcy_t - \mxN \begin{bmatrix} \vcx_t\\ \vcu_t \otimes \vcx_t \\ \vcu_t\end{bmatrix}\right)^\tr \mxS_\vcv^{-1}
     \left( \vcy_t -  \mxN \begin{bmatrix} \vcx_t\\ \vcu_t \otimes \vcx_t \\ \vcu_t\end{bmatrix}\right) \right)\right]
\end{split}
\end{equation}
is radially unbounded with respect to $\mxS_\vcv$. Since $\mxS_\vcv$ is positive definite, let $\mxS_\vcv := \mxV^\tr \Lambda \mxV$, where $\mxV^\tr \mxV = \mathbf{I}_{\Ny}$. Furthermore, define $\vca_t :=  \mxV \left(\vcy_t - \mxN \begin{bmatrix} \vcx_t\\ \vcu_t \otimes \vcx_t \\ \vcu_t\end{bmatrix} \right)$. Under these notations, we can rewrite $V_k$ as
\begin{align}
    V_k &= \frac{\nD}{2} \log\det (\mxV^\tr \Lambda \mxV) + \frac{1}{2} \sum_{t=0}^{\nD-1} \mathbb{E}_{p(\mathbf{x}|\mathbf{y},\hat{\vctheta}_k,\mathbf{u})} \big[ \vca_t^\tr \Lambda^{-1} \vca_t \big]\\
    &= \frac{\nD}{2} \log\det \Lambda + \frac{1}{2} \sum_{t=0}^{\nD-1} \mathbb{E}_{p(\mathbf{x}|\mathbf{y},\hat{\vctheta}_k,\mathbf{u})} \big[ \vca_t^\tr \Lambda^{-1} \vca_t \big]\\
    &= \frac{\nD}{2} \log\det \Lambda + \frac{1}{2} \trace \left( \Lambda^{-1} \sum_{t=0}^{\nD-1} \mathbb{E}_{p(\mathbf{x}|\mathbf{y},\hat{\vctheta}_k,\mathbf{u})} \big[\vca_t \vca_t^\tr \big] \right)\\
    &= \frac{\nD}{2} \log\det \Lambda + \frac{1}{2} \trace \left( \Lambda^{-1} \mxV\mxG_k(\mxN)\mxV^\tr \right)
    \\&
    \geq \frac{\nD}{2} \log\det \Lambda + \frac{\lambda_{\min} \left(\mxY - \Pi^\tr \Psi^{-1} \Pi \right)}{2} \trace \left( \Lambda^{-1} \right).
\end{align}
The last inequality holds since $\mxG_k(\mxN) \succeq \lambda_{\min} \left(\mxY - \Pi^\tr \Psi^{-1} \Pi \right)\mathbf{I}$ and $\mxV$ is full rank. Finally, Lemma~\ref{lem:invexity_logdet_trinv} implies that $V_k$ is radially unbounded, which concludes the proof. 

\qed

\section{Proof of Proposition 4}\label{app:recursive feaisbility}
We need to show that $\hat{\mxS}_{\vcw,k+1}$, $\hat{\mxS}_{\vcv,k+1}$ and $\hat{\mxS}_{\vcx_0,k+1}$ are positive definite matrices, which subsequently implies that \eqref{eq: closed form M}-\eqref{eq: closed form Sx0} provide a feasible solution for \eqref{eq:EM optimization J_k}. 
For $\hat{\mxS}_{\vcw,k+1}$, by substituting \eqref{eq: closed form C} into \eqref{eq:G_theta|hat_theta PD 2} and comparing with \eqref{eq: closed form Sv}, we have
    \begin{equation}
       \hat{\mxS}_{\vcv,k+1} = \frac{1}{\nD}\mxG_k(\hat{\mxN}_{k+1}).
    \end{equation}
From Proposition~\ref{pro: postive definite}, we know that $\mxG_k(\mxN)$ is positive definite for any $\mxN \in \Rbb^{\Ny \times (\Nu + \Nx + \Nx\Nu)}$, which implies that $\hat{\mxS}_{\vcv,k+1}$ is positive definite.
One can see that similar line of argument holds for $\hat{\mxS}_{\vcv,k+1}$ and $\hat{\mxS}_{\vcx_0,k+1}$, 
which concludes the proof.
\qed

\bibliographystyle{IEEEtran}
\bibliography{mainbib_bilinear}



\end{document}